\documentclass[12pt]{article}

\usepackage[colorlinks,citecolor=blue,urlcolor=blue,breaklinks]{hyperref}
\usepackage{latexsym,amsmath,amssymb,amsfonts,amsthm,bbm,dsfont,enumitem,xcolor,soul,graphicx}
\usepackage{natbib}
\usepackage[ruled, vlined]{algorithm2e}
\usepackage{multirow}

\newtheorem{thm}{Theorem}
\newtheorem{prop}[thm]{Proposition}
\newtheorem{lemma}[thm]{Lemma}
\newtheorem{cor}[thm]{Corollary}

\DeclareMathOperator*{\argmin}{argmin}
\DeclareMathOperator*{\argmax}{argmax}
\DeclareMathOperator*{\soft}{\textbf{\textup{soft}}}

\title{High-dimensional changepoint estimation via sparse projection}
\author{Tengyao Wang\footnote{Research supported by a Benefactors' Scholarship from St John's College, Cambridge.} \, and Richard J. Samworth\footnote{Research supported by an EPSRC Fellowship, an EPSRC Programme Grant and a Philip Leverhulme Prize.}\\University of Cambridge}
\date{(\today)}

\begin{document}
\maketitle

\begin{abstract}
Changepoints are a very common feature of Big Data that arrive in the form of a data stream.  In this paper, we study high-dimensional time series in which, at certain time points, the mean structure changes in a sparse subset of the coordinates.  The challenge is to borrow strength across the coordinates in order to detect smaller changes than could be observed in any individual component series.  We propose a two-stage procedure called \texttt{inspect} for estimation of the changepoints: first, we argue that a good projection direction can be obtained as the leading left singular vector of the matrix that solves a convex optimisation problem derived from the CUSUM transformation of the time series.  We then apply an existing univariate changepoint estimation algorithm to the projected series.  Our theory provides strong guarantees on both the number of estimated changepoints and the rates of convergence of their locations, and our numerical studies validate its highly competitive empirical performance for a wide range of data generating mechanisms.  Software implementing the methodology is available in the \textbf{R} package \textbf{InspectChangepoint}.
\end{abstract}

\section{Introduction}

One of the most commonly-encountered issues with Big Data is heterogeneity.  When collecting vast quantities of data, it is usually unrealistic to expect that stylised, traditional statistical models of independent and identically distributed observations can adequately capture the complexity of the underlying data generating mechanism.  Departures from such models may take many forms, including missing data, correlated errors and data combined from multiple sources, to mention just a few.  

When data are collected over time, heterogeneity often manifests itself through non-stationarity, where the data generating mechanism varies with time.  Perhaps the simplest form of non-stationarity assumes that population changes occur at a relatively small number of discrete time points.  If correctly estimated, these `changepoints' can be used to partition the original data set into shorter segments, which can then be analysed using methods designed for stationary time series.  Moreover, the locations of these changepoints are often themselves of significant practical interest.

In this paper, we study high-dimensional time series that may have changepoints; moreover, we consider in particular settings where at a changepoint, the mean structure changes in a sparse subset of the coordinates.  Despite their simplicity, such models are of great interest in a wide variety of applications.  For instance, in the case of stock price data, it may well be the case that stocks in related industry sectors experience virtually simultaneous `shocks' \citep{ChenGupta1997}.  In internet security monitoring, a sudden change in traffic at multiple routers may be an indication of a distributed denial of service attack \citep{PengLeckieRamamohanarao2004}.  In functional Magnetic Resonance Imaging (fMRI) studies, a rapid change in blood oxygen level dependent (BOLD) contrast in a subset of voxels may suggest neurological activity of interest \citep{AstonKirch2012}.

Our main contribution is to propose a new method for estimating the number and locations of the changepoints in such high-dimensional time series, a challenging task in the absence of knowledge of the coordinates that undergo a change.  In brief, we first seek a good projection direction, which should ideally be closely aligned with the vector of mean changes.  We can then apply an existing univariate changepoint estimation algorithm to the projected series.  For this reason, we call our algorithm \texttt{inspect}, short for \underline{in}formative \underline{s}parse \underline{p}rojection for \underline{e}stimation of \underline{c}hangepoin\underline{t}s; it is implemented in the \textbf{R} package \textbf{InspectChangepoint} \citep{WangSamworth2016a}.

In more detail, in the single changepoint case, our first observation is that at the population level, the vector of mean changes is the leading left singular vector of the matrix obtained as the cumulative sum (CUSUM) transformation of the mean matrix of the time series.  This motivates us to begin by applying the CUSUM transformation to the time series.  Unfortunately, computing the $k$-sparse leading left singular vector of a matrix is a combinatorial optimisation problem, but nevertheless, we are able to formulate an appropriate convex relaxation of the problem, from which we derive our projection direction.  At the second stage of our algorithm, we compute the vector of CUSUM statistics for the projected series, identifying a changepoint if the maximum absolute value of this vector is sufficiently large.  For the case of multiple changepoints, we combine our single changepoint algorithm with the method of Wild Binary Segmentation  \citep{Fryzlewicz2014} to identify changepoints recursively.

A brief illustration of the \texttt{inspect} algorithm in action is given in Figure~\ref{Fig:Example}.  Here, we simulated a $2000 \times 1000$ data matrix having independent normal columns with identity covariance and with three changepoints in the mean structure at locations $500, 1000$ and $1500$.  Changes occur in 40 coordinates, where consecutive changepoints overlap in half of their coordinates, and the squared $\ell_2$ norms of the vectors of mean changes were $0.4$, $0.9$ and $1.6$ respectively.  The top-left panel shows the original data matrix and the top-right shows its CUSUM transformation, while the bottom-left panel shows overlays for the three detected changepoints of the univariate CUSUM statistics after projection.  Finally, the bottom-right panel displays the largest absolute values of the projected CUSUM statistics obtained by running the wild binary segmentation algorithm to completion (in practice, we would apply a termination criterion instead, but this is still helpful for illustrative purposes).  We see that the three detected changepoints are very close to their true locations, and it is only for these three locations that we obtain a sufficiently large CUSUM statistic to declare a changepoint.  We emphasise that our focus here is on the so-called \emph{offline} version of the changepoint estimation problem, where we observe the whole data set before seeking to locate changepoints.  The corresponding online problem, where one aims to declare a changepoint as soon as possible after it has occurred, is also of great interest \citep{TNB2014}, but is beyond the scope of the current work. 
\begin{figure}[htbp]
\begin{center}
\begin{tabular}{rr}
\includegraphics[width=0.4\textwidth]{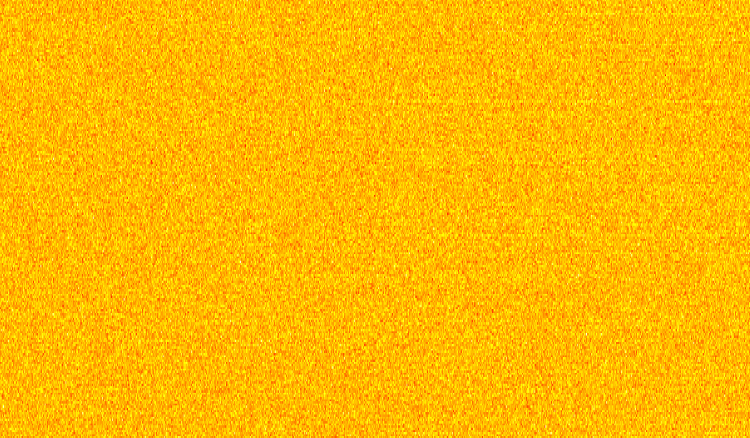} 
&
\includegraphics[width=0.4\textwidth]{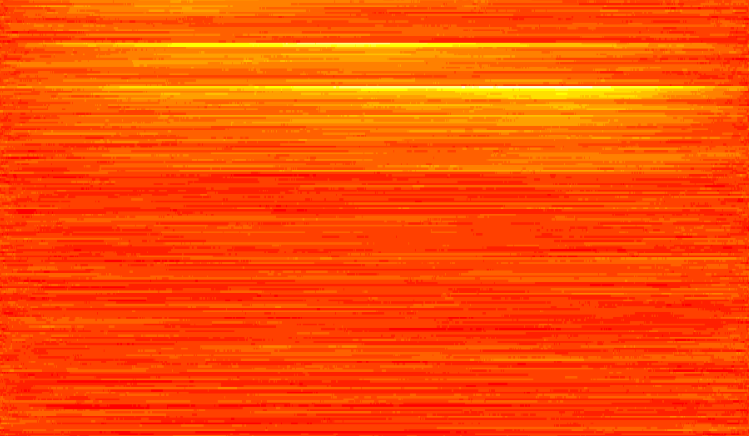}\\
\includegraphics[width=0.45\textwidth]{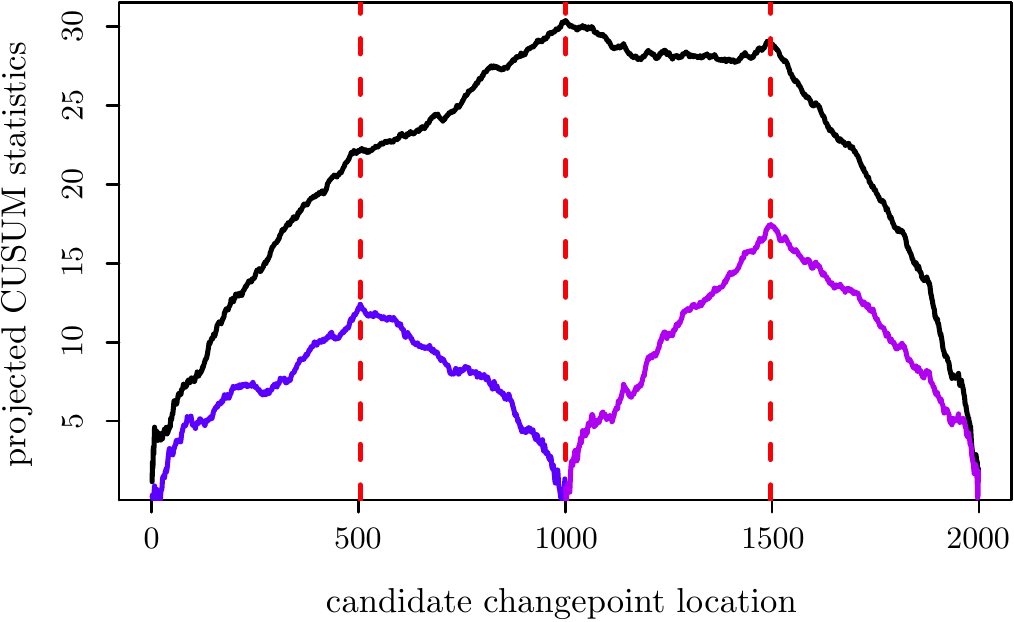} 
&
\includegraphics[width=0.45\textwidth]{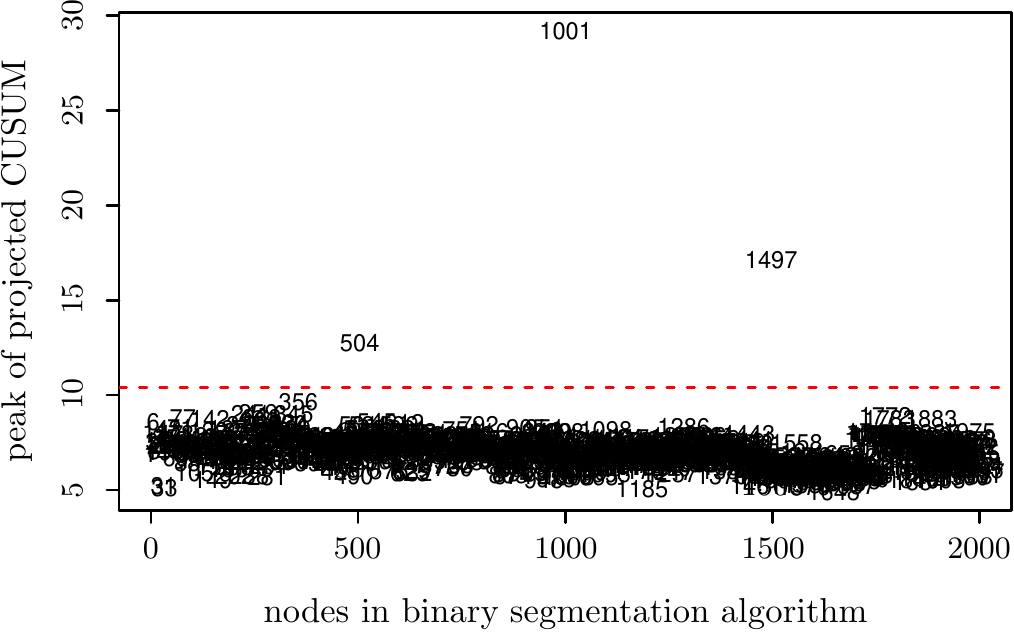}
\end{tabular}
\caption{Example of \texttt{inspect} algorithm in action. Top-left: visualisation of the data matrix. Top-right: its CUSUM transformation. Bottom-left: overlay of the projected CUSUM statistics for the three changepoints detected. Bottom-right: visualisation of thresholding; the three detected changepoints are above the threshold (dotted red line) whereas the remaining numbers are the test statistics obtained if we run the wild binary segmentation to completion without applying a termination criterion.}
\label{Fig:Example}
\end{center}
\end{figure}

Our theoretical development proceeds first by controlling the angle between the estimated projection direction and the optimal direction, which is given by the normalised vector of mean changes.  Under appropriate conditions, this enables us to provide finite-sample bounds which guarantee that with high probability we both recover the correct number of changepoints, and estimate their locations to within a specified accuracy.  Indeed, in the single changepoint case, the rate of convergence for the changepoint location estimation of our method is within a doubly logarithmic factor of the minimax optimal rate.  Our extensive numerical studies indicate that the algorithm performs extremely well in a wide variety of settings.

The study of changepoint problems dates at least back to \citet{Page1955}, and has since found applications in many different areas, including genetics \citep{Olshenetal2004}, disease outbreak watch \citep{SparksKeighleyMuscatello2010} and aerospace engineering \citep{HenrySimaniPatton2010}, in addition to those already mentioned.  There is a vast and rapidly growing literature on different methods for changepoint detection and localisation, especially in the univariate problem.  Surveys of various methods can be found in \citet{CsorgoHorvath1997} and \citet{HorvathRice2014}. In the case of univariate changepoint estimation, state-of-the-art methods include Pruned Exact Linear Time method (PELT) \citep{KFE2012}, Wild Binary Segmentation (WBS) \citep{Fryzlewicz2014} and Simultaneous Multiscale Changepoint Estimator (SMUCE) \citep{FMS2014}.

Some of the univariate changepoint methodologies have been extended to multivariate settings. Examples include \citet{HorvathKokoszkaSteinebach1999}, \citet{Ombaoetal2005}, \citet{Aueetal2009} and \citet{KirchMushalOmbao2014}. However, there are fewer available tools for high-dimensional changepoint problems, where both the dimension $p$ and the length $n$ of the data stream may be large, and where we may allow a sparsity assumption on the coordinates of change. \citet{Bai2010} investigates the performance of the least squares estimator of a single changepoint in the high-dimensional setting. \citet{Zhangetal2010}, \citet{HorvathHuskova2012} and \citet{EnikeevaHarchaoui2014} consider estimators based on $\ell_2$ aggregations of CUSUM statistics in all coordinates, but without using any sparsity assumptions.  \citet{EnikeevaHarchaoui2014} also consider a scan statistic that takes sparsity into account.  \citet{Jirak2015} considers an $\ell_\infty$ aggregation of the CUSUM statistics that works well for sparse changepoints. \citet{ChoFryzlewicz2015} propose Sparse Binary Segmentation, which also takes sparsity into account and can be viewed as a hard-thresholding of the CUSUM matrix followed by an $\ell_1$ aggregation. \citet{Cho2016} proposes a double-CUSUM algorithm that performs a CUSUM transformation along the location axis on the columwise-sorted CUSUM matrix. In a slightly different setting, \citet{LavielleTeyssiere2006}, \citet{Aueetal2009}, \citet{Bucheretal2014}, \citet{Preussetal2015} and \citet{CribbenYu2015} deal with changes in cross-covariance, while \citet{SohChandrasekaran2017} study a high-dimensional changepoint problem where all mean vectors are sparse. \citet{AstonKirch2014} considered the asymptotic efficiency of detecting a single changepoint in a high-dimensional setting, and the oracle projection-based estimator under cross-sectional dependence structure. 

The outline of the rest of the paper is as follows.  In Section~\ref{Sec:Problem}, we give a formal description of the problem and the class of data generating mechanisms under which our theoretical results hold.  Our methodological development in the single changepoint setting is presented in Section~\ref{Sec:SingleCP}, and includes theoretical guarantees on both the projection direction and location of the estimated changepoint in the simplest case of observations that are independent across both space and time.  Section~\ref{Sec:MultipleCP} extends these ideas to the case of multiple changepoints with the aid of Wild Binary Segmentation, and our numerical studies are given in Section~\ref{Sec:Simulation}.  Section~\ref{Sec:Dep} studies in detail important cases of temporal and spatial dependence.  For temporal dependence, no change to our methodology is required, but new arguments are needed to provide theoretical guarantees; for spatial dependence, we show how to modify our methodology to try to maximise the signal to noise ratio of the projected univariate series, and also provide corresponding theoretical results on the performance of this variant of the basic \texttt{inspect} algorithm.  Proofs of our main results are given in Section~\ref{Appendix:proofs}; additional results and their proofs are given in the online supplementary material \citet{WangSamworth2016}, hereafter referred to simply as the online supplement.

We conclude this section by introducing some notation used throughout the paper. For a vector $u = (u_1,\ldots,u_M)^\top \in\mathbb{R}^M$, a matrix $A = (A_{ij})\in\mathbb{R}^{M\times N}$ and for $q\in[1,\infty)$, we write $\|u\|_q:= \bigl(\sum_{i=1}^M|u_i|^q\bigr)^{1/q}$ and $\|A\|_q:= \bigl(\sum_{i=1}^M\sum_{j=1}^N |A_{ij}|^q\bigr)^{1/q}$ for their (entrywise) $\ell_q$-norms, as well as $\|u\|_\infty := \max_{i=1,\ldots,M} |u_i|$ and $\|A\|_\infty := \max_{i=1,\ldots,M,j=1,\ldots,N} |A_{ij}|$. We write $\|A\|_{*} := \sum_{i=1}^{\min(M,N)}\sigma_i(A)$ and $\|A\|_{\mathrm{op}} := \max_{i}\sigma_i(A)$ respectively for the nuclear norm and operator norm of matrix $A$, where $\sigma_1(A),\ldots,\sigma_{\min(M,N)}(A)$ are its singular values. We also write $\|u\|_0 := \sum_{i=1}^M \mathds{1}_{\{u_i\neq 0\}}$. For $S\subseteq \{1,\ldots,M\}$ and $T\subseteq\{1,\ldots,N\}$, we write $u_S := (u_i: i\in S)^\top$ and write $M_{S,T}$ for the $|S|\times |T|$ submatrix of $A$ obtained by extracting the rows and columns with indices in $S$ and $T$ respectively. For two matrices $A,B\in\mathbb{R}^{M\times N}$, we denote their trace inner product as $\langle A, B\rangle = \mathrm{tr}(A^\top B)$. For two non-zero vectors $u,v\in\mathbb{R}^p$, we write $\angle(u,v) := \cos^{-1}(\frac{|\langle u,v\rangle|}{\|u\|_2\|v\|_2})$ for the acute angle bounded between them. We let $\mathbb{S}^{p-1}:= \{x \in \mathbb{R}^p: \|x\|_2=1\}$ be the unit Euclidean sphere in $\mathbb{R}^p$, and let $\mathbb{S}^{p-1}(k) := \{x\in\mathbb{S}^{p-1} : \|x\|_0\leq k\}$.  Finally, we write $a_n \asymp b_n$ to mean $0 < \liminf_{n \rightarrow \infty} |a_n/b_n| \leq \limsup_{n \rightarrow \infty} |a_n/b_n| < \infty$.

\section{Problem description}
\label{Sec:Problem}

We initially study the following basic model: let $X_1,\ldots,X_n$ be independent $p$-dimensional random vectors sampled from
\begin{equation}
X_t \sim N_p(\mu_t, \sigma^2 I_p),\qquad 1\leq t\leq n,
\label{Eq:Normal}
\end{equation}
and combine the observations into a matrix $X = (X_1,\ldots,X_n)\in\mathbb{R}^{p\times n}$.  Extensions to settings of both temporal and spatial dependence will be studied in detail in Section~\ref{Sec:Dep}.  We assume that the mean vectors follow a piecewise-constant structure with $\nu+1$ segments. In other words, there exists $\nu$ \emph{changepoints} 
\[
1 \leq z_1 < z_2 < \cdots < z_\nu \leq n-1
\]
such that
\begin{equation}
\label{Eq:mus}
\mu_{z_i+1} = \cdots = \mu_{z_{i+1}} =: \mu^{(i)},\qquad \forall\; 0\leq i\leq \nu,
\end{equation}
where we adopt the convention that $z_0 := 0$ and $z_{\nu+1} := n$.  For $i=1,\ldots,\nu$, write 
\begin{equation}
\label{Eq:thetas}
\theta^{(i)} := \mu^{(i)} - \mu^{(i-1)}
\end{equation}
for the (non-zero) difference in means between consecutive stationary segments.  We will later assume that the changes in mean are sparse in the sense that there exists $k\in\{1,\ldots,p\}$ (typically $k$ is much smaller than $p$) such that 
\begin{equation}
\label{Eq:thetasparsity}
\|\theta^{(i)}\|_0 \leq k
\end{equation}
for each $i=1,\ldots,\nu$, since our methodology performs best when aggregating signals spread across an (unknown) sparse subset of coordinates; see also the discussion after Corollary~\ref{Cor:Proj} below.  However, we remark that our methodology does not require the knowledge of the sparsity level and can be applied in non-sparse settings as well. 

Our goal is to estimate the set of changepoints $\{z_1,\ldots,z_{\nu}\}$ in the high-dimensional regime, where $p$ may be comparable to, or even larger than, the length $n$ of the series.  The signal strength of the estimation problem is determined by the magnitude of mean changes $\{\theta^{(i)}: 1\leq i\leq \nu\}$ and the lengths of stationary segments $\{z_{i+1} - z_i : 0\leq i\leq \nu\}$, whereas the noise is related to the variance $\sigma^2$ and the dimensionality $p$ of the observed data points.  For our theoretical results, we will assume that the changepoint locations satisfy 
\begin{equation}
\label{Eq:tau}
n^{-1}\min\{z_{i+1}- z_i : 0\leq i\leq \nu\} \geq \tau,
\end{equation}
and the magnitudes of mean changes are such that 
\begin{equation}
\label{Eq:Vartheta}
\|\theta^{(i)}\|_2 \geq \vartheta,\qquad \forall \; 1\leq i\leq \nu.
\end{equation}
Suppose that an estimation procedure outputs $\hat{\nu}$ changepoints located at $1 \leq \hat{z}_1 < \cdots < \hat{z}_{\hat{\nu}} \leq n-1$. Our finite-sample bounds will imply a rate of convergence for \texttt{inspect} in an asymptotic setting where the problem parameters are allowed to depend on $n$. Suppose that $\mathcal{P}_n$ is a class of distributions of $X \in \mathbb{R}^{p\times n}$ with sample size $n$.  In this context, we follow the convention in the literature \citep[e.g.][]{Venkatraman1992} and say that the procedure is consistent for $\mathcal{P}_n$ with rate of convergence $\rho_n$ if
\begin{equation}
\inf_{P\in\mathcal{P}_n} \mathbb{P}_P\bigl\{\text{$\hat{\nu} = \nu$ and $|\hat{z}_i - z_i|\leq n\rho_n$ for all $1\leq i\leq \nu$}\bigr\} \to 1
\label{Eq:ConvergenceRate}
\end{equation}
as $n\to\infty$. 

\section{Data-driven projection estimator for a single changepoint}
\label{Sec:SingleCP}
We first consider the problem of estimating a single changepoint (i.e.~$\nu = 1$) in a high-dimensional time series dataset $X\in\mathbb{R}^{p\times n}$. For simplicity, write $z:=z_1$, $\theta = (\theta_1,\ldots,\theta_p)^\top := \theta^{(1)}$ and $\tau:= n^{-1}\min\{z,n-z\}$. We seek to aggregate the rows of the data matrix $X$ in an almost optimal way so as to maximise the signal-to-noise ratio, and then locate the changepoint using a one-dimensional procedure. For any $a\in\mathbb{S}^{p-1}$, $a^\top X$ is a one-dimensional time series with 
\[
a^\top X_t \sim N(a^\top \mu_t, \sigma^2).
\]
Hence, the choice $a = \theta/\|\theta\|_2$ maximises the magnitude of the difference in means between the two segments. However, $\theta$ is typically unknown in practice, so we should seek a projection direction that is close to the oracle projection direction $v := \theta/\|\theta\|_2$. Our strategy is to perform sparse singular value decomposition on the CUSUM transformation of $X$. The method and limit theory of CUSUM statistics in the univariate case can be traced back to \citet{DarlingErdos1956}. For $p \in \mathbb{N}$ and $n\geq 2$, we define the CUSUM transformation $\mathcal{T}_{p,n}: \mathbb{R}^{p\times n}\to\mathbb{R}^{p\times (n-1)}$ by
\begin{align}
\label{Eq:CUSUM}
[\mathcal{T}_{p,n}(M)]_{j,t} :&= \sqrt\frac{t(n-t)}{n}\biggl(\frac{1}{n-t}\!\sum_{r=t+1}^n M_{j,r} - \frac{1}{t}\sum_{r=1}^t M_{j,r}\biggr) \nonumber \\
&= \sqrt\frac{n}{t(n-t)}\biggl(\frac{t}{n}\sum_{r=1}^n M_{j,r} - \sum_{r=1}^t M_{j,r}\biggr).
\end{align}
In fact, to simplify the notation, we will write $\mathcal{T}$ for $\mathcal{T}_{p,n}$, since $p$ and $n$ can be inferred from the dimensions of the argument of $\mathcal{T}$.  Note also that $\mathcal{T}$ reduces to computing the vector of classical one-dimensional CUSUM statistics when $p=1$. We write 
\[
X = \boldsymbol{\mu} + W,
\]
where $\boldsymbol{\mu} = (\mu_1,\ldots,\mu_n) \in \mathbb{R}^{p\times n}$ and $W = (W_1,\ldots,W_n)$ is a $p \times n$ random matrix with independent $N_p(0,\sigma^2I_p)$ columns. Let $T := \mathcal{T}(X)$, $A := \mathcal{T}(\boldsymbol{\mu})$ and $E := \mathcal{T}(W)$, so by the linearity of the CUSUM transformation we have the decomposition
\[
T = A + E.
\]
We remark that when $\sigma$ is known, each $|T_{j,t}|$ is the likelihood ratio statistic for testing the null hypothesis that the $j$th row of $\boldsymbol{\mu}$ is constant against the alternative that the $j$th row of $\boldsymbol{\mu}$ undergoes a single change at time $t$.  Moreover, if the direction $v \in \mathbb{S}^{p-1}$ of the potential single change at a given time $t$ were known, then the most powerful test of whether or not $\vartheta = 0$ would be based on $|(v^\top T)_t|$.  In the single changepoint case, the entries of the matrix $A$ can be computed explicitly:
\[
A_{j,t} = \begin{cases} \sqrt\frac{t}{n(n-t)}(n-z)\theta_j, & \text{if $t\leq z$}\\ \sqrt\frac{n-t}{nt} z\theta_j, &\text{if $t>z$}.\end{cases}
\]
Hence we can write
\begin{equation}
\label{Eq:A}
A = \theta\gamma^\top,
\end{equation}
where
\begin{equation}
\gamma := \frac{1}{\sqrt{n}}\biggl(\sqrt\frac{1}{n-1}(n-z), \sqrt\frac{2}{n-2}(n-z),\ldots,\sqrt{z(n-z)},\sqrt\frac{n-z-1}{z+1}z,\ldots,\sqrt\frac{1}{n-1}z\biggr)^\top.
\label{Eq:gamma}
\end{equation}
In particular, this implies that the oracle projection direction is the leading left singular vector of the rank 1 matrix $A$. 
In the ideal case where $k$ is known, we could in principle let $\hat v_{\mathrm{max},k}$ be a $k$-sparse leading left singular vector of $T$, defined by
\begin{equation}
\label{Eq:nonconvex}
\hat v_{\mathrm{max},k} \in \argmax_{\tilde{v}\in\mathbb{S}^{p-1}(k)} \|T^\top \tilde{v}\|_2,
\end{equation}
and it can then be shown using a perturbation argument akin to the Davis--Kahan `$\sin\theta$' theorem (cf.\ \citet{DavisKahan1970, YuWangSamworth2015}) that $\hat v_{\mathrm{max},k}$ is a consistent estimator of the oracle projection direction $v$ under mild conditions (see Proposition~\ref{Prop:Wedin} in the online supplement). However, the optimisation problem in~\eqref{Eq:nonconvex} is non-convex and hard to implement. In fact, computing the $k$-sparse leading left singular vector of a matrix is known to be NP-hard (e.g.\ \citet{TillmannPfetsch2014}). The naive algorithm that scans through all possible $k$-subsets of the rows of $T$ has running time exponential in $k$, which quickly becomes impractical to run for even moderate sizes of $k$. 

A natural approach to remedy this computational issue is to work with a convex relaxation of the optimisation problem~\eqref{Eq:nonconvex} instead.  In fact, we can write
\begin{align}
\max_{u\in\mathbb{S}^{p-1}(k)} \|u^\top T\|_2 & = \max_{u\in\mathbb{S}^{p-1}(k), w\in\mathbb{S}^{n-2}} u^\top T w\nonumber\\
& = \max_{u\in\mathbb{S}^{p-1}, w\in\mathbb{S}^{n-2},\|u\|_0\leq k} \langle uw^\top, T\rangle = \max_{M \in \mathcal{M}} \langle M, T\rangle,
\label{Eq:Relaxation}
\end{align}
where $\mathcal{M} := \{M \in \mathbb{R}^{p \times (n-1)}: \|M\|_* = 1, \mathrm{rank}(M) = 1, \text{$M$ has at most $k$ non-zero rows}\}$.  The final expression in~\eqref{Eq:Relaxation} has a convex (linear) objective function $M\mapsto \langle M, T\rangle$. The requirement $\mathrm{rank}(M) = 1$ in the constraint set $\mathcal{M}$ is equivalent to $\|\sigma(M)\|_0 = 1$, where $\sigma(M) := (\sigma_1(M),\ldots,\sigma_{\min(p,n-1)}(M))^\top$ is the vector of singular values of $M$. This motivates us to absorb the rank constraint into the nuclear norm constraint, which we relax from an equality constraint to an inequality constraint in order to make it convex. Furthermore, we can relax the row sparsity constraint in the definition of $\mathcal{M}$ to an entrywise $\ell_1$-norm penalty. The optimisation problem of finding
\begin{equation}
\label{Eq:convex}
\hat{M} \in \argmax_{M\in\mathcal{S}_1} \bigl\{\langle T, M\rangle - \lambda\|M\|_1\bigr\},
\end{equation}
where $\mathcal{S}_1 := \{M\in\mathbb{R}^{p\times (n-1)}: \|M\|_*\leq 1\}$ and $\lambda > 0$ is a tuning parameter to be chosen later, is therefore a convex relaxation of~\eqref{Eq:nonconvex}.  We remark that a similar convex relaxation has appeared in the different context of sparse principal component estimation \citep{dAspremontetal2007}, where the sparse leading left singular vector is also the optimisation target.  The convex problem~\eqref{Eq:convex} may be solved using the alternating direction method of multipliers algorithm (ADMM, see \citet{GabayMercier1976, Boydetal2011}) as in Algorithm~\ref{Algo:ADMM}. More specifically, the optimisation problem in \eqref{Eq:convex} is equivalent to maximising $\langle T, Y\rangle - \lambda\|Z\|_1 - \mathbb{I}_{\mathcal{S}_1}(Y)$ subject to $Y = Z$, where $\mathbb{I}_{\mathcal{S}_1}$ is the function that is $0$ on $\mathcal{S}_1$ and $\infty$ on $\mathcal{S}_1^{\mathrm{c}}$. Its augmented Lagrangian is given by
\[
L(Y,Z,R) := \langle T, Y\rangle - \mathbb{I}_{\mathcal{S}_1}(Y) - \lambda\|Z\|_1 - \langle R, Y-Z\rangle - \frac{1}{2}\|Y-Z\|_2^2,
\]
with the Lagrange multiplier $R$ being the dual variable. Each iteration of the main loop in Algorithm~\ref{Algo:ADMM} first performs a primal update by maximising $L(Y,Z,R)$ marginally with respect to $Y$ and $Z$, then followed by a dual gradient update of $R$ with constant step size.  The function $\Pi_{\mathcal{S}_1}(\cdot)$ in Algorithm~\ref{Algo:ADMM} denotes projection onto the convex set $\mathcal{S}_1$ with respect to the Frobenius norm distance. If $A = UDV^\top$ is the singular value decomposition of $A \in \mathbb{R}^{p \times (n-1)}$ with $\mathrm{rank}(A) = r$, where $D$ is a diagonal matrix with diagonal entries $d_1,\ldots,d_r$, then $\Pi_{\mathcal{S}_1}(A) = U\tilde{D}V^\top$, where $\tilde{D}$ is a diagonal matrix with entries $\tilde{d}_1,\ldots,\tilde{d}_r$ such that $(\tilde{d}_1,\ldots,\tilde{d}_r)^\top$ is the Euclidean projection of the vector $(d_1,\ldots,d_r)^\top$ onto the standard $(r-1)$-simplex 
\[
\Delta^{r-1}:=\biggl\{(x_1,\ldots,x_r)^\top\in\mathbb{R}^r : \text{$\sum_{\ell=1}^{r}x_\ell = 1$ and $x_\ell\geq 0$ for all $\ell$}\biggr\}.
\]
For an efficient algorithm for such simplicial projection, see \citet{ChenYe2011}. The $\soft$ function in Algorithm~\ref{Algo:ADMM} denotes an entrywise soft-thresholding operator defined by $\bigl(\soft(A,\lambda)\bigr)_{ij} := \mathrm{sgn}(A_{ij})\max\{ |A_{ij}| - \lambda, 0\}$ for any $\lambda\geq 0$ and matrix $A = (A_{ij})$.

\begin{algorithm}[htbp!]
\SetAlgoLined
\IncMargin{1em}
\DontPrintSemicolon
\SetKwRepeat{Do}{repeat}{until}
\KwIn{
$T\in\mathbb{R}^{p\times (n-1)}$, $\lambda > 0$.
}
\vskip 0.5ex
\textbf{Set}: {$Y = Z = R = \mathbf{0}\in\mathbb{R}^{p\times (n-1)}$}

\Do{$Y-Z$ converges to 0}{
$Y\leftarrow \Pi_{\mathcal{S}_1}(Z - R + T)$\\
$Z\leftarrow \soft(Y+R, \lambda)$\\
$R \leftarrow R + (Y-Z)$
}
$\hat{M} \leftarrow Y$
\vskip 0.5ex
\KwOut{$\hat{M}$}
\vskip 1ex
\caption{Pseudo-code for an ADMM algorithm that computes the solution to the optimisation problem~\eqref{Eq:convex}.}
\label{Algo:ADMM}
\end{algorithm}

We remark that one may be interested to further relax~\eqref{Eq:convex} by replacing $\mathcal{S}_1$ with the larger set $\mathcal{S}_2:= \{M\in\mathbb{R}^{p\times (n-1)}: \|M\|_2\leq 1\}$ defined by the entrywise $\ell_2$-unit ball. We see from Proposition~\ref{Lemma:Duality} in the online supplement that the smoothness of $\mathcal{S}_2$ results in a simple dual formulation, which implies that
\begin{equation}
\label{Eq:SoftThresholding}
\tilde{M} := \frac{\soft(T,\lambda)}{\|\soft(T,\lambda)\|_2} = \argmax_{M\in\mathcal{S}_2}\bigl\{\langle T,M\rangle -\lambda\|M\|_1\bigr\} 
\end{equation}
is the unique optimiser of the primal problem. The soft-thresholding operation is significantly faster than the ADMM algorithm in Algorithm~\ref{Algo:ADMM}. Hence by enlarging $\mathcal{S}_1$ to $\mathcal{S}_2$, we can significantly speed up the running time of the algorithm in exchange for some loss in statistical efficiency caused by the further relaxation of the constraint set.  See Section~\ref{Sec:Simulation} for further discussion.

Let $\hat{v}$ be the leading left singular vector of 
\begin{equation}
\label{Eq:convexS}
\hat{M} \in \argmax_{M\in\mathcal{S}} \bigl\{\langle T, M\rangle - \lambda\|M\|_1\bigr\},
\end{equation}
for either $\mathcal{S} = \mathcal{S}_1$ or $\mathcal{S} = \mathcal{S}_2$.  In order to describe the theoretical properties of $\hat{v}$ as an estimator of the oracle projection direction $v$, we introduce the following class of distributions: let $\mathcal{P}(n,p,k,\nu, \vartheta, \tau, \sigma^2)$ denote the class of distributions of $X = (X_1,\ldots,X_n)\in \mathbb{R}^{p\times n}$ with independent columns drawn from~\eqref{Eq:Normal}, where the changepoint locations satisfy~\eqref{Eq:tau} and the vectors of mean changes are such that~\eqref{Eq:thetasparsity} and~\eqref{Eq:Vartheta} hold.  Although this notation accommodates the multiple changepoint setting studied in Section~\ref{Sec:MultipleCP} below, we emphasise that our focus here is on the single changepoint setting.  The error bound in Proposition~\ref{Prop:SinTheta} below relies on a generalisation of the curvature lemma in \citet[Lemma~3.1]{Vuetal2013}, presented as Lemma~\ref{Lemma:CurvatureSingularVector} in the online supplement.
\begin{prop}
\label{Prop:SinTheta}
Suppose that $\hat{M}$ satisfies~\eqref{Eq:convexS} for either $\mathcal{S} = \mathcal{S}_1$ or $\mathcal{S} = \mathcal{S}_2$. Let $\hat{v} \in \argmax_{\tilde{v}\in\mathbb{S}^{p-1}} \|\hat{M}^\top \tilde{v}\|_2$ be the leading left singular vector of $\hat{M}$. If $n\geq 6$ and if we choose $\lambda \geq 2\sigma\sqrt{\log(p\log n)}$, then 
\[
\sup_{P \in \mathcal{P}(n,p,k,1,\vartheta,\tau,\sigma^2)} \mathbb{P}_P\biggl(\sin\angle(\hat{v},v) > \frac{32\lambda\sqrt{k}}{\tau\vartheta\sqrt{n}}\biggr) \leq \frac{4}{(p\log n)^{1/2}}.
\]
\end{prop}
The following corollary restates the rate of convergence of the projection estimator in a simple asymptotic regime.
\begin{cor}
\label{Cor:Proj}
Consider an asymptotic regime where $\log p = O(\log n)$, $\sigma$ is a constant, $\vartheta \asymp n^{-a}$, $\tau \asymp n^{-b}$ and $k\asymp n^{c}$ for some $a\in\mathbb{R}$, $b\in[0,1]$ and $c\geq 0$. Then, setting $\lambda := 2\sigma\sqrt{\log(p\log n)}$ and provided $a+b+c/2 < 1/2$, we have for every $\delta > 0$ that
\[
\sup_{P \in \mathcal{P}(n,p,k,1,\vartheta,\tau,\sigma^2)} \mathbb{P}_P\bigl(\angle(\hat{v},v) > n^{-(1 - 2a - 2b - c)/2 + \delta}\bigr) \rightarrow 0.
\]
\end{cor}
Proposition~\ref{Prop:SinTheta} and Corollary~\ref{Cor:Proj} illustrate the benefits of assuming that the changes in mean structure occur only in a sparse subset of the coordinates.  Indeed, these results mimic similar findings in other high-dimensional statistical problems where sparsity plays a key role, indicating that one pays a logarithmic price for absence of knowledge of the true sparsity set.  See, for instance, \citet{BRT2009} in the context of the Lasso in high-dimensional linear models, or \citet{JohnstoneLu2009,WBS2016} in the context of Sparse Principal Component Analysis.
\begin{algorithm}[htbp!]
\SetAlgoLined
\IncMargin{1em}
\DontPrintSemicolon
\SetKwInput{One}{Step 1}\SetKwInput{Two}{Step 2}\SetKwInput{Three}{Step 3}\SetKwInput{Four}{Step 4}
\KwIn{
$X\in\mathbb{R}^{p\times n}$, $\lambda > 0$.
}
\vskip 0.5ex
\One{Perform the CUSUM transformation $T\leftarrow \mathcal{T}(X)$}
\Two{Use Algorithm~\ref{Algo:ADMM} or~\eqref{Eq:SoftThresholding} (with inputs $T$, $\lambda$ in either case) to solve for an optimiser $\hat{M}$ of~\eqref{Eq:convexS} for $\mathcal{S} = \mathcal{S}_1$ or $\mathcal{S}_2$}
\Three{Find $\hat{v}\in\argmax_{\tilde{v}\in\mathbb{S}^{p-1}} \|\hat{M}^\top \tilde{v}\|_2.$ }
\Four{Let $\hat{z} \in \argmax_{1\leq t\leq n-1} |\hat{v}^\top T_t|$, where $T_t$ is the $t$th column of $T$, and set $\bar{T}_{\max} \leftarrow |\hat{v}^\top T_{\hat{z}}|$}
\vskip 0.5ex
\KwOut{$\hat{z}$, $\bar{T}_{\max}$}
\vskip 1ex
\caption{Pseudo-code for a single high-dimensional changepoint estimation algorithm.}
\label{Algo:SingleCP}
\end{algorithm}

After obtaining a good estimator $\hat{v}$ of the oracle projection direction, the natural next step is to project the data matrix $X$ along the direction $\hat{v}$, and apply an existing one-dimensional changepoint localisation method on the projected data. In this work, we apply a one-dimensional CUSUM transformation to the projected time series and estimate the changepoint by the location of the maximum of the CUSUM vector. Our overall procedure for locating a single changepoint in a high-dimensional time series is given in Algorithm~\ref{Algo:SingleCP}.  In our description of this algorithm, the noise level $\sigma$ is assumed to be known.  If $\sigma$ is unknown, we can estimate it robustly using, e.g., the median absolute deviation of the marginal one-dimensional time series \citep{Hampel1974}. Note that for convenience of later reference, we have required Algorithm~\ref{Algo:SingleCP} to output both the estimated changepoint location $\hat{z}$ and the associated maximum absolute post-projection one-dimensional CUSUM statistic $\bar{T}_{\max}$. 

From a theoretical point of view, the fact that $\hat{v}$ is estimated using the entire dataset $X$ makes it difficult to analyse the post-projection noise structure. For this reason, in the analysis below, we work with a slight variant of Algorithm~\ref{Algo:SingleCP}. We assume for convenience that $n = 2n_1$ is even, and define $X^{(1)},X^{(2)} \in \mathbb{R}^{p\times n_1}$ by
\begin{equation}
\label{Eq:Splitting}
X^{(1)}_{j,t} := X_{j,2t-1}\quad \text{and}\quad X^{(2)}_{j,t} := X_{j,2t} \quad \text{for $1\leq j\leq p, 1\leq t\leq n_1$}.
\end{equation}
We then use $X^{(1)}$ to estimate the oracle projection direction and use $X^{(2)}$ to estimate the changepoint location after projection (see Algorithm~\ref{Algo:SingleCPVariant}). However, we recommend using Algorithm~\ref{Algo:SingleCP} in practice to exploit the full signal strength in the data.

\begin{algorithm}[htbp!]
\SetAlgoLined
\IncMargin{1em}
\DontPrintSemicolon
\SetKwInput{One}{Step 1}\SetKwInput{Two}{Step 2}\SetKwInput{Three}{Step 3}\SetKwInput{Four}{Step 4}
\KwIn{
$X\in\mathbb{R}^{p\times n}$, $\lambda > 0$.
}
\vskip 0.5ex
\One{Perform the CUSUM transformation $T^{(1)}\leftarrow \mathcal{T}(X^{(1)})$ and $T^{(2)} \leftarrow \mathcal{T}(X^{(2)})$.}
\Two{Use Algorithm~\ref{Algo:ADMM} or~\eqref{Eq:SoftThresholding} (with inputs $T^{(1)}$, $\lambda$ in either case) to solve for $\hat{M}^{(1)} \in \argmax_{M\in\mathcal{S}}\bigl\{\langle T^{(1)}, M\rangle - \lambda\|M\|_1\bigr\}$ with $\mathcal{S}= \{M\in\mathbb{R}^{p\times (n_1-1)}: \|M\|_*\leq 1\}$ or $\{M\in\mathbb{R}^{p\times (n_1-1)}: \|M\|_2\leq 1\}$.}
\Three{Find $\hat{v}^{(1)}\in\argmax_{\tilde{v} \in \mathbb{S}^{p-1}} \|(\hat{M}^{(1)})^\top \tilde{v}\|_2$.}
\Four{Let $\hat{z}\in 2\argmax_{1 \leq t \leq n_1-1} \bigl|(\hat{v}^{(1)})^\top T_t^{(2)}\bigr|$, where $T^{(2)}_t$ is the $t$th column of $T^{(2)}$, and set $\bar{T}_{\max}\leftarrow \bigl|(\hat{v}^{(1)})^\top T_{\hat{z}/2}^{(2)}\bigr|$.}
\vskip 0.5ex
\KwOut{$\hat{z}, \bar{T}_{\max}$}
\vskip 1ex
\caption{Pseudo-code for a sample-splitting variant of Algorithm~\ref{Algo:SingleCP}.}
\label{Algo:SingleCPVariant}
\end{algorithm}

We summarise the overall estimation performance of Algorithm~\ref{Algo:SingleCPVariant} in the following theorem.
\begin{thm}
\label{Thm:SingleCP}
Suppose $\sigma>0$ is known. Let $\hat{z}$ be the output of Algorithm~\ref{Algo:SingleCPVariant} with input $X \sim P \in \mathcal{P}(n,p,k,1,\vartheta,\tau,\sigma^2)$ and $\lambda := 2\sigma\sqrt{\log(p\log n)}$.  There exist universal constants $C,C' > 0$ such that if $n\geq 12$ is even, $z$ is even and
\begin{equation}
\label{Eq:T3cond}
\frac{C\sigma}{\vartheta\tau}\sqrt{\frac{k\log (p\log n)}{n}}\leq 1,
\end{equation}
then
\[
\mathbb{P}_P\biggl(\frac{1}{n}|\hat{z} - z| \leq \frac{C'\sigma^2\log\log n}{n\vartheta^2}\biggr) \geq 1-\frac{4}{\{p\log (n/2)\}^{1/2}} - \frac{9}{\log (n/2)}.
\]
\end{thm}
We remark that under the conditions of the theorem, the rate of convergence obtained is minimax optimal up to a factor of $\log \log n$; see Proposition~\ref{Prop:Minimax} in the online supplement.  It is interesting to note that, once~\eqref{Eq:T3cond} is satisfied, the final rate of changepoint estimation does not depend on $\tau$.  
\begin{cor}
\label{Cor:SingleCPRate}
Suppose that $\sigma$ is a constant, $\log p = O(\log n)$, $\vartheta \asymp n^{-a}$, $\tau\asymp n^{-b}$ and $k\asymp n^{c}$ for some $a\in\mathbb{R}$ and $b \in [0,1]$ and $c\geq 0$. If $a+b+c/2<1/2$, then the output $\hat{z}$ of Algorithm~\ref{Algo:SingleCPVariant} with $\lambda := 2\sigma\sqrt{\log(p\log n)}$ is a consistent estimator of the true changepoint $z$ with rate of convergence $\rho_n = o(n^{-1+2a+\delta})$ for any $\delta > 0$.
\end{cor}
Finally in this section, we remark that this asymptotic rate of convergence has previously been observed in \citet[][Theorem~2.8.2]{CsorgoHorvath1997} for a CUSUM procedure in the special case of univariate observations with $\tau$ bounded away from zero (i.e.\ $b=0$ in Corollary~\ref{Cor:SingleCPRate} above). 

\section{Estimating multiple changepoints}
\label{Sec:MultipleCP}
Our algorithm for estimating a single changepoint can be combined with the wild binary segmentation scheme of \citet{Fryzlewicz2014} to locate sequentially multiple changepoints in high-dimensional time series. The principal idea behind a wild binary segmentation procedure is as follows. We first randomly sample a large number of pairs, $(s_1,e_1),\ldots,(s_Q,e_Q)$ uniformly from the set $\{(\ell,r)\in\mathbb{Z}^2: 0\leq \ell < r\leq n\}$, and then apply our single changepoint algorithm to $X^{[q]}$, for $1\leq q\leq Q$, where $X^{[q]}$ is defined to be the submatrix of $X$ obtained by extracting columns $\{s_q+1,\ldots,e_q\}$ of $X$. For each $1\leq q\leq Q$, the single changepoint algorithm (Algorithm~\ref{Algo:SingleCP} or~\ref{Algo:SingleCPVariant}) will estimate an optimal sparse projection direction $\hat{v}^{[q]}$, compute a candidate changepoint location $s_q + \hat{z}^{[q]}$ within the time window $[s_q+1, e_q]$ and return a maximum absolute CUSUM statistic $\bar T_{\max}^{[q]}$ along the projection direction. We aggregate the $q$ candidate changepoint locations by choosing one that maximises the largest projected CUSUM statistic, $T_{\max}^{[q]}$, as our best candidate. If $T_{\max}^{[q]}$ is above a certain threshold value $\xi$, we admit the best candidate to the set $\hat{Z}$ of estimated changepoint locations and repeat the above procedure recursively on the sub-segments to the left and right of the estimated changepoint. Note that while recursing on a sub-segment, we only consider those time windows that are completely contained in the sub-segment. The precise algorithm is detailed in Algorithm~\ref{Algo:MultipleCP}. 

Algorithm~\ref{Algo:MultipleCP} requires three tuning parameters: a regularisation parameter $\lambda$, a Monte Carlo parameter $Q$ for the number of random time windows and a thresholding parameter $\xi$ that determines termination of recursive segmentation. Theorem~\ref{Thm:MultipleCP} below provides choices for $\lambda$, $Q$ and $\xi$ that yield theoretical guarantees for consistent estimation of all changepoints as defined in~\eqref{Eq:ConvergenceRate}. 

\begin{algorithm}[htb]
\SetAlgoLined
\IncMargin{1em}
\DontPrintSemicolon
\SetKwFunction{wbs}{wbs}\SetKwProg{Fn}{Function}{}{end}\SetKwInput{One}{Step 1}\SetKwInput{Two}{Step 2}\SetKwInput{Three}{Step 3}
\KwIn{$X\in\mathbb{R}^{p\times n}$,  $\lambda > 0$, $\xi > 0$, $\beta > 0$, $Q\in\mathbb{N}$.}
\vskip 0.5ex
\One{Set $\hat Z\leftarrow \emptyset$. Draw $Q$ pairs of integers $(s_1, e_1), \ldots, (s_Q, e_Q)$ uniformly at random from the set $\{(\ell,r)\in\mathbb{Z}^2 : 0\leq \ell < r\leq n\}$.}
\Two{Run \wbs{$0$, $n$} where $\wbs$ is defined below.}
\Three{Let $\hat{\nu}\leftarrow |\hat Z|$ and sort elements of $\hat Z$ in increasing order to yield $\hat{z}_1<\cdots<\hat{z}_{\hat{\nu}}$.}
\vskip 0.5ex
\KwOut{$\hat{z}_1,\ldots,\hat{z}_{\hat{\nu}}$}
\vskip 0.5ex
\Fn{\wbs{$s$, $e$}}{
Set $\mathcal{Q}_{s,e}\leftarrow \{q: s + n\beta \leq s_q< e_q\leq e-n\beta\}$\;
\For{$q\in\mathcal{Q}_{s,e}$}{
Run Algorithm~\ref{Algo:SingleCP} with $X^{[q]}$, $\lambda$ as input, and let $\hat z^{[q]}, \bar T_{\max}^{[q]}$ be the output.
}
Find $q_0 \in \argmax_{q\in \mathcal{Q}_{s,e}}\bar T_{\max}^{[q]}$ and set $b \leftarrow  s_{q_0} +\hat z^{[q_0]}$\;
\If{$\bar T_{\max}^{[q_0]} > \xi$}{
$\hat Z\leftarrow \hat Z\cup \{b\}$\;
\wbs{$s$, $b$}\;
\wbs{$b$, $e$}
}
}
\vskip 1ex
\caption{Pseudo-code for multiple changepoint algorithm based on sparse singular vector projection and wild binary segmentation.}
\label{Algo:MultipleCP}
\end{algorithm}

We remark that if we apply Algorithm~\ref{Algo:SingleCP} or~\ref{Algo:SingleCPVariant} on the entire dataset $X$ instead of random time windows of $X$, and then iterate after segmentation, we arrive at a multiple changepoint algorithm based on the classical binary segmentation scheme. The main disadvantage of this classical binary segmentation procedure is its sensitivity to model misspecification. Algorithms~\ref{Algo:SingleCP} and~\ref{Algo:SingleCPVariant} are designed to optimise the detection of a single changepoint. When we apply them in conjunction with classical binary segmentation to a time series containing more than one changepoint, the signals from multiple changepoints may cancel each other out in two different ways that will lead to a loss of power. First, as \citet{Fryzlewicz2014} points out in the one-dimensional setting, multiple changepoints may offset each other in CUSUM computation, resulting in a smaller peak of the CUSUM statistic that is more easily contaminated by the noise. Moreover, in a high-dimensional setting, different changepoints can undergo changes in different sets of (sparse) coordinates. This also attenuates the signal strength in the sense that the estimated oracle projection direction from Algorithm~\ref{Algo:ADMM} is aligned to some linear combination of $\theta^{(1)},\ldots,\theta^{(\nu)}$, but not necessarily well-aligned to any one particular $\theta^{(i)}$. The wild binary segmentation scheme addresses the model misspecification issue by examining sub-intervals of the entire time length. When the number of time windows $Q$ is sufficiently large and $\tau$ is not too small, with high probability we have reasonably long time windows that contain each individual changepoint. Hence the single changepoint algorithm will perform well on these segments. 

Just as in the case of single changepoint detection, it is easier to analyse the theoretical performance of a sample-splitting version of Algorithm~\ref{Algo:MultipleCP}. However, to avoid notational clutter, we will prove a theoretical result without sample splitting, but with the assumption that whenever Algorithm~\ref{Algo:SingleCP} is used within Algorithm~\ref{Algo:MultipleCP}, its second and third steps (i.e.\ the steps for estimating the oracle projection direction) are carried out on an independent copy $X'$ of $X$. We refer to such a variant of the algorithm with an access to an independent sample $X'$ as Algorithm~\ref{Algo:MultipleCP}$'$. Theorem~\ref{Thm:MultipleCP} below, which proves theoretical guarantees of Algorithm~\ref{Algo:MultipleCP}$'$, can then be readily adapted to work for a sample-splitting version of Algorithm~\ref{Algo:MultipleCP}, where we replace $n$ by $n/2$ where necessary. 
\begin{thm}
\label{Thm:MultipleCP}
Suppose $\sigma >0$ is known and $X,X'\stackrel{\mathrm{iid}}{\sim} P\in\mathcal{P}(n,p,k,\nu,\vartheta,\tau,\sigma^2)$. Let $\hat{z}_1 < \cdots < \hat{z}_{\hat{\nu}}$ be the output of Algorithm~\ref{Algo:MultipleCP}$'$ with input $X$, $X'$, $\lambda := 4\sigma\sqrt{\log(np)}$, $\xi := \lambda$, $\beta$ and $Q$. Define $\rho = \rho_n := \lambda^2 n^{-1}\vartheta^{-2}\tau^{-4}$, and assume that $n\tau\geq 14$.  There exist universal constants $C,C' > 0$ such that if $\rho < \beta/2 \leq \tau/C$ and $C \rho k\tau^2 \leq 1$, then
\[
\mathbb{P}_P\bigl\{\text{$\hat{\nu} = \nu$ and $|\hat{z}_i - z_i|\leq C'n\rho$ for all $1\leq i\leq \nu$}\bigr\} \geq 1 - \tau^{-1}e^{-\tau^2Q/9} - 4n^{-1}p^{-4}\log n.
\]
\end{thm}
\begin{cor}
\label{Cor:MultipleCP}
Suppose that $\sigma$ is a constant, $\vartheta\asymp n^{-a}$, $\tau\asymp n^{-b}$, $k\asymp n^{c}$ and $\log p = O(\log n)$.  If $a+b+c/2<1/2$ and $2a+5b<1$, then there exists $\beta = \beta_n$ such that Algorithm~\ref{Algo:MultipleCP}$'$ with $\lambda := 4\sigma\sqrt{\log(np)}$ consistently estimates all changepoints with rate of convergence $\rho_n = o(n^{-(1-2a-4b)+\delta})$ for any $\delta > 0$.
\end{cor}
We remark that the consistency described in Corollary~\ref{Cor:MultipleCP} is a rather strong notion, in the sense that it implies convergence in several other natural metrics. For example, if we let 
\[
d_{\mathrm{H}}(A,B) := \max\Bigl\{\sup_{a\in A}\inf_{b\in B} |a-b|, \sup_{b\in B} \inf_{a\in A} |a-b|\Bigr\}
\]
denote the Hausdorff distance between non-empty sets $A$ and $B$ on $\mathbb{R}$, then~\eqref{Eq:ConvergenceRate} implies that with probability tending to 1, 
\[
\frac{1}{n}d_{\mathrm{H}}\bigl(\{\hat{z}_i: 1\leq i\leq \hat{\nu}\}, \{z_i : 1\leq i\leq \nu\}\bigr) \leq \rho_n.
\]
Similarly, denote the $L_1$-Wasserstein distance between probability measures $P$ and $Q$ on $\mathbb{R}$ by
\[
d_{\mathrm{W}}(P,Q) := \inf_{(U,V) \sim (P,Q)}\mathbb{E}|U-V|,
\]
where the infimum is taken over all pairs of random variables $U$ and $V$ defined on the same probability space with $U \sim P$ and $V \sim Q$.  Then~\eqref{Eq:ConvergenceRate} also implies that with probability tending to 1,
\[
\frac{1}{n}d_{\mathrm{W}}\biggl(\frac{1}{\hat\nu}\sum_{i=1}^{\hat\nu} \delta_{\hat{z}_i}, \frac{1}{\nu}\sum_{i=1}^\nu \delta_{z_i}\biggr) \leq \rho_n,
\]
where $\delta_a$ denotes a Dirac point mass at $a$.

\section{Numerical studies}
\label{Sec:Simulation}
In this section, we examine the empirical performance of the \texttt{inspect} algorithm in a range of settings, and compare it with a variety of other recently-proposed methods.  In both single- and multiple-changepoint scenarios, the implementation of \texttt{inspect} requires the choice of a regularisation parameter $\lambda > 0$ to be used in Algorithm~\ref{Algo:ADMM} (which is called in Algorithms~\ref{Algo:SingleCP} and~\ref{Algo:MultipleCP}).  In our experience, the theoretical choices $\lambda = 2\sigma\sqrt{\log(p\log n)}$ and $\lambda = 4\sigma\sqrt{\log(np)}$ used in Theorems~\ref{Thm:SingleCP} and~\ref{Thm:MultipleCP} produce consistent estimators as predicted by the theory, but are slightly conservative, and in practice we recommend the choice $\lambda = \sigma\sqrt{2^{-1}\log(p\log n)}$ in both cases.  Figure~\ref{Fig:Lambda} illustrates the dependence of the performance of our algorithm on the regularisation parameter, and reveals in this case (as in the other examples that we tried) that this choice of $\lambda$ is sensible. In the implementation of our algorithm, we do not assume the noise level $\sigma$ is known, nor even that it is constant across different components.  Instead, we estimate the error variance for each individual time series using the median absolute deviation of first-order differences with scaling constant of $1.05$ for the normal distribution \citep{Hampel1974}.  We then normalise each series by its estimated standard deviation and use the choices of $\lambda$ given above with $\sigma$ replaced by 1.   

\begin{figure}[htbp]
\begin{center}
\begin{tabular}{cc}
\includegraphics[width=0.45\textwidth]{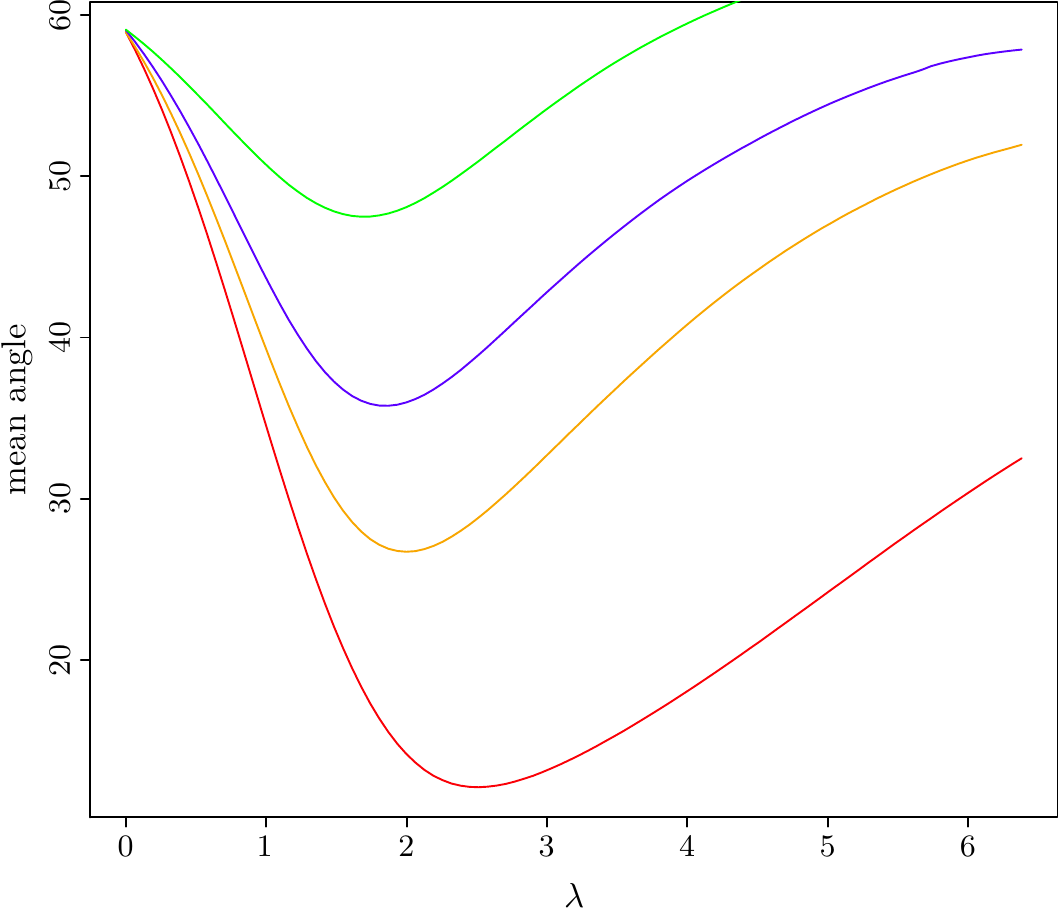} 
&
\includegraphics[width=0.45\textwidth]{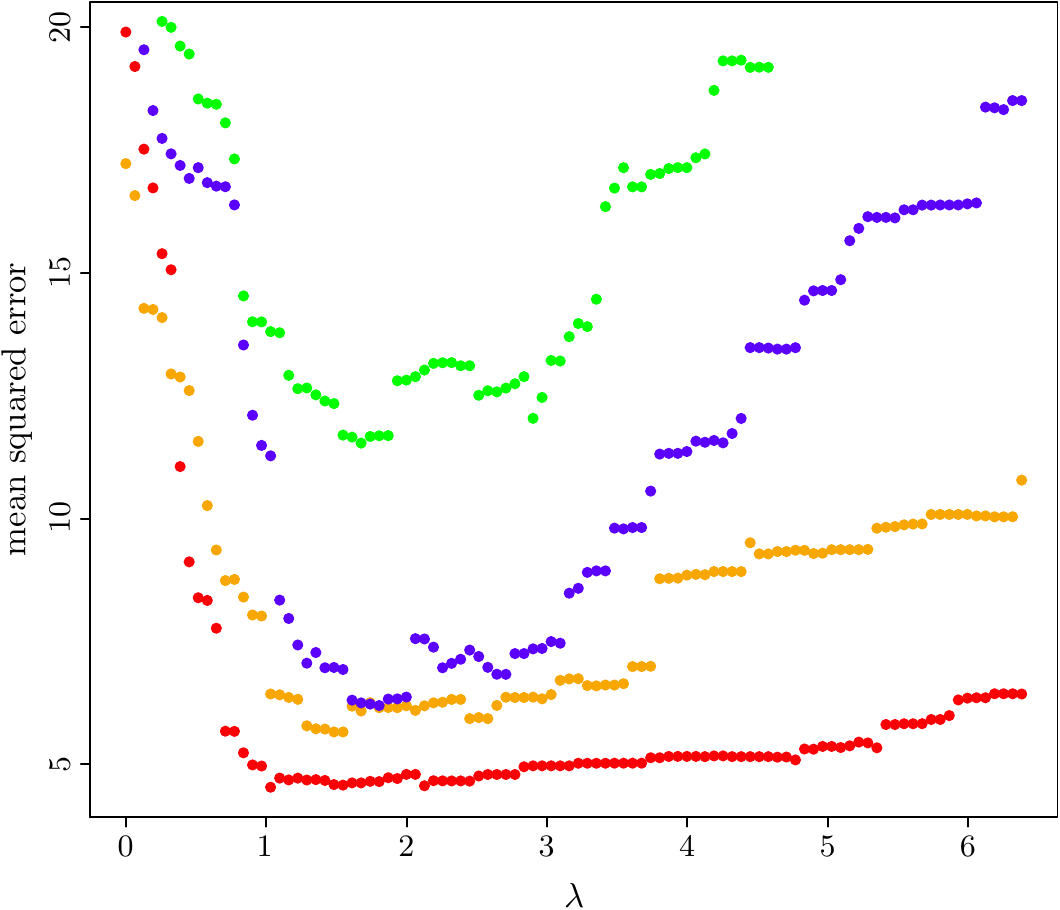}
\end{tabular}
\caption{Dependence of estimation performance on $\lambda$. Left panel: mean angle in degrees between estimated projection direction and oracle projection direction over 100 experiments. Right panel: mean squared error of estimated changepoint location over 100 experiments. Parameters: $n = 1000$, $p = 500$, $k = 3$ (red) or $10$ (orange) or $22$ (blue) or $100$ (green), $z = 400$, $\vartheta = 1$, $\sigma^2 = 1$.  For these parameters, our choice of $\lambda$ is $\sigma\sqrt{2^{-1}\log(p\log n)} \approx 2.02$.}
\label{Fig:Lambda}
\end{center}
\end{figure}



In Step~2 of Algorithm~\ref{Algo:SingleCP}, we also have a choice between using $\mathcal{S} = \mathcal{S}_1$ and $\mathcal{S}_2$. The following numerical experiment demonstrates the difference in performance of the algorithm for these two choices.  We took $n = 500$, $p = 1000$, $k = 30$ and $\sigma^2=1$, with a single changepoint located at $z = 200$. Table~\ref{Tab:Angles} shows the angles between the oracle projection direction and estimated projection directions using both $\mathcal{S}_1$ and $\mathcal{S}_2$ as the signal level $\vartheta$ varies from $0.5$ to $5.0$.  We have additionally reported the benchmark performance of the naive estimator using the leading left singular vector of $T$, which illustrates that the convex optimisation algorithms significantly improve the naive estimator by exploiting the sparsity structure.  It can be seen that further relaxation from $\mathcal{S}_1$ to $\mathcal{S}_2$ incurs a relatively low cost in terms of the estimation quality of the projection direction, but it offers great improvement in running time due to the closed-form solution (cf.\ Proposition~\ref{Lemma:Duality} in the online supplement).  Thus, even though the use of $\mathcal{S}_1$ remains a viable practical choice for offline data sets of moderate size, we use $\mathcal{S}=\mathcal{S}_2$ in the simulations that follow. 
\begin{table}[htbp]
\begin{center}
\begin{tabular}{c|cccccccccc}
\hline
$\vartheta$ & 0.5 & 1.0 & 1.5 & 2.0 & 2.5 & 3.0 & 3.5 & 4.0 & 4.5 & 5.0\\
\hline
$\angle(\hat{v}_{\mathcal{S}_1}, v)$ & 75.3 & 60.2 & 44.6 & 32.1 & 24.0 & 19.7 & 15.9 & 12.6 & 10.0 & 7.7\\
$\angle(\hat{v}_{\mathcal{S}_2}, v)$ & 75.7 & 61.7 & 46.8 & 34.4 & 26.5 & 21.7 & 18.1 & 15.2 & 12.2 & 10.2\\
$\angle(\hat{v}_{\max}, v)$ & 83.4 & 77.2 & 64.8 & 57.1 & 51.5 & 47.4 & 44.5 & 40.8 & 38.1 & 35.2\\
\hline
\end{tabular}
\caption{Angles (in degrees) between oracle projection direction $v$ and estimated projection directions $\hat{v}_{\mathcal{S}_1}$ (using $\mathcal{S}_1$), $\hat{v}_{\mathcal{S}_2}$ (using $\mathcal{S}_2$) and $\hat{v}_{\max}$ (leading left singular vector of $T$), for different choices of $\vartheta$. Each reported value is averaged over 100 repetitions. Other simulation parameters: $n = 500$, $p = 1000$, $k = 30$, $z = 200$, $\sigma^2 = 1$.}
\label{Tab:Angles}
\end{center} 
\end{table}

We compare the performance of the \texttt{inspect} algorithm with the following recently proposed methods for high-dimensional changepoint estimation. These include sparsified binary segmentation (\texttt{sbs}) \citep{ChoFryzlewicz2015}, the double CUSUM algorithm (\texttt{dc}) of \citet{Cho2016}, a scan statistic-based algorithm (\texttt{scan}) derived from the work of \citet{EnikeevaHarchaoui2014}, the $\ell_\infty$ CUSUM aggregation algorithm ($\texttt{agg}_\infty$) of \citet{Jirak2015} and the $\ell_2$ CUSUM aggregation algorithm ($\texttt{agg}_2$) of \citet{HorvathHuskova2012}.
We remark that the latter three works primarily concern the test for the existence of a changepoint. However, their relevant test statistics can be naturally modified into a changepoint location estimator.  They can then be extended a multiple changepoint estimation algorithm via a wild binary segmentation scheme in a similar way to our algorithm, in which the termination criterion is chosen by five-fold cross validation.  Whenever tuning parameters are required in running these algorithms, we adopt the choices suggested by their authors in the relevant papers.

\subsection{Single changepoint estimation}
\label{Sec:SingleCPSim}

All algorithms in our simulation study are top-down algorithms in the sense that their multiple changepoint procedure is built upon a single changepoint estimation submodule, which is used to locate recursively all changepoints via a (wild) binary segmentation scheme. It is therefore instructive first to compare their performance in the single changepoint estimation task.  Our simulations were run for $ n, p \in\{500, 1000,2000\}$, $k \in \{3,\lceil p^{1/2}\rceil, 0.1p, p\}$, $z = 0.4n$, $\sigma^2=1$ and $\vartheta = 0.8$, with $\theta \propto (1,2^{-1/2},\ldots,k^{-1/2},0,\ldots,0)^\top \in\mathbb{R}^p$.  For definiteness, we let the $n$ columns of $X$ be independent, with the leftmost $z$ columns drawn from $N_p(0, \sigma^2 I_p)$ and the remaining columns drawn from $N_p(\theta, \sigma^2 I_p)$.  To avoid the influence of different threshold levels on the performance of the algorithms and to focus solely on their estimation precision, we assume that the existence of a single changepoint is known \emph{a priori} and make all algorithms output their estimate of its location; estimation of the number of changepoints in a multiple-changepoint setting is studied in Section~\ref{Sec:MultipleCPSim} below.  Table~\ref{Tab:SingleCP} compares the performance of \texttt{inspect} and other competing algorithms under various parameter settings. All algorithms were run on the same data matrices and the root mean squared estimation error over 1000 repetitions is reported.  Although, in the interests of brevity, we report the root mean squared estimation error only for $\vartheta=0.8$, simulation results for other values of $\vartheta$ were qualitatively similar.  We also remark that the four choices for the parameter $k$ correspond to constant/logarithmic sparsity, polynomial sparsity and two levels of non-sparse settings respectively.  In addition to comparing the practical algorithms, we also computed the changepoint estimator based on the oracle projection direction (which of course is typically unknown); the performance of this oracle estimator depends only on $n$, $z$, $\vartheta$ and $\sigma^2$ (and not on $k$ or $p$), and the corresponding root mean squared errors in Table~\ref{Tab:SingleCP} were $10.0$, $8.1$ and $7.8$ when $(n,z,\vartheta,\sigma^2) = (500,200,0.8,1), (1000,400,0.8,1), (2000,800,0.8,1)$ respectively.  Thus the performance of our \texttt{inspect} algorithm is very close to that of the oracle estimator when $k$ is small, as predicted by our theory.

As a graphical illustration of the performance of the different methods, Figure~\ref{Fig:SingleCP} displays density estimates of their estimated changepoint locations in two different settings taken from Table~\ref{Tab:SingleCP}.  One difficulty in presenting such estimates with kernel density estimators is the fact that different algorithms would require different choices of bandwidth, and these would need to be locally adaptive, due to the relatively sharp peaks.  In order to avoid the choice of bandwidth skewing the visual representation, we therefore use the log-concave maximum likelihood estimators for each method \citep[e.g.][]{DumbgenRufibach2009,CuleSamworthStewart2010}, which is both locally adaptive and tuning-parameter free.
\begin{figure}[htbp]
\begin{center}
\begin{tabular}{cc}
\includegraphics[width=0.45\textwidth]{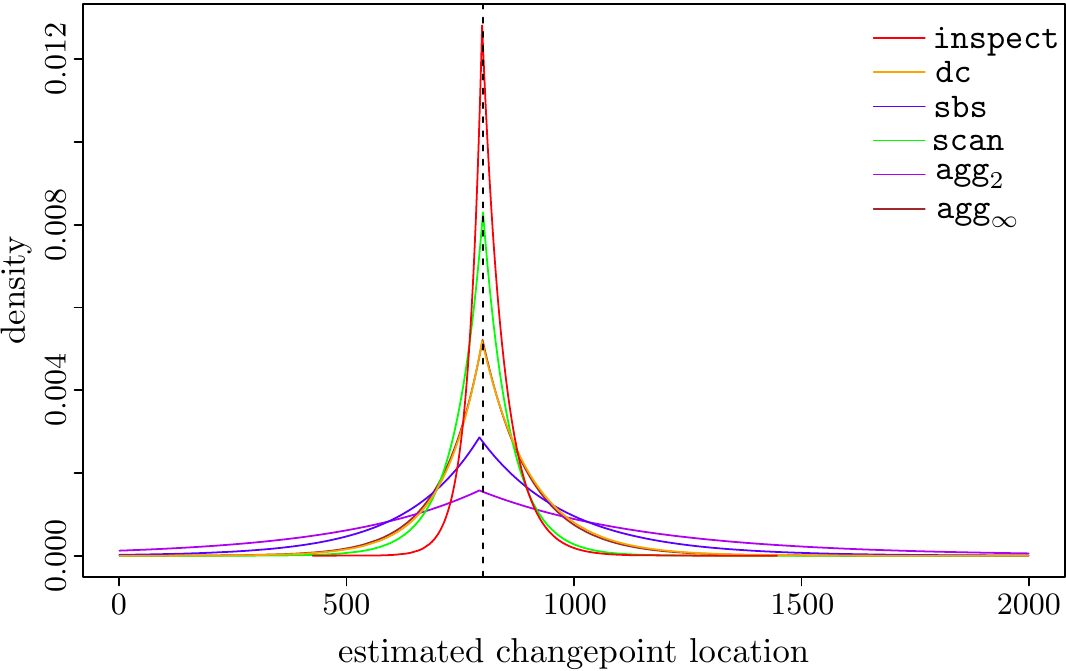} 
&
\includegraphics[width=0.45\textwidth]{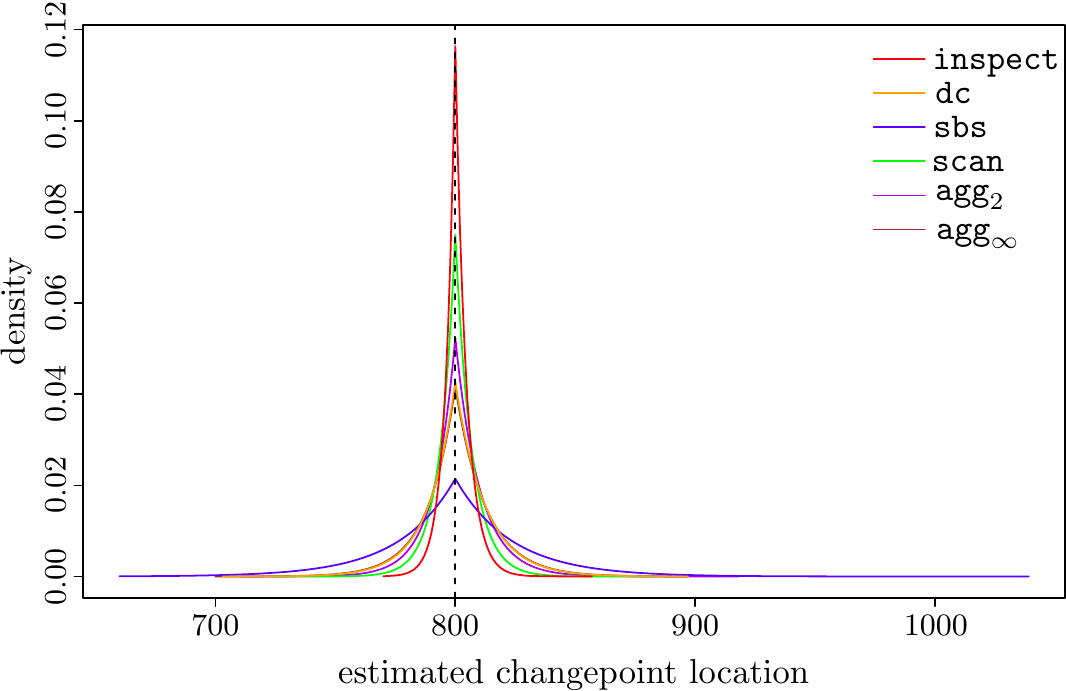}
\end{tabular}
\caption{Estimated densities of location of changepoint estimates by \texttt{inspect}, \texttt{dc}, \texttt{sbs} and \texttt{scan}. Left panel: $(n,p,k,z,\vartheta,\sigma^2) = (2000,1000,32,800,0.5,1)$; right panel: $(n,p,k,z,\vartheta,\sigma^2) = (2000,1000,32,800,1,1)$.}
\label{Fig:SingleCP}
\end{center}
\end{figure}

\begin{table}[htbp]
\begin{center}
\begin{tabular}{cccccccccc}
\hline\hline
$n$ & $p$ & $k$ & $z$ & \texttt{inspect} & \texttt{dc} & \texttt{sbs} & \texttt{scan} & $\texttt{agg}_2$ &  $\texttt{agg}_\infty$\\
\hline
$500$ & $500$ & $3$ & $200$ & $\mathbf{11.2}$ & $22.2$ & $72.7$ & $11.6$ & $115.9$ & $22.4$\\
$500$ & $500$ & $22$ & $200$ & $\mathbf{31.0}$ & $80.8$ & $87.1$ & $65.7$ & $113.2$ & $83.1$\\
$500$ & $500$ & $50$ & $200$ & $\mathbf{35.3}$ & $105.9$ & $102.9$ & $86.8$ & $112.7$ & $107.9$\\
$500$ & $500$ & $500$ & $200$ & $\mathbf{48.8}$ & $147.7$ & $129.6$ & $120.0$ & $114.6$ & $150.8$\\
$500$ & $1000$ & $3$ & $200$ & $\mathbf{13.0}$ & $21.3$ & $83.6$ & $14.3$ & $145.6$ & $19.6$\\
$500$ & $1000$ & $32$ & $200$ & $\mathbf{34.9}$ & $104.6$ & $114.9$ & $95.0$ & $144.9$ & $107.5$\\
$500$ & $1000$ & $100$ & $200$ & $\mathbf{45.0}$ & $124.8$ & $132.0$ & $122.9$ & $145.3$ & $133.6$\\
$500$ & $1000$ & $1000$ & $200$ & $\mathbf{55.0}$ & $140.4$ & $146.5$ & $146.8$ & $144.2$ & $159.5$\\
$500$ & $2000$ & $3$ & $200$ & $\mathbf{18.4}$ & $56.0$ & $99.4$ & $26.4$ & $163.0$ & $26.6$\\
$500$ & $2000$ & $45$ & $200$ & $\mathbf{43.5}$ & $152.3$ & $133.8$ & $126.8$ & $164.9$ & $132.6$\\
$500$ & $2000$ & $200$ & $200$ & $\mathbf{52.8}$ & $159.1$ & $151.6$ & $150.6$ & $163.2$ & $158.4$\\
$500$ & $2000$ & $2000$ & $200$ & $\mathbf{59.6}$ & $162.1$ & $162.4$ & $166.1$ & $163.0$ & $176.0$\\
$1000$ & $500$ & $3$ & $400$ & $\mathbf{8.4}$ & $12.5$ & $101.1$ & $8.6$ & $65.4$ & $13.9$\\
$1000$ & $500$ & $22$ & $400$ & $\mathbf{14.1}$ & $44.2$ & $60.6$ & $18.7$ & $66.7$ & $44.4$\\
$1000$ & $500$ & $50$ & $400$ & $\mathbf{19.7}$ & $61.5$ & $72.1$ & $24.7$ & $66.7$ & $62.4$\\
$1000$ & $500$ & $500$ & $400$ & $\mathbf{36.8}$ & $137.8$ & $114.8$ & $77.4$ & $72.8$ & $142.6$\\
$1000$ & $1000$ & $3$ & $400$ & $9.5$ & $14.6$ & $117.2$ & $\mathbf{9.0}$ & $154.9$ & $15.0$\\
$1000$ & $1000$ & $32$ & $400$ & $\mathbf{20.7}$ & $61.1$ & $83.6$ & $26.4$ & $150.1$ & $57.2$\\
$1000$ & $1000$ & $100$ & $400$ & $\mathbf{33.1}$ & $101.0$ & $122.0$ & $59.2$ & $158.3$ & $106.4$\\
$1000$ & $1000$ & $1000$ & $400$ & $\mathbf{57.7}$ & $159.9$ & $186.3$ & $145.2$ & $152.7$ & $195.2$\\
$1000$ & $2000$ & $3$ & $400$ & $10.8$ & $15.4$ & $132.9$ & $\mathbf{10.3}$ & $232.8$ & $15.5$\\
$1000$ & $2000$ & $45$ & $400$ & $\mathbf{29.6}$ & $121.0$ & $137.0$ & $39.1$ & $237.5$ & $73.4$\\
$1000$ & $2000$ & $200$ & $400$ & $\mathbf{47.4}$ & $176.8$ & $187.7$ & $123.6$ & $235.4$ & $158.2$\\
$1000$ & $2000$ & $2000$ & $400$ & $\mathbf{67.2}$ & $219.6$ & $240.0$ & $210.3$ & $233.4$ & $245.8$\\
$2000$ & $500$ & $3$ & $800$ & $\mathbf{8.6}$ & $15.5$ & $159.7$ & $\mathbf{8.6}$ & $22.6$ & $15.5$\\
$2000$ & $500$ & $22$ & $800$ & $\mathbf{12.4}$ & $31.2$ & $48.7$ & $17.0$ & $25.9$ & $32.1$\\
$2000$ & $500$ & $50$ & $800$ & $\mathbf{14.6}$ & $39.6$ & $57.7$ & $20.4$ & $25.3$ & $38.6$\\
$2000$ & $500$ & $500$ & $800$ & $\mathbf{23.9}$ & $72.7$ & $86.1$ & $35.6$ & $25.1$ & $71.8$\\
$2000$ & $1000$ & $3$ & $800$ & $\mathbf{8.1}$ & $14.2$ & $178.3$ & $8.3$ & $42.6$ & $14.4$\\
$2000$ & $1000$ & $32$ & $800$ & $\mathbf{12.5}$ & $36.1$ & $58.7$ & $16.9$ & $40.6$ & $38.2$\\
$2000$ & $1000$ & $100$ & $800$ & $\mathbf{17.0}$ & $46.7$ & $75.8$ & $24.6$ & $40.0$ & $47.3$\\
$2000$ & $1000$ & $1000$ & $800$ & $\mathbf{31.0}$ & $89.0$ & $111.2$ & $45.4$ & $39.9$ & $91.0$\\
$2000$ & $2000$ & $3$ & $800$ & $9.3$ & $15.9$ & $215.7$ & $\mathbf{9.0}$ & $143.6$ & $16.1$\\
$2000$ & $2000$ & $45$ & $800$ & $\mathbf{16.7}$ & $35.8$ & $100.7$ & $21.3$ & $152.5$ & $39.2$\\
$2000$ & $2000$ & $200$ & $800$ & $\mathbf{25.6}$ & $56.7$ & $126.5$ & $32.0$ & $151.8$ & $59.1$\\
$2000$ & $2000$ & $2000$ & $800$ & $\mathbf{48.4}$ & $107.9$ & $208.0$ & $66.1$ & $150.6$ & $153.5$\\
\hline
\hline
\end{tabular}
\caption{Root mean squared error for \texttt{inspect}, \texttt{dc}, \texttt{sbs}, \texttt{scan}, $\texttt{agg}_2$ and $\texttt{agg}_\infty$ in single changepoint estimation. The smallest root mean squared error is given in bold.  Other parameters: $\vartheta = 0.8$, $\sigma^2 = 1$.}
\label{Tab:SingleCP}
\end{center}
\end{table}

It can be seen from Table~\ref{Tab:SingleCP} and Figure~\ref{Fig:SingleCP} that \texttt{inspect} has extremely competitive performance for the single changepoint estimation task.  In particular, despite the fact that  it is designed for estimation of sparse changepoints, \texttt{inspect} performs relatively well even when $k = p$ (i.e.\ when the signal is highly non-sparse), especially when the signal strength is relatively large. 

\subsection{Model misspecification}

We now extend the ideas of Section~\ref{Sec:SingleCPSim} by investigating empirical performance under several other types of model misspecification. Recall that the noise matrix is $W = (W_{j,t}) := X - \boldsymbol{\mu}$ and we define $W_1,\ldots,W_n$ to be the column vectors of $W$. In models $\mathrm{M}_{\mathrm{unif}}$ and $\mathrm{M}_{\mathrm{exp}}$, we replace Gaussian noise by $W_{j,t}\stackrel{\mathrm{iid}}{\sim}\mathrm{Unif}[-\sqrt{3}\sigma,\sqrt{3}\sigma]$ and $W_{j,t}\stackrel{\mathrm{iid}}{\sim}\mathrm{Exp}(\sigma)-\sigma$ respectively.  We note that the correct Hampel scaling constants are approximately $0.99$ and $1.44$ in these two cases, though we continue to use the constant $1.05$ for normally distributed data.  In model $\mathrm{M}_{\mathrm{cs,loc}}(\rho)$, we allow the noise to have a short-range cross-sectional dependence by sampling $W_1,\ldots,W_n \stackrel{\mathrm{iid}}{\sim}N_p(0, \Sigma)$ for $\Sigma := (\rho^{|j-j'|})_{j,j'}$.  In model $\mathrm{M}_{\mathrm{cs}}(\rho)$, we extend this to global cross-sectional dependence by sampling $W_1,\ldots,W_n \stackrel{\mathrm{iid}}{\sim}N_p(0, \Sigma)$ for $\Sigma := (1-\rho)I_p + \frac{\rho}{p} \mathbf{1}_{p} \mathbf{1}_{p}^\top$, where $\mathbf{1}_{p}\in\mathbb{R}^p$ is an all-one vector.  In model $\mathrm{M}_{\mathrm{temp}}(\rho)$, we consider an auto-regressive AR(1) temporal dependence in the noise by first sampling $W'_{j,t}\stackrel{\mathrm{iid}}{\sim} N(0,\sigma^2)$ and then setting $W_{j,1} := W'_{j,1}$ and $W_{j,t} := \rho^{1/2}W_{j,t-1} + (1-\rho)^{1/2}W'_{j,t}$ for $2\leq t\leq n$.  In $\mathrm{M}_{\mathrm{async}}(L)$, we model asynchronous changepoint location in the signal coordinates by drawing changepoint locations for individual coordinates independently from a uniform distribution on $\{z-L,\ldots,z+L\}$.  We report the performance of the different algorithms in the parameter setting $n = 2000$, $p = 1000$, $k = 32$, $z = 800$, $\vartheta = 0.25$, $\sigma^2=1$ in Table~\ref{Tab:Misspecification}. It can be seen that \texttt{inspect} is robust to both temporal and spatial dependence structures, as well as noise misspecification. 
\begin{table}[htbp]
\begin{center}
\begin{tabular}{cccccccccccc}
\hline\hline
Model & $n$ & $p$ & $k$ & $z$ & $\vartheta$ & \texttt{inspect} & \texttt{dc} & \texttt{sbs} & \texttt{scan} & $\texttt{agg}_2$ & $\texttt{agg}_\infty$\\
\hline
$\mathrm{M}_{\mathrm{unif}}$& 	 2000&	1000&	32&	800&	1.5&	\textbf{2.7}&	9.6&	17.1&	4.9 & 4.3 & 10.2\\
$\mathrm{M}_{\mathrm{exp}}$& 	 2000&	1000&	32&	800&	1.5&	\textbf{2.6}&	9.6&	42.6&	5.0 & 4.7 & 9.6\\
$\mathrm{M}_{\mathrm{cs,loc}}(0.2)$& 2000&	1000&	32&	800&	1.5&	\textbf{3.5}&	9.7&	19.2&	7.0 & 5.4 & 9.8\\
$\mathrm{M}_{\mathrm{cs,loc}}(0.5)$& 2000&	1000&	32&	800&	1.5&	\textbf{5.8}&	9.7&	24.6&	8.7 & 9.3 & 9.6\\
$\mathrm{M}_{\mathrm{cs}}(0.5)$&    2000&	1000&	32&	800&	1.5&	\textbf{1.5}&	7.7&	14.9&	3.0 & 3.6 & 6.7\\
$\mathrm{M}_{\mathrm{cs}}(0.9)$&	 2000&	1000&	32&	800&	1.5&	\textbf{2.7}&	9.9&	18.6&	4.7& 4.7 & 9.6\\
$\mathrm{M}_{\mathrm{temp}}(0.1)$&  2000&	1000&	32&	800&	1.5&	\textbf{6.1}&	20.3&	102.8&	9.4& 10.9 & 20.2\\
$\mathrm{M}_{\mathrm{temp}}(0.3)$&  2000&	1000&	32&	800&	1.5&	\textbf{30.1}&	32.4&	276.4&	38.8 & 38.2 & 34.8\\
$\mathrm{M}_{\mathrm{async}}(10)$&  2000&	1000&	32&	800&	1.5&	\textbf{5.8}&	11.5&	18.5&	7.8& 7.0 & 11.3\\
\hline\hline
\end{tabular}
\caption{Root mean squared error for \texttt{inspect}, \texttt{dc}, \texttt{sbs}, \texttt{scan}, $\texttt{agg}_2$ and $\texttt{agg}_\infty$ in single changepoint estimation, under different forms of model misspecification.}
\label{Tab:Misspecification}
\end{center}
\end{table}

\subsection{Multiple changepoint estimation}
\label{Sec:MultipleCPSim}

The use of the `burn-off' parameter $\beta$ in Algorithm~\ref{Algo:MultipleCP} was mainly to facilitate our theoretical analysis. In our simulations, we found that taking $\beta = 0$ rarely resulted in the changepoint being estimated more than once, and we therefore recommend setting $\beta = 0$ in practice, unless prior knowledge of the distribution of the changepoints suggests otherwise.  To choose $\xi$ in the multiple changepoint estimation simulation studies, for each $(n,p)$, we first applied \texttt{inspect} to 1000 data sets drawn from the null model with no changepoint, and took $\xi$ to be the largest value of $\bar{T}_{\mathrm{max}}$ from Algorithm~\ref{Algo:SingleCP}.  We also set $Q = 1000$. 

We consider the simulation setting where $n = 2000$, $p = 200$, $k = 40$, $\sigma^2=1$ and $z = (500, 1000, 1500)$. Define $\vartheta^{(i)} := \|\theta^{(i)}\|_2$ to be the signal strength at the $i$th changepoint. We set $(\vartheta^{(1)}, \vartheta^{(2)}, \vartheta^{(3)}) = (\vartheta,2\vartheta,3\vartheta)$ and take $\vartheta \in \{0.4,0.6\}$ to see the performance of the algorithms at different signal strengths.  We also considered different levels of overlap between the coordinates in which the three changes in mean structure occur: in the \emph{complete overlap} case, changes occur in the same $k$ coordinates at each changepoint; in the \emph{half overlap} case, the changes occur in coordinates $\frac{i-1}{2}k + 1,\ldots,\frac{i+1}{2}k$ for $i = 1,2,3$; in the \emph{no overlap} case, the changes occur in disjoint sets of coordinates.  Table~\ref{Tab:MultipleCP} summarises the results. We report both the frequency counts of the number of changepoints detected over 100 runs (all algorithms were compared over the same set of randomly generated data matrices) and two quality measures of the location of changepoints.  In particular, since changepoint estimation can be viewed as a special case of classification, the quality of the estimated changepoints can be measured by the Adjusted Rand Index (ARI) of the estimated segmentation against the truth \citep{Rand1971,HubertArabie1985}.  We report both the average ARI over all runs and the percentage of runs for which a particular method attains the largest ARI among the six.  Figure~\ref{Fig:MultipleCP} gives a pictorial representation of the results for one particular collection of parameter settings.  Again, we find that the performance of \texttt{inspect} is very encouraging on all performance measures, though we remark that $\texttt{agg}_2$ is also competitive, and \texttt{scan} tends to output the fewest false positives.
\begin{table}[htbp]
\begin{center}
\begin{tabular}{c|*{9}{c}}
\hline\hline
\multirow{2}{*}{$(\vartheta^{(1)},\vartheta^{(2)},\vartheta^{(3)})$} & \multirow{2}{*}{method} & \multicolumn{6}{c}{$\hat\nu$} & \multirow{2}{*}{ARI} & \multirow{2}{*}{\% best}\\
& & 0 & 1 & 2 & 3 & 4 & 5 & & \\
\hline
\multirow{6}{*}{$(0.6,1.2,1.8)$} & \texttt{inspect}  & 0 & 0 & 20 & 72 & 8 & 0 & 0.90 & 55 \\
& \texttt{dc}  & 0 & 0 & 21 & 54 & 23 & 2 & 0.85 & 22\\
& \texttt{sbs}  & 0 & 0 & 12 & 64 & 22 & 2 & 0.86 & 15\\
& \texttt{scan}  & 0 & 0 & 72 & 27 & 1 & 0 & 0.77 & 8\\
& $\texttt{agg}_2$  & 0 & 0 & 18 & 73 & 8 & 1 & 0.87 & 1\\
& $\texttt{agg}_\infty$  & 0 & 0 & 29 & 57 & 13 & 1 & 0.83 & 17\\

\hline
\multirow{6}{*}{$(0.4,0.8,1.2)$} & \texttt{inspect}  & 0 & 0 & 62 & 34 & 4 & 0 & 0.74 & 50 \\
& \texttt{dc}  & 0 & 0 & 62 & 32 & 5 & 1 & 0.69 & 19\\
& \texttt{sbs}  & 0 & 0 & 54 & 44 & 1 & 1 & 0.70 & 21\\
& \texttt{scan}  & 0 & 2 & 95 & 3 & 0 & 0 & 0.68 & 19\\
& $\texttt{agg}_2$  & 0 & 0 & 81 & 17 & 2 & 0 & 0.71 & 2\\
& $\texttt{agg}_\infty$  & 0 & 0 & 68 & 29 & 3 & 0 & 0.68 & 8\\
\hline
\hline
\multirow{6}{*}{$(0.6,1.2,1.8)$} & \texttt{inspect}  & 0 & 0 & 20 & 70 & 10 & 0 & 0.90 & 51 \\
& \texttt{dc}  & 0 & 0 & 24 & 58 & 17 & 1 & 0.87 & 27\\
& \texttt{sbs}  & 0 & 0 & 17 & 61 & 17 & 5 & 0.85 & 11\\
& \texttt{scan}  & 0 & 0 & 74 & 26 & 0 & 0 & 0.78 & 15\\
& $\texttt{agg}_2$  & 0 & 0 & 30 & 67 & 2 & 1 & 0.86 & 3\\
& $\texttt{agg}_\infty$  & 0 & 0 & 32 & 58 & 9 & 1 & 0.85 & 15\\


\hline
\multirow{6}{*}{$(0.4,0.8,1.2)$} & \texttt{inspect}  & 0 & 0 & 65 & 31 & 4 & 0 & 0.73 & 44 \\
& \texttt{dc}  & 0 & 0 & 73 & 25 & 2 & 0 & 0.70 & 18\\
& \texttt{sbs}  & 0 & 0 & 65 & 29 & 6 & 0 & 0.68 & 16\\
& \texttt{scan}  & 0 & 2 & 96 & 2 & 0 & 0 & 0.70 & 29\\
& $\texttt{agg}_2$  & 0 & 0 & 83 & 14 & 3 & 0 & 0.71 & 5\\
& $\texttt{agg}_\infty$  & 0 & 0 & 82 & 17 & 1 & 0 & 0.69 & 12\\
\hline
\hline
\multirow{6}{*}{$(0.6,1.2,1.8)$} & \texttt{inspect}  & 0 & 0 & 19 & 71 & 9 & 1 & 0.90 & 55 \\
& \texttt{dc}  & 0 & 0 & 28 & 53 & 17 & 2 & 0.85 & 22\\
& \texttt{sbs}  & 0 & 0 & 18 & 67 & 14 & 1 & 0.85 & 14\\
& \texttt{scan}  & 0 & 0 & 74 & 26 & 0 & 0 & 0.78 & 14\\
& $\texttt{agg}_2$  & 0 & 0 & 23 & 66 & 10 & 1 & 0.87 & 0\\
& $\texttt{agg}_\infty$  & 0 & 0 & 32 & 58 & 9 & 1 & 0.85 & 10\\

\hline
\multirow{6}{*}{$(0.4,0.8,1.2)$} & \texttt{inspect}  & 0 & 0 & 66 & 30 & 4 & 0 & 0.74 & 50 \\
& \texttt{dc}  & 0 & 0 & 75 & 23 & 2 & 0 & 0.70 & 18\\
& \texttt{sbs}  & 0 & 0 & 62 & 30 & 7 & 1 & 0.69 & 11\\
& \texttt{scan}  & 0 & 1 & 98 & 1 & 0 & 0 & 0.70 & 29\\
& $\texttt{agg}_2$  & 0 & 0 & 86 & 12 & 2 & 0 & 0.72 & 5\\
& $\texttt{agg}_\infty$  & 0 & 0 & 82 & 15 & 3 & 0 & 0.70 & 7\\
\hline\hline
\end{tabular}
\caption{Multiple changepoint simulation results.  The top, middle and bottom blocks refer to the complete, half and no overlap settings respectively.  Other simulation parameters: $n = 2000$, $p = 200$, $k = 40$, $z = (500, 1000, 1500)$ and $\sigma^2=1$.}
\label{Tab:MultipleCP}
\end{center}
\end{table}
\begin{figure}[htbp]
\begin{center}
\begin{tabular}{cc}
\includegraphics[width=0.45\textwidth]{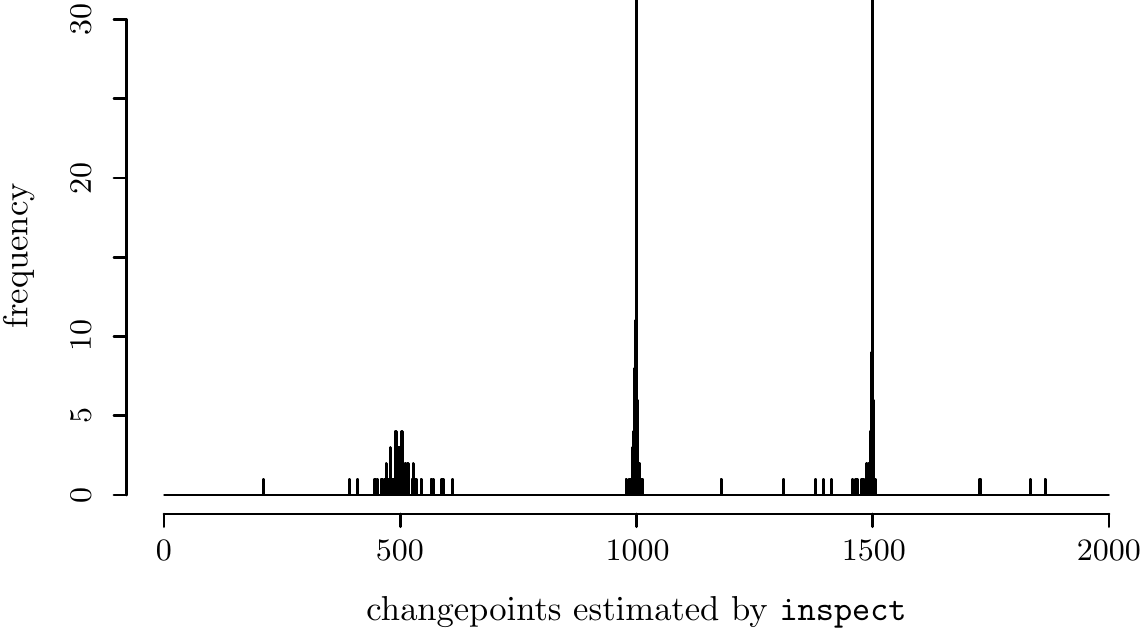} 
&
\includegraphics[width=0.45\textwidth]{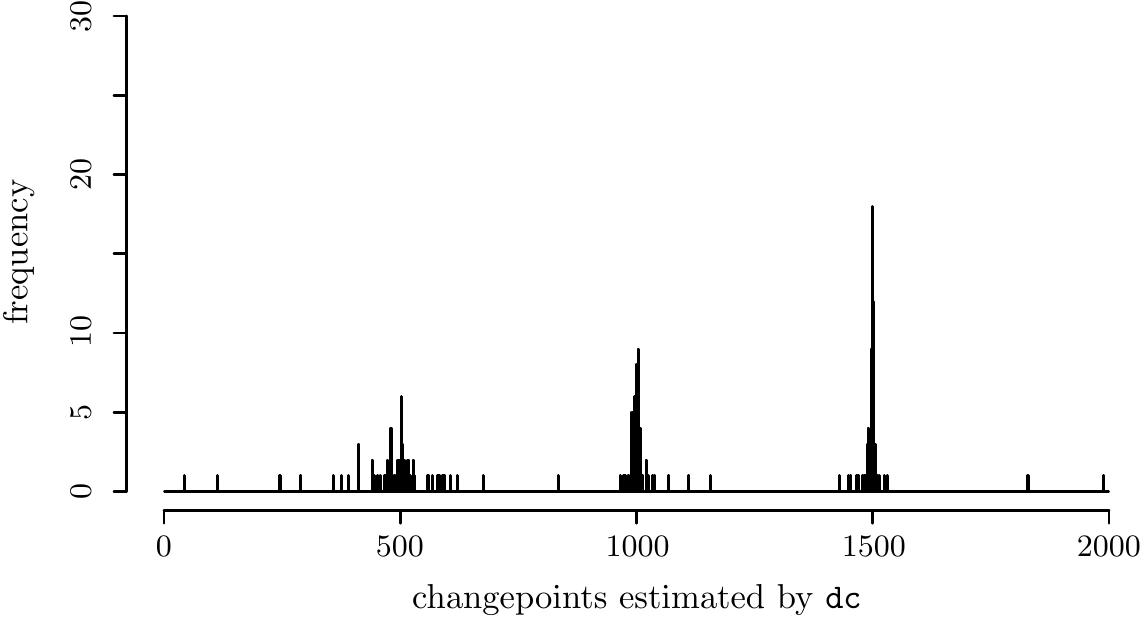}\\
\includegraphics[width=0.45\textwidth]{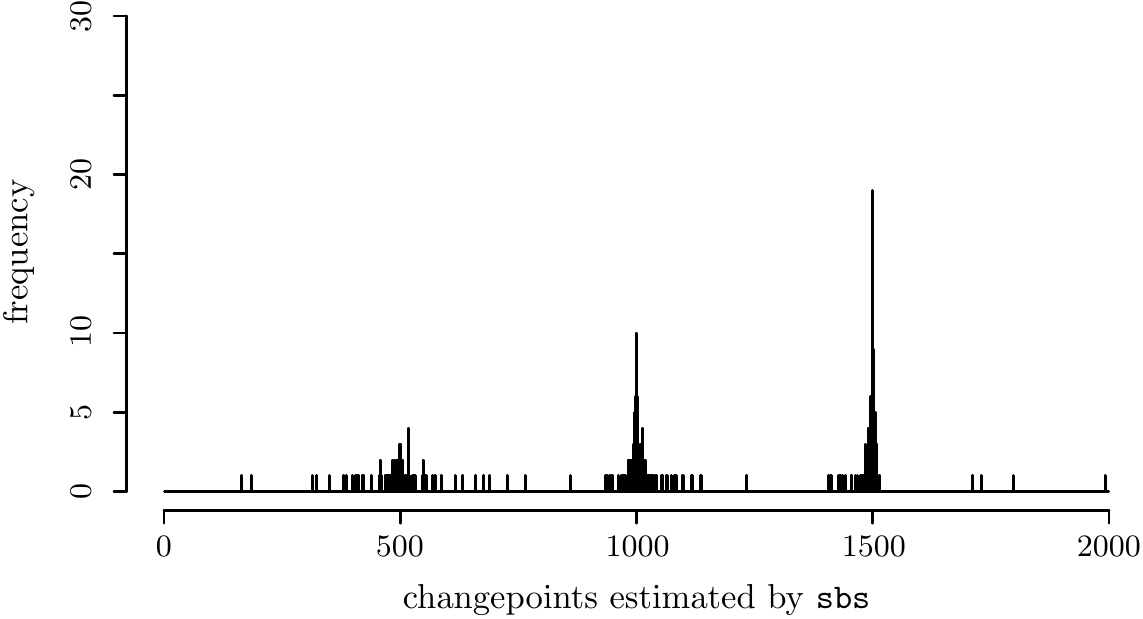} 
&
\includegraphics[width=0.45\textwidth]{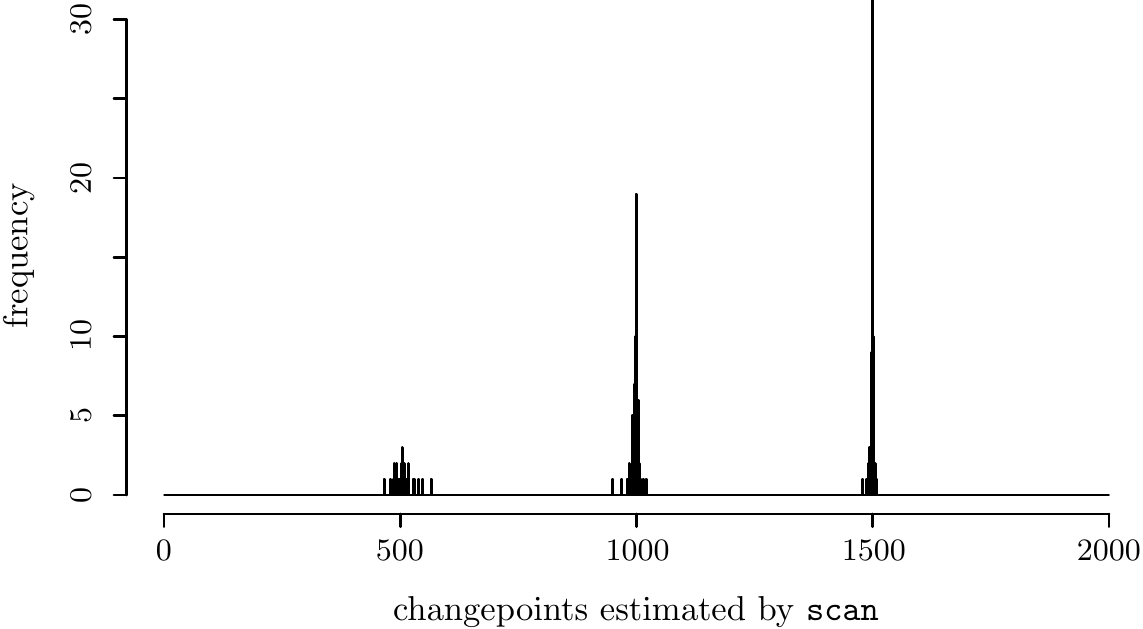}\\
\includegraphics[width=0.45\textwidth]{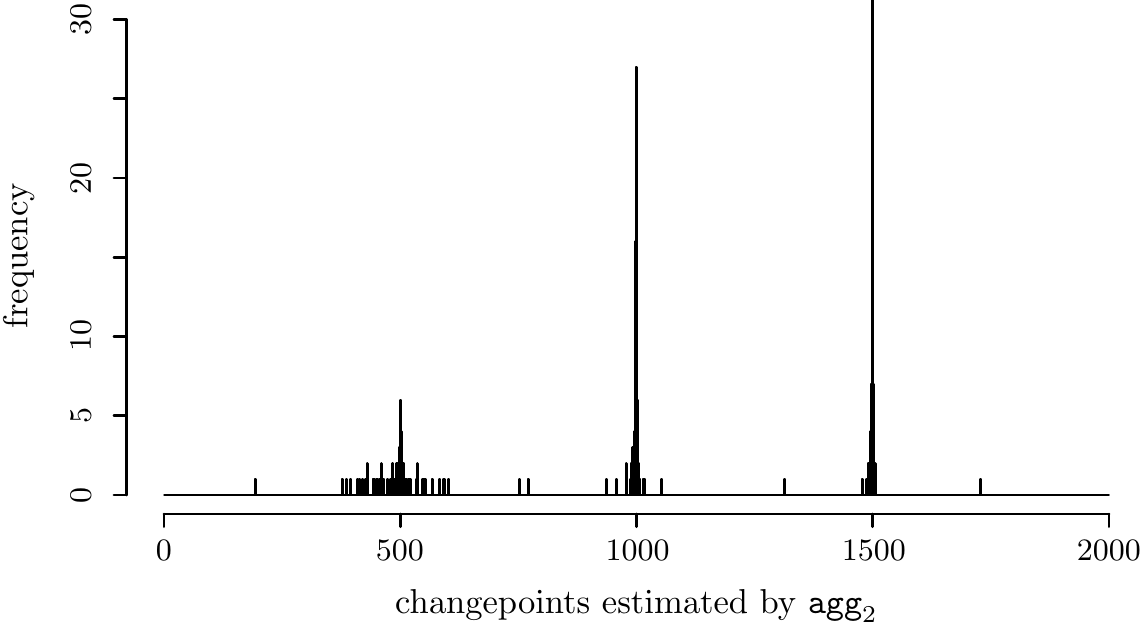} 
&
\includegraphics[width=0.45\textwidth]{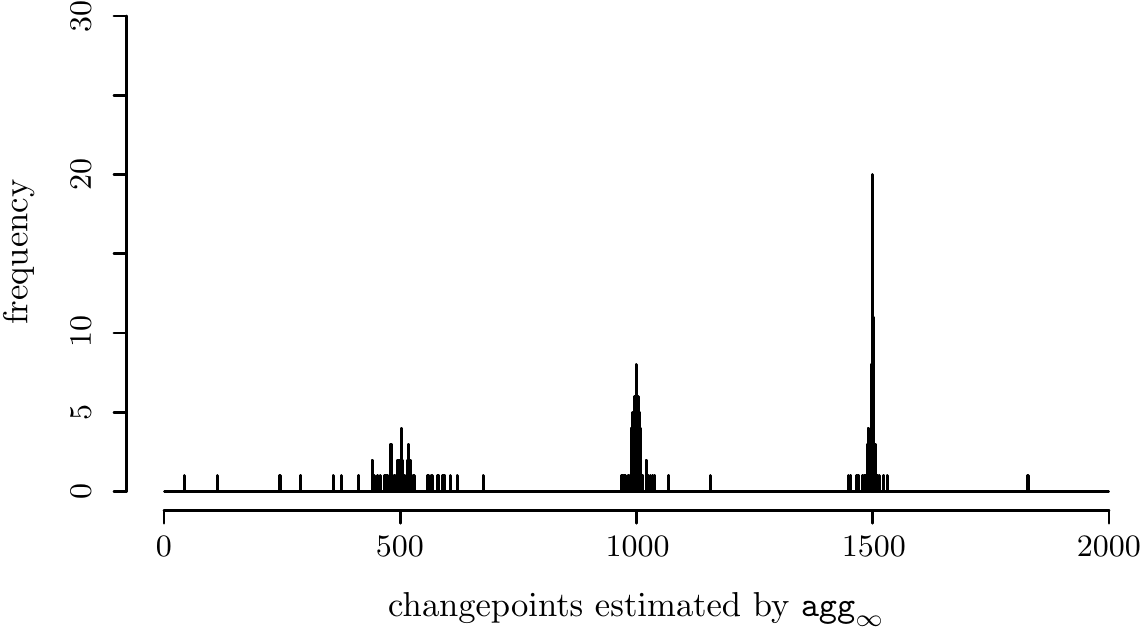}
\end{tabular}
\caption{Histograms of estimated changepoint locations by \texttt{inspect} (top-left), \texttt{dc} (top-right), \texttt{sbs} (bottom-left) and \texttt{scan} (bottom-right) in the half overlap case. Parameter settings: $n = 2000$, $p = 200$, $k = 40$, $z = (500,1000,1500)$, $(\vartheta^{(1)},\vartheta^{(2)},\vartheta^{(3)}) = (0.6,1.2,1.8)$, $\sigma^2 = 1$.}
\label{Fig:MultipleCP}
\end{center}
\end{figure}
\subsection{Real data application}
We study the comparative genomic hybridisation (CGH) microarray dataset from \citet{BleakleyVert2011}, available in the \texttt{ecp} R package \citep{JamesMatteson2015}. CGH is a technique that allows detection of chromosomal copy number abnormality by comparing the fluorescence intensity levels of DNA fragments from a test sample and a reference sample. This dataset contains  (test to reference) log intensity ratio measurements of 43 individuals with bladder tumour at 2215 different loci on their genome. The log intensity ratios for the first ten individuals are plotted in Figure~\ref{Fig:CGH}. While some of the copy number variations are specific to one individual, some copy number abnormality regions (e.g. between loci 2044 and 2143) are shared across several different individuals and are more likely to be disease-related. The \texttt{inspect} algorithm aggregates the changes present in different individuals and estimates the start and end points of copy number changes. Due to the large number of individual-specific copy number changes and the presence of measurement outliers, direct application of \texttt{inspect} with the default threshold level identifies 254 changepoints. However, practitioners can use the associated $\bar{T}_{\max}^{[q_0]}$ score to identify the most significant changes. The 30 most significant identified changepoints are plotted as red dashed lines in Figure~\ref{Fig:CGH}. 
\begin{figure}[htbp]
\begin{center}
\label{Fig:CGH}
 \includegraphics[width=\textwidth]{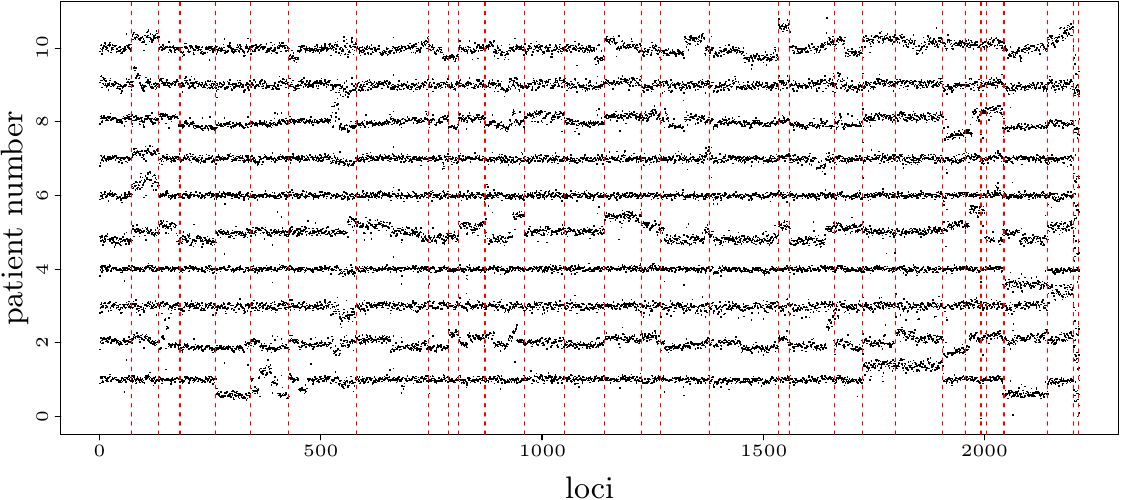}
\end{center}
\caption{Log-intensity ratio measurements of microarray data. Only the first ten patients are shown. Changepoints (red dashed vertical lines) are estimated using all patients in the dataset.}
\end{figure}

\section{Extensions: temporal or spatial dependence}
\label{Sec:Dep}

In this section, we explore how our method and its analysis can be extended to handle more realistic streaming data settings where our data exhibit temporal or spatial dependence. For simplicity, we focus on the single changepoint case, and assume the same mean structure for $\boldsymbol{\mu} = \mathbb{E}(X)$ as described in Section~\ref{Sec:Problem}, in particular~\eqref{Eq:mus},~\eqref{Eq:thetas},~\eqref{Eq:thetasparsity},~\eqref{Eq:tau} and~\eqref{Eq:Vartheta}.

\subsection{Temporal dependence}
A natural way of relaxing the assumption of independence of the columns of our data matrix is to assume that the noise vectors $W_1,\ldots,W_n$ are stationary. Writing $K(u) := \mathrm{cov}(W_t, W_{t+u})$, we assume here that $W = (W_1,\ldots,W_n)$ forms a centred, stationary Gaussian process with covariance function $K$. As we are mainly interested in the temporal dependence in this subsection, we assume each component time series evolves independently, so that $K(u)$ is a diagonal matrix for every $u$.  Further, writing $\sigma^2 := \|K(0)\|_{\mathrm{op}}$, we will assume that the dependence is short-ranged, in the sense that
\begin{equation}
 \label{Eq:ShortRange}
 \biggl\|\sum_{u=0}^{n-1} K(u)\biggr\|_{\mathrm{op}} \leq B\sigma^2
\end{equation}
for some universal constant $B > 0$.  In this case, the oracle projection direction is still $v := \theta/\|\theta\|_2$ and our \texttt{inspect} algorithm does not require any modification.  In terms of its performance in this context, we have the following result:
\begin{thm}
\label{Thm:TemporalDependence}
Suppose $\sigma,B>0$ are known. Let $\hat{z}$ be the output of Algorithm~\ref{Algo:SingleCPVariant} with input $X$ and $\lambda := \sigma\sqrt{8B\log(np)}$.  There exist universal constants $C,C' > 0$ such that if $n\geq 12$ is even, $z$ is even and
\begin{equation}
\label{Eq:T3condprime}
\frac{C\sigma}{\vartheta\tau}\sqrt{\frac{kB\log (np)}{n}}\leq 1,
\end{equation}
then
\[
\mathbb{P}\biggl(\frac{1}{n}|\hat{z} - z| \leq \frac{C'\sigma^2B\log n}{n\vartheta^2}\biggr) \geq 1- \frac{12}{n}.
\]
\end{thm}


\subsection{Spatial dependence}
Now consider the case where we have spatial dependence between the different coordinates of the data stream. More specifically, suppose that the noise vectors satisfy $W_1,\ldots, W_n \stackrel{\mathrm{iid}}{\sim} N_p(0,\Sigma)$, for some positive definite matrix $\Sigma \in \mathbb{R}^{p \times p}$. This turns out to be a more complicated setting, where our initial algorithm requires modification. To see this, observe now that for $a\in\mathbb{S}^{p-1}$, 
\[
 a^\top X_t \sim N(a^\top \mu_t, a^\top \Sigma a).
\]
It follows that the oracle projection direction in this case is
\[
 v_{\mathrm{proj}} := \argmax_{a\in\mathbb{S}^{p-1}} \frac{|a^\top \theta|}{\sqrt{a^\top \Sigma a}} = \Sigma^{-1/2} \argmax_{b\in\mathbb{S}^{p-1}} |b^\top \Sigma^{-1/2} \theta| = \frac{\Sigma^{-1}\theta}{\|\Sigma^{-1}\theta\|_2}.
\]
If $\hat{\Theta}$ is an estimator of the precision matrix $\Theta := \Sigma^{-1}$, and $\hat{v}$ is a leading left singular vector of $\hat{M}$ as computed in Step 3 of Algorithm~\ref{Algo:SingleCP}, then we can estimate the oracle projection direction by $\hat{v}_{\mathrm{proj}} := \hat\Theta\hat{v}/\|\hat\Theta\hat{v}\|_2$.  The sample-splitting version of this algorithm is therefore given in Algorithm~\ref{Algo:SpatDep}.  Lemma~\ref{Lemma:AuBv} in the online supplement allows us to control $\sin \angle(\hat{v}_{\mathrm{proj}},v_{\mathrm{proj}})$ in terms of $\sin \angle(\hat{v},v)$ and $\|\hat{\Theta} - \Theta\|_{\mathrm{op}}$, as well as the extreme eigenvalues of $\Theta$.  Since Proposition~\ref{Prop:SinTheta} does not rely on the independence of the different coordinates, it can still be used to control $\sin \angle(\hat{v},v)$.   In general, controlling $\|\hat{\Theta} - \Theta\|_{\mathrm{op}}$ in high-dimensional cases requires assumptions of additional structure on $\Theta$ (or equivalently, on $\Sigma$).  For convenience of our theoretical analysis, we assume that we have access to observations $W_1',\ldots,W_m' \stackrel{\mathrm{iid}}{\sim} N_p(0,\Sigma)$, independent of $X^{(2)}$, with which we can estimate $\Theta$.  In practice, if a lower bound on $\tau$ were known, we could take $W_1',\ldots,W_m'$ to be scaled, disjoint first-order differences of the observations in $X^{(1)}$ that are within $n_1\tau$ of the endpoints of the data stream; more precisely, we can let $W_t' := 2^{1/2}(X_{2t}^{(1)} - X_{2t-1}^{(1)})$ for $t = 1,\ldots,\lfloor n_1\tau/2 \rfloor$ and $W_{\lfloor n_1\tau/2 \rfloor + t}' := 2^{1/2}(X_{n_1-2t}^{(1)} - X_{n_1-2t+1}^{(1)})$, so that $m = 2 \lfloor n_1\tau/2 \rfloor$.  In fact, Lemmas~\ref{Lemma:LocalCS} and~\ref{Lemma:GlobalCS} in the online supplement indicate that, at least for certain dependence structures, the operator norm error in estimation of $\Theta$ is often negligible by comparison with $\sin \angle(\hat{v},v)$, so a fairly crude lower bound on $\tau$ would often suffice.

\begin{algorithm}[htbp!]
\SetAlgoLined
\IncMargin{1em}
\DontPrintSemicolon
\SetKwInput{One}{Step 1}\SetKwInput{Two}{Step 2}\SetKwInput{Three}{Step 3}\SetKwInput{Four}{Step 4}\SetKwInput{Five}{Step 5}
\KwIn{
$X\in\mathbb{R}^{p\times n}$, $\lambda > 0$.
}
\vskip 0.5ex
\One{Perform the CUSUM transformation $T^{(1)}\leftarrow \mathcal{T}(X^{(1)})$ and $T^{(2)} \leftarrow \mathcal{T}(X^{(2)})$.}
\Two{Use Algorithm~\ref{Algo:ADMM} or~\eqref{Eq:SoftThresholding} (with inputs $T^{(1)}$, $\lambda$ in either case) to solve for $\hat{M}^{(1)} \in \argmax_{M\in\mathcal{S}}\bigl\{\langle T^{(1)}, M\rangle - \lambda\|M\|_1\bigr\}$ with $\mathcal{S}= \{M\in\mathbb{R}^{p\times (n_1-1)}: \|M\|_*\leq 1\}$ or $\{M\in\mathbb{R}^{p\times (n_1-1)}: \|M\|_2\leq 1\}$.}
\Three{Find $\hat{v}^{(1)}\in\argmax_{\tilde{v} \in \mathbb{S}^{p-1}} \|(\hat{M}^{(1)})^\top \tilde{v}\|_2$.}
\Four{Let $\hat{\Theta}^{(1)} = \hat{\Theta}^{(1)}(X^{(1)})$ be an estimator of $\Theta$. Let $\hat{v}^{(1)}_{\mathrm{proj}} \leftarrow \hat\Theta^{(1)}\hat{v}^{(1)}$.}
\Five{Let $\hat{z} \in 2\argmax_{1 \leq t \leq n_1-1} |(\hat{v}^{(1)}_{\mathrm{proj}})^\top T_t^{(2)}|$, and set $\bar{T}_{\max} \leftarrow |(\hat{v}^{(1)}_{\mathrm{proj}})^\top T_{\hat{z}/2}^{(2)}|$.}
\vskip 0.5ex
\KwOut{$\hat{z}, \bar{T}_{\max}$}
\vskip 1ex
\caption{Pseudo-code for a sample-splitting variant of Algorithm~\ref{Algo:SingleCP} for spatially dependent data.}
\label{Algo:SpatDep}
\end{algorithm}

Theoretical guarantees on the performance of the spatially dependent version of the \texttt{inspect} algorithm in illustrative examples of both local and global dependence structures are provided in Theorem~\ref{Thm:SpatDep} in the online supplement.  The main message of these results is that, provided the dependence is not too strong, and we have a reasonable estimate of $\Theta$, we attain the same rate of convergence as when there is no spatial dependence.  However, Theorem~\ref{Thm:SpatDep} also quantifies the way in which this rate of convergence deteriorates as the dependence approaches the boundary of its range.  

In Figure~\ref{Fig:SpatDep}, we compare the performances of the vanilla \texttt{inspect} algorithm and Algorithm~\ref{Algo:SpatDep} on simulated datasets with local and spatial dependence structures. We observe that Algorithm~\ref{Algo:SpatDep} offers improved performance across all values of $\lambda$ considered by accounting for the spatial dependence, as suggested by our theoretical arguments. 
\begin{figure}[hbtp]
\begin{center}
\begin{tabular}{cc}
\includegraphics[width=0.45\textwidth]{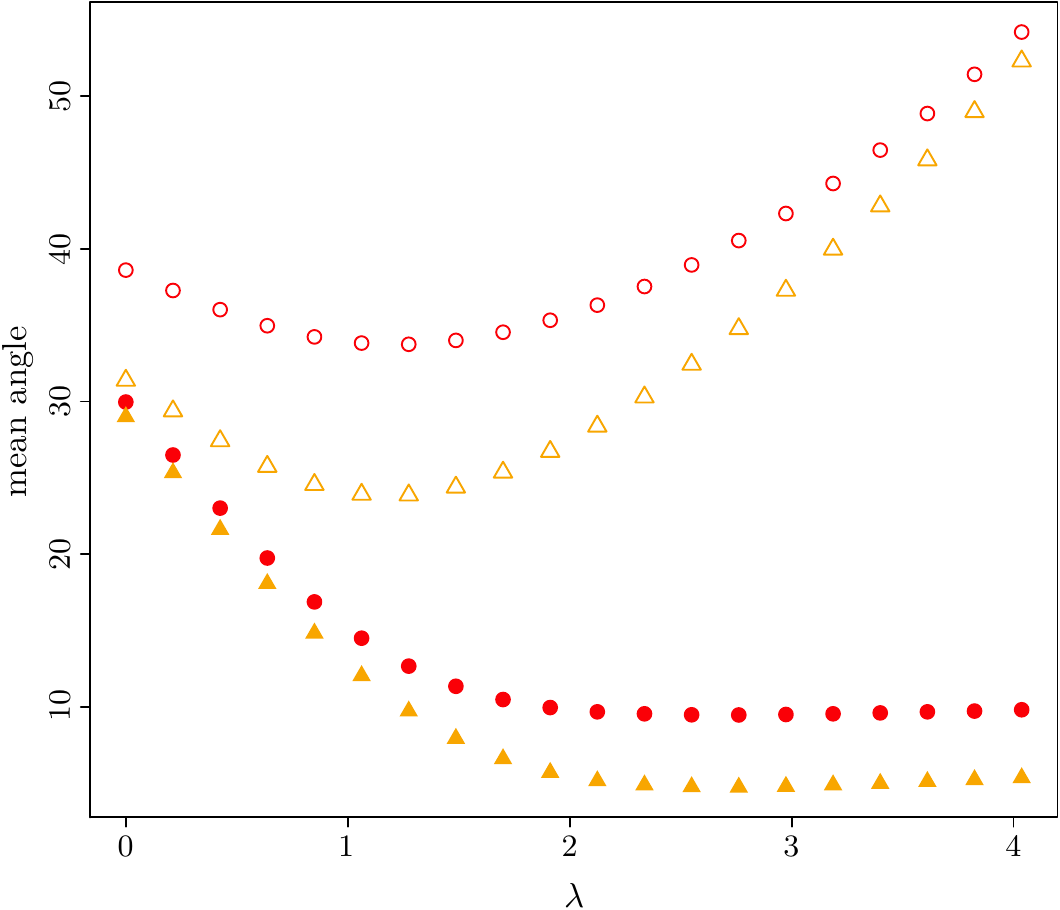} 
&
\includegraphics[width=0.45\textwidth]{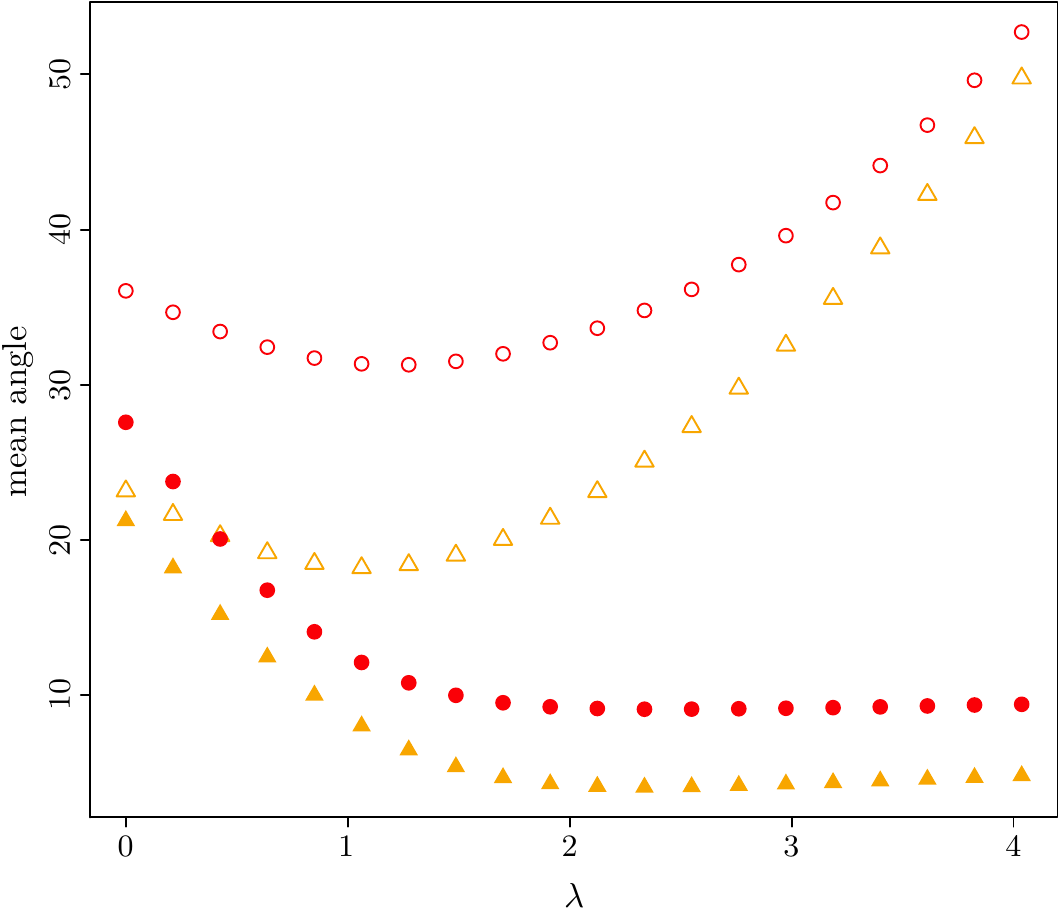}
\end{tabular}
\caption{Mean angle between the estimated projection direction and the optimal projection direction $v_{\mathrm{proj}}$ over 100 experiments. Parameters: $n = 1000$, $p = 500$, $k = 10$ (solid shapes) or $k = 100$ (empty shapes), $z = 400$, $\vartheta = 3$, $\Sigma = (\Sigma_{i,j}) = 2^{-|i-j|}$ (left panel) or $\Sigma = I_p + \mathbf{1}_p\mathbf{1}_p^\top/2$ (right panel). Red circles are from the vanilla \texttt{inspect} algorithm and orange triangles are from Algorithm~\ref{Algo:SpatDep}.}
\label{Fig:SpatDep}
\end{center}
\end{figure}

\section{Proofs of main results}
\label{Appendix:proofs}
\begin{proof}[Proof of Proposition~\ref{Prop:SinTheta}]
We note that the matrix $A$ as defined in Section~\ref{Sec:SingleCP} has rank 1, and its only non-zero singular value is $\|\theta\|_2\|\gamma\|_2$. By Proposition~\ref{Prop:Perturbation} in the online supplement, on the event $\Omega_* := \{\|E\|_\infty \leq \lambda\}$, we have
\[
\sin\angle(\hat{v},v)\leq \frac{8\lambda\sqrt{kn}}{\|\theta\|_2\|\gamma\|_2}.
\]
By definition, $\|\theta\|_2\geq \vartheta$, and by Lemma~\ref{Lemma:gamma} in the online supplement, $\|\gamma\|_2\geq \frac{1}{4}n\tau$. Thus, $\sin\angle(\hat{v},v)\leq \frac{32\lambda \sqrt{k}}{\vartheta\tau\sqrt{n}}$ on $\Omega_*$. It remains to verify that $\mathbb{P}(\Omega_*^{\mathrm{c}})\leq 4(p\log n)^{-1/2}$ for $n \geq 6$.  By Lemma~\ref{Lemma:BrownianBridge} in the online supplement, 
\begin{align}
\mathbb{P}\bigl(\|E\|_\infty\ge 2\sigma\sqrt{\log(p\log n)}\bigr)&\leq 2\sqrt\frac{2}{\pi} p \lceil\log n\rceil\sqrt{\log(p\log n)} \Bigl\{1 + \frac{1}{\log(p \log n)}\Bigr\}(p\log n)^{-2}\nonumber\\
& \leq 6(p\log n)^{-1}\sqrt{\log(p\log n)} \leq 4(p\log n)^{-1/2},\label{Eq:TailProbability}
\end{align}
as desired.
\end{proof} 

\begin{proof}[Proof of Theorem~\ref{Thm:SingleCP}]
Recall the definition of $X^{(2)}$ in~\eqref{Eq:Splitting} and the definition $T^{(2)} := \mathcal{T}(X^{(2)})$.  Define similarly $\boldsymbol{\mu}^{(2)} = (\mu_1^{(2)},\ldots,\mu_{n_1}^{(2)}) \in \mathbb{R}^{p \times n_1}$ and a random matrix $W^{(2)} = (W_1^{(2)},\ldots,W_{n_1}^{(2)})$ taking values in $\mathbb{R}^{p \times n_1}$ by $\mu^{(2)}_t := \mu_{2t}$ and $W^{(2)}_t=W_{2t}$; now let $A^{(2)} := \mathcal{T}(\boldsymbol{\mu}^{(2)})$ and $E^{(2)} := \mathcal{T}(W^{(2)})$. Furthermore, we write $\bar{X} := (\hat{v}^{(1)})^\top X^{(2)}$, $\bar{\mu} := (\hat{v}^{(1)})^\top \boldsymbol{\mu}^{(2)}$, $\bar{W} := (\hat{v}^{(1)})^\top W^{(2)}$, $\bar{T} := (\hat{v}^{(1)})^\top T^{(2)}$, $\bar{A} := (\hat{v}^{(1)})^\top A^{(2)}$ and $\bar{E} := (\hat{v}^{(1)})^\top E^{(2)}$ for the one-dimensional projected images (as row vectors) of the corresponding $p$-dimensional quantities. We note that $\bar{T} = \mathcal{T}(\bar{X})$, $\bar{A} = \mathcal{T}(\bar{\mu})$ and $\bar{E} = \mathcal{T}(\bar{W})$. 

Now, conditional on $\hat{v}^{(1)}$, the random variables $\bar{X}_1,\ldots,\bar{X}_{n_1}$ are independent, with
\[
\bar{X}_t \mid \hat{v}^{(1)} \sim N(\bar{\mu}_t, \sigma^2),
\]
and the row vector $\bar{\mu}$ undergoes a single change at $z^{(2)}:= z/2$ with magnitude of change
\[
\bar\theta:= \bar{\mu}_{z^{(2)}+1} - \bar{\mu}_{z^{(2)}} = (\hat{v}^{(1)})^\top \theta.
\]
Finally, let $\hat{z}^{(2)} \in \argmax_{1\leq t\leq n_1-1} |\bar{T}_t|$, so the first component of the output of the algorithm is $\hat{z} = 2\hat{z}^{(2)}$.  Consider the set
\[
\Upsilon := \{\tilde{v}\in\mathbb{S}^{p-1}: \angle(\tilde{v}, v) \leq \pi/6\}. 
\]
By condition~\eqref{Eq:T3cond} in the statement of the theorem and Proposition~\ref{Prop:SinTheta}, 
\begin{equation}
\label{Eq:T3tmp1}
\mathbb{P}(\hat{v}^{(1)}\in\Upsilon) \geq 1-4(p\log n_1)^{-1/2}. 
\end{equation}
Moreover, for $\hat{v}^{(1)}\in\Upsilon$, we have $\bar\theta\geq \sqrt{3}\vartheta/2$.  Note also that $\hat{v}^{(1)}$ and $W^{(2)}$ are independent, so $\bar{W}$ has independent $N(0,\sigma^2)$ entries.  Define $\lambda_1 := 3\sigma\sqrt{\log\log n_1}$.  By Lemma~\ref{Lemma:BrownianBridge} in the online supplement, and the fact that $n\geq 12$, we have
\begin{equation}
\mathbb{P}\bigl(\|\bar{E}\|_\infty \geq \lambda_1\bigr) \leq \sqrt{2/\pi}\lceil\log n_1\rceil\biggl(3\sqrt{\log\log n_1}+\frac{2}{3\sqrt{\log \log n_1}}\biggr) (\log n_1)^{-9/2} \leq (\log n_1)^{-1}. \label{Eq:T3tmp2}
\end{equation}
Since $\bar{T} = \bar{A} + \bar{E}$, and since $(\bar{A}_t)_t$ and $(\bar{T}_t)_t$ are respectively maximised at $t = z^{(2)}$ and $t = \hat{z}^{(2)}$, we have on the event $\Omega_0 :=\bigl\{\hat{v}^{(1)} \in \Upsilon, \|\bar E\|_\infty \leq \lambda_1\bigr\}$ that
\begin{align*}
\bar{A}_{z^{(2)}} - \bar{A}_{\hat{z}^{(2)}} &= (\bar{A}_{z^{(2)}} - \bar{T}_{z^{(2)}}) + (\bar{T}_{z^{(2)}} - \bar{T}_{\hat{z}^{(2)}}) + (\bar{T}_{\hat{z}^{(2)}} - \bar{A}_{\hat{z}^{(2)}})\\
& \leq |\bar{A}_{z^{(2)}} - \bar{T}_{z^{(2)}}| + |\bar{T}_{\hat{z}^{(2)}} - \bar{A}_{\hat{z}^{(2)}}| \leq 2\lambda_1.
\end{align*}
The row vector $\bar{A}$ has the following explicit form:
\[
\bar{A}_{t} = \begin{cases} \sqrt\frac{t}{n_1(n_1-t)}(n_1-z^{(2)})\bar\theta, & \text{if $t\leq z^{(2)}$}\\ \sqrt\frac{n_1-t}{n_1 t} z^{(2)}\bar\theta, &\text{if $t>z^{(2)}$}.\end{cases}
\]
Hence, by Lemma~\ref{Lemma:Peak} in the online supplement, on the event $\Omega_0$ we have that
\begin{equation}
\label{Eq:zhat2}
 \frac{|\hat{z}^{(2)} - z^{(2)}|}{n_1\tau} \leq \frac{3\sqrt{6}\lambda_1}{\bar\theta(n_1\tau)^{1/2}} = \frac{9\sqrt{6}\sigma}{\bar\theta}\sqrt\frac{\log\log n_1}{n_1\tau}\leq \frac{36\sigma}{\vartheta}\sqrt\frac{\log\log n}{n\tau}. 
\end{equation}
Now define the event
\begin{equation}
\label{Eq:Omega1}
 \Omega_1 := \biggl\{\biggl|\sum_{r=1}^s\bar W_r - \sum_{r=1}^{t}\bar W_r \biggr|\leq \lambda_1 \sqrt{|s-t|}, \quad \forall\; 0\leq t\leq n_1, s \in \{0,z^{(2)},n_1\}\biggr\}.
\end{equation}
From~\eqref{Eq:zhat2} and the condition~\eqref{Eq:T3cond}, provided $C \geq 72$, we have $|\hat{z}^{(2)} - z^{(2)}| \leq n_1\tau/2$.  We can therefore apply Lemmas~\ref{Lemma:NoiseControl} and~\ref{Lemma:Peak} in the online supplement and conclude that on $\Omega_0\cap\Omega_1$, we have
\begin{align*}
 |\bar{E}_{z^{(2)}} - \bar{E}_{\hat{z}^{(2)}}| &\leq 2\sqrt{2}\lambda_1\sqrt\frac{|z^{(2)} - \hat{z}^{(2)}|}{n_1\tau}  + 8\lambda_1 \frac{|z^{(2)} - \hat{z}^{(2)}|}{n_1\tau},\\
 \bar{A}_{z^{(2)}} - \bar{A}_{\hat{z}^{(2)}} & \geq \frac{2\bar\theta}{3\sqrt{6}} |z^{(2)} - \hat{z}^{(2)}| (n_1\tau)^{-1/2}.
\end{align*}
Since $\bar{T}_{z^{(2)}} \leq \bar{T}_{\hat{z}^{(2)}}$, we have that on $\Omega_0\cap\Omega_1$, 
\begin{align*}
 1 \leq \frac{|\bar{E}_{z^{(2)}} - \bar{E}_{\hat{z}^{(2)}}|}{\bar{A}_{z^{(2)}} - \bar{A}_{\hat{z}^{(2)}}} &\leq \frac{6\sqrt{3}\lambda_1}{\bar\theta |z^{(2)} - \hat{z}^{(2)}|^{1/2}} + \frac{12\sqrt{6}\lambda_1}{\bar\theta(n_1\tau)^{1/2}}\\
 & \leq \frac{36\sqrt{2}\sigma}{\vartheta}\sqrt\frac{\log\log n}{|z - \hat{z}|} +  \frac{144\sigma}{\vartheta}\sqrt\frac{\log\log n}{n\tau}.
\end{align*}
We conclude from condition~\eqref{Eq:T3cond} again, that on $\Omega_0\cap\Omega_1$, for $C \geq 288$, we have
\[
 |\hat{z}-z| \leq C'\sigma^2\vartheta^{-2}\log\log n
\]
for some universal constant $C' > 0$. 

It remains to show that $\Omega_0\cap\Omega_1$ has the desired probability.  From~\eqref{Eq:T3tmp1},~\eqref{Eq:T3tmp2}, as well as Lemma~\ref{Lemma:LIL} in the online supplement,
\[
 \mathbb{P}(\Omega_0^{\mathrm{c}} \cup \Omega_1^{\mathrm{c}}) \leq 4(p\log n_1)^{-1/2} + (\log n_1)^{-1}+8(\log n_1)^{-5/4} \leq 4(p\log n_1)^{-1/2} + 9(\log n_1)^{-1}
\]
as desired.
\end{proof}

\begin{proof}[Proof of Theorem~\ref{Thm:MultipleCP}]
For $i \in \{0,1,\ldots,\nu\}$, we define $J_i := \bigl[z_i + \lceil\frac{z_{i+1} - z_i}{3} \rceil, z_{i+1} - \lceil\frac{z_{i+1} - z_i}{3} \rceil]$ and
\[
\Omega_1 := \bigcap_{i=1}^\nu \bigcup_{q=1}^Q \{s_q\in J_{i-1},e_q\in J_i\}.
\]
By a union bound, we have
\begin{align*}
\mathbb{P}(\Omega_1^{\mathrm{c}}) &\leq \nu \biggl(1-\frac{(z_i-z_{i-1}-2\lceil\frac{z_i - z_{i-1}}{3}\rceil)(z_{i+1}-z_i-2\lceil\frac{z_{i+1} - z_i}{3}\rceil)}{n(n+1)/2}\biggr)^Q \\
&\leq \nu\biggl(1 - \frac{(z_i-z_{i-1})(z_{i+1}-z_i)}{9n^2}\biggr)^Q \leq \tau^{-1}(1-\tau^2/9)^Q\leq \tau^{-1}e^{-\tau^2Q/9},
\end{align*}
where the second inequality uses the fact that $n\tau \geq 14$. For any matrix $M \in \mathbb{R}^{p\times n}$ and $1\leq \ell\leq r\leq n$, we write $M^{[\ell,r]}$ for the submatrix obtained by extracting columns $\{\ell,\ell+1,\ldots,r\}$ of $M$. Also define $\boldsymbol{\mu}' := \mathbb{E} X'= \boldsymbol{\mu}$ and $W' := X' - \boldsymbol{\mu}'$. Let $\hat{v}^{[\ell,r]}$ be a leading left singular vector of a maximiser of 
\[
M\mapsto \langle \mathcal{T}(X'^{[\ell,r]}), M\rangle -\lambda\|M\|_1,
\]
for $M\in\mathcal{S}$, where $\mathcal{S} = \mathcal{S}_1$ or $\mathcal{S}_2$. For definiteness, we assume both the maximiser and its leading left singular vector are chosen to be the lexicographically smallest possibilities.  For $q=1,\ldots,Q$, we also write $M^{[q]}$ for $M^{[s_q+1,e_q]}$ and $\hat{v}^{[q]}$ for $\hat{v}^{[s_q+1,e_q]}$. Define events
\begin{align*}
\Omega_2 &:= \bigcap_{1 \leq \ell < r \leq n} \{\|\mathcal{T}(W'^{[\ell,r]})\|_\infty \leq \lambda\}, \\
\Omega_3 &:= \bigcap_{1 \leq \ell < r \leq n} \{\|(\hat{v}^{[\ell,r]})^\top\mathcal{T}(W^{[\ell,r]})\|_\infty \leq \lambda\},\\
\Omega_4 &:= \bigcap_{1 \leq \ell < r \leq n} \bigcap_{0\leq i\leq \nu+1} \bigcap_{0\leq t\leq n}\biggl\{ (\hat{v}^{[\ell,r]})^\top \biggl|\sum_{r=1}^{z_i} W_r - \sum_{r=1}^t W_r\biggr| \leq \lambda |z_i-t|^{1/2}\biggr\}.
\end{align*}
Recall that by definition, $z_0 = 0$ and $z_{\nu+1} = n$.  By Lemma~\ref{Lemma:BrownianBridge} in the online supplement, 
\[
\mathbb{P}(\Omega_2^{\mathrm{c}}) \leq \binom{n}{2}\sqrt\frac{2}{\pi} p\lceil \log n\rceil \biggl(4\sqrt{\log(np)} + \frac{1}{2\sqrt{\log(np)}}\biggr)(np)^{-8} \leq n^{-5}p^{-6}.
\]
Also, since $\hat{v}^{[\ell,r]}$ and $X$ are independent, $(\hat{v}^{[\ell,r]})^\top \mathcal{T}(W)$ has the same distribution as $\mathcal{T}(G)$, where $G$ is a row vector of length $r-\ell+1$ with independent $N(0,\sigma^2)$ entries. So by Lemma~\ref{Lemma:BrownianBridge} again, 
\[
\mathbb{P}(\Omega_3^{\mathrm{c}}) \leq \binom{n}{2}\mathbb{P}\bigl\{\|\mathcal{T}(G)\|_\infty > \lambda\bigr\} \leq n^{-5}p^{-6}.
\]
Moreover, by Lemma~\ref{Lemma:LIL} in the online supplement, we have that
\[
 \mathbb{P}(\Omega_4^{\mathrm{c}}) \leq (2\nu+2) \binom{n}{2} 2(np)^{-4}\log n \leq 2n^{-1} p^{-4}\log n.
\]
We claim that the desired event $\Omega^* := \{\text{$\hat{\nu} = \nu$ and $|\hat{z}_i-z_i|\leq n\rho$ for all $1\leq i\leq \nu$}\}$ occurs if the following two statements hold every time the function \texttt{wbs} is called in Algorithm~\ref{Algo:MultipleCP}$'$:
\begin{enumerate}[label={(\roman*)}, noitemsep]
\item There exist unique $i_1, i_2 \in \{0,1,\ldots,\nu+1\}$ such that $|s - z_{i_1}|\leq n\rho$ and $|e - z_{i_2}|\leq n\rho$, where $(s,e)$ is the pair of arguments of the \texttt{wbs} function call.
\item $\bar{T}_{\max}^{[q_0]} > \xi$ if and only if $i_2 - i_1 \geq 2$, where $i_1$ and $i_2$ are the indices defined in~(i).
\end{enumerate}
To see this, observe that the set of all arguments used in the calls of the function \texttt{wbs} is $\hat{Z}\cup\{0,n\}$, so~(i) ensures that
\[
\max_{\hat{z} \in \hat{Z} \cup \{0,n\}} \min_{i \in \{0,1,\ldots,\nu+1\}} |\hat{z} - z_i| \leq n\rho.
\]
If $|\hat{z} - z_i| \leq n\rho$, we say $\hat{z}$ is `identified' to $z_i$.  Moreover, each candidate changepoint $b$ identified by the function call \texttt{wbs}$(s,e)$ in Algorithm~\ref{Algo:MultipleCP}$'$ satisfies $\min\{b-s, e-b\} \geq n\beta > 2n\rho$.  It follows that different elements of $\hat{Z} \cup \{0,n\}$ cannot be identified to the same $z_i$, so no element of $\hat{Z}$ is identified to $z_0$ or $z_{\nu+1}$, and the second part of the event $\Omega^*$ holds.  It remains to show that each element of $\{z_1,\ldots,z_\nu\}$ is identified by some element of $\hat{Z}$.  To see this, note that if $z_i$ is not identified, we can let $(s^*,e^*)$ be the shortest interval such that $s^* +1 \leq z_i \leq e^*$ and such that $(s^*,e^*)$ are a pair of arguments called by the \texttt{wbs} function in Algorithm~\ref{Algo:MultipleCP}$'$.  By~(i), the two endpoints $s^*$ and $e^*$ are identified to $z_{i_1}$ and $z_{i_2}$ respectively, say, for some $i_1 \leq i-1$ and $i_2 \geq i+1$. But then by~(ii) a new point $b$ will be added to $\hat{Z}$ and the recursion continues on the pairs $(s^*,b)$ and $(b,e^*)$, contradicting the minimality of the pair $(s^*,e^*)$. 

We now prove by induction on the depth of the recursion that on $\Omega_1\cap\Omega_2\cap\Omega_3 \cap\Omega_4$, statements~(i) and~(ii) hold every time \texttt{wbs} is called in Algorithm~\ref{Algo:MultipleCP}$'$.  The first time \texttt{wbs} is called, $s = 0$ and $e = n$, so (i) is satisfied with the unique choice $i_1=0$ and $i_2=\nu+1$. This proves the base case. Now suppose \texttt{wbs} is called with the pair $(s,e)$ satisfying (i), yielding indices $i_1,i_2 \in \{0,1,\ldots,\nu+1\}$ with $|s-z_{i_1}| \leq n\rho$, $|e -z_{i_2}| \leq n\rho$. To complete the inductive step, we need to show that~(ii) also holds, and if a new changepoint $b$ is detected, then~(i) holds for the pairs of arguments $(s,b)$ and $(b,e)$. We have two cases.

\emph{Case 1}: $i_2 - i_1 = 1$. In this case, $(s+n\beta , e -  n\beta]$ contains no changepoint. Since $\xi = \lambda$, on $\Omega_3$ we always have 
\[
\bar T_{\max}^{[q_0]} = \max_{q\in\mathcal{Q}_{s,e}} \|(\hat{v}^{[q]})^\top \mathcal{T}(X^{[q]})\|_\infty \leq \xi,
\]
so (ii) is satisfied with no additional changepoint detected. 

\emph{Case 2}: $i_2 - i_1 \geq 2$. On the event $\Omega_1$, for any $i^*\in\{i_1+1,\ldots,i_2-1\}$, there exists $q^* \in\{1,\ldots,Q\}$ such that $s_{q^*} \in J_{i^*-1}$ and $e_{q^*} \in J_{i^*}$. Moreover, since $\min\{s_{q^*} - s, e - e_{q^*}\} \geq \lceil n\tau/3\rceil - n\rho > n\beta$ provided $C \geq 9$ in the condition on $\beta$ in the theorem, we have $q^* \in\mathcal{Q}_{s,e}$. Since there is precisely one changepoint within the segment $(s_{q^*},e_{q^*}]$, the matrix $\mathcal{T}(\boldsymbol{\mu}'^{[q^*]})$ has rank 1; cf.~\eqref{Eq:A}. On $\Omega_2$, we have $\|\mathcal{T}(W'^{[q^*]})\|_\infty \leq \lambda$. Thus, by Proposition~\ref{Prop:Perturbation} and Lemma~\ref{Lemma:gamma} in the online supplement, 
\[
\sin \angle \bigl(\hat{v}^{[q^*]}, \theta^{(i^*)}/\|\theta^{(i^*)}\|_2\bigr) \leq \frac{8\lambda \sqrt{k(e_{q^*}-s_{q^*})}}{\|\theta^{(i^*)}\|_2 n\tau/12}\leq \frac{96\lambda\sqrt{k}}{\vartheta\tau\sqrt{n}} = 96(\rho k\tau^{2})^{1/2} \leq 96 C^{-1/2}
\]
under the conditions of the theorem. Therefore, recalling the definition of $q_0$ in Algorithm~\ref{Algo:MultipleCP}$'$, and on the event $\Omega_2 \cap \Omega_3$, 
\begin{align}
\bar T_{\max}^{[q_0]} &\geq \bar T_{\max}^{[q^*]} = \|(\hat{v}^{[q^*]})^\top \mathcal{T}(X^{[q^*]})\|_\infty \geq \|(\hat{v}^{[q^*]})^\top \mathcal{T}(\boldsymbol{\mu}^{[q^*]})\|_\infty - \|(\hat{v}^{[q^*]})^\top \mathcal{T}(W^{[q^*]})\|_\infty\nonumber\\
&\geq \bigl|(\hat{v}^{[q^*]})^\top \theta^{(i^*)}\bigr| \sqrt{\frac{(z_{i^*} - s_{q^*})(e_{q^*} - z_{i^*})}{e_{q^*}-s_{q^*}}} - \lambda \nonumber \\
&\geq \sqrt{1-96^2/C}\|\theta^{(i^*)}\|_2  \sqrt{\frac{n\tau}{6}} - \lambda > 0.5\sqrt{n\tau} \|\theta^{(i^*)}\|_2 - \lambda,
\label{Eq:Tmax}
\end{align}
for sufficiently large $C > 0$.  In particular, by the condition, $C\rho k\tau^3\leq 1$, we have for sufficiently large $C >0$ that
\[
\bar T_{\max}^{[q_0]} \geq 0.5\vartheta\sqrt{n\tau} - \lambda = \lambda(0.5\rho^{-1/2}\tau^{-3/2}-1) > \lambda = \xi.
\]
Thus~(ii) is satisfied with a new changepoint $b:=s_{q_0}+\hat{z}^{[q_0]}$ detected. It remains to check that~(i) holds for the pairs of arguments $(s,b)$ and $(b,e)$, for which it suffices to show that $\min_{1\leq i\leq \nu} |b - z_i| \leq n\rho$.  To this end, we study the behaviour of univariate CUSUM statistics of the projected series $(\hat{v}^{[q_0]})^\top X^{[q_0]}$. To simplify notation, we define $\bar{X} := (\hat{v}^{[q_0]})^\top X^{[q_0]}$, $\bar{\mu} := (\hat{v}^{[q_0]})^\top \boldsymbol{\mu}^{[q_0]}$, $\bar{W} := (\hat{v}^{[q_0]})^\top W^{[q_0]}$, $\bar{T} := \mathcal{T}(\bar{X})$, $\bar{A} := \mathcal{T}(\bar{\mu})$ and $\bar{E} := \mathcal{T}(\bar{W})$. The row vector $\bar{\mu} \in \mathbb{R}^{e_{q_0} - s_{q_0}}$ is piecewise constant with changepoints at $z_{i_1+1}-s_{q_0},\ldots,z_{i_2-1}-s_{q_0}$. Recall that $\hat{z}^{[q_0]} \in \argmax_{1 \leq t \leq e_{q_0}-s_{q_0}-1} |\bar{T}_t|$. We may assume that $\bar{T}_{\hat{z}^{[q_0]}} > 0$ (the case $\bar{T}_{\hat{z}^{[q_0]}} < 0$ can be handled similarly). On $\Omega_3$, $\bar{A}_{\hat{z}^{[q_0]}} \geq \bar{T}_{\hat{z}^{[q_0]}} - \lambda = \bar T_{\max}^{[q_0]} - \lambda > 0$, and hence there is at least one changepoint in $(s_{q_0}, e_{q_0}]$. We may assume that $\hat{z}^{[q_0]}$ is not equal to $z_i-s_{q_0}$ for any $i_1+1\leq i\leq i_2-1$, since otherwise $\min_{1\leq i\leq \nu}|b-z_i| = 0$ and we are done. By Lemma~\ref{Lemma:CusumShape} in the online supplement and after possibly reflecting the time direction, we may also assume that there is at least one changepoint to the left of $\hat{z}^{[q_0]}$, and that if $z_{i_0}-s_{q_0}$ is the changepoint immediately left of $\hat{z}^{[q_0]}$, then the series $\{\bar A_t: z_{i_0}-s_{q_0}\leq t\leq \hat{z}^{[q_0]}\}$ is positive and strictly decreasing. By~\eqref{Eq:Tmax} with $i_0$ in place of $i^*$, we have that on $\Omega_3$, 
\begin{align}
\bar{A}_{z_{i_0}-s_{q_0}} \geq \bar{A}_{\hat{z}^{[q_0]}} &\geq \bar T_{\max}^{[q_0]} - \lambda \geq 0.5\sqrt{n\tau} \|\theta^{(i_0)}\|_2 - 2\lambda\label{Eq:Abarbound1} \\
&\geq \lambda (0.5\rho^{-1/2}\tau^{-3/2}-2)\geq 0.4\lambda \rho^{-1/2}\tau^{-3/2}
\label{Eq:Abarbound2}
\end{align}
for sufficiently large $C > 0$.  
%
Our strategy here is to characterise the magnitude of $\bar{E}_{z_{i_0}-s_{q_0}} - \bar{E}_{\hat{z}^{[q_0]}}$ and the rate of decay of the series $\{\bar A_t: z_{i_0}-s_{q_0}\leq t\leq \hat{z}^{[q_0]}\}$ from its left endpoint, so that we can conclude from $\bar{A}_{\hat{z}^{[q_0]}} + \bar{E}_{\hat{z}^{[q_0]}}\geq \bar{A}_{z_{i_0}-s_{q_0}} + \bar{E}_{z_{i_0}-s_{q_0}}$ that $\hat{z}^{[q_0]}$ is close to $z_{i_0}-s_{q_0}$. This is achieved by considering the following three cases: (a) there is no changepoint to the right of $\hat{z}^{[q_0]}$, i.e.\ $z_{i_0+1} \geq e_{q_0}$; (b) $z_{i_0+1} \leq e_{q_0}-1$ and $\bar A_{z_{i_0}-s_{q_0}} \geq \bar A_{z_{i_0+1}-s_{q_0}}$; (c) $z_{i_0+1} \leq e_{q_0}-1$ and $\bar A_{z_{i_0}-s_{q_0}} < \bar A_{z_{i_0+1}-s_{q_0}}$.

In case (a), define $\tilde\vartheta := (z_{i_0}-s_{q_0})^{-1}\sum_{t=1}^{z_{i_0}-s_{q_0}}\bar\mu_t - \bar\mu_{z_{i_0}+1-s_{q_0}}$, so that
\begin{equation}
 \bar A_{z_{i_0}-s_{q_0}} = \tilde\vartheta\sqrt\frac{(z_{i_0}-s_{q_0})(e_{q_0}-z_{i_0})}{e_{q_0}-s_{q_0}} \leq \tilde\vartheta \sqrt{\min(z_{i_0}-s_{q_0}, e_{q_0}-z_{i_0})}.
 \label{Eq:Case_a}
\end{equation}
Comparing~\eqref{Eq:Case_a} with~\eqref{Eq:Abarbound2}, we have that $\tilde\vartheta \sqrt{\min(z_{i_0}-s_{q_0}, e_{q_0}-z_{i_0})} \geq 0.4\lambda \rho^{-1/2}\tau^{-3/2}$.  We apply Lemma~\ref{Lemma:Peak} in the online supplement with $e_{q_0}-s_{q_0}$ and $z_{i_0}-s_{q_0}$ taking the roles of $n$ and $z$ in the lemma respectively. On the event $\Omega_3$, we have that
\[
 \frac{\hat{z}^{[q_0]} - (z_{i_0} - s_{q_0})}{\min(z_{i_0}-s_{q_0}, e_{q_0}-z_{i_0})} \leq \frac{3\sqrt{6}}{2}\frac{\bar{A}_{z_{i_0}-s_{q_0}} - \bar{A}_{\hat{z}^{[q_0]}}}{\tilde\vartheta\sqrt{\min(z_{i_0}-s_{q_0}, e_{q_0}-z_{i_0})}} \leq 20\rho^{1/2}\tau^{3/2} < \frac{1}{2}
\]
for sufficiently large $C > 0$. Hence, by Lemma~\ref{Lemma:NoiseControl} and Lemma~\ref{Lemma:Peak} in the online supplement, on the event $\Omega_3\cap\Omega_4$, we have
\begin{align*}
 \bigl|\bar{E}_{z_{i_0}-s_{q_0}} - \bar{E}_{\hat{z}^{[q_0]}}\bigr| & \leq 2\sqrt{2}\lambda\sqrt\frac{\hat{z}^{[q_0]} - (z_{i_0} - s_{q_0})}{\min(z_{i_0}-s_{q_0}, e_{q_0}-z_{i_0})} + 8\lambda\frac{\hat{z}^{[q_0]} - (z_{i_0} - s_{q_0})}{\min(z_{i_0}-s_{q_0}, e_{q_0}-z_{i_0})},\\
 \bar{A}_{z_{i_0}-s_{q_0}} - \bar{A}_{\hat{z}^{[q_0]}} & \geq \frac{2\tilde\vartheta}{3\sqrt{6}}\frac{\hat{z}^{[q_0]} - (z_{i_0} - s_{q_0})}{\sqrt{\min(z_{i_0}-s_{q_0}, e_{q_0}-z_{i_0})}}.
\end{align*}
Since $\bar{T}_{z_{i_0}-s_{q_0}} \leq \bar{T}_{\hat{z}^{[q_0]}}$, we must have
\begin{align*}
 1\leq \frac{\bigl|\bar{E}_{z_{i_0}-s_{q_0}} - \bar{E}_{\hat{z}^{[q_0]}}\bigr|}{\bar{A}_{z_{i_0}-s_{q_0}} - \bar{A}_{\hat{z}^{[q_0]}}} &\leq \frac{6\sqrt{3}\lambda}{\tilde\vartheta \sqrt{\hat{z}^{[q_0]} - (z_{i_0} - s_{q_0})}} + \frac{12\sqrt{6}\lambda}{\tilde\vartheta \sqrt{\min(z_{i_0}-s_{q_0}, e_{q_0}-z_{i_0})}}\\
 &\leq \frac{6\sqrt{3}\lambda}{\tilde\vartheta \sqrt{\hat{z}^{[q_0]} - (z_{i_0} - s_{q_0})}} + 30\sqrt{6}\rho^{1/2}\tau^{3/2}.
\end{align*}
Thus, using the condition that $C\rho k\tau^3\leq 1$ again, we have that for sufficiently large $C$, 
\[
 \hat{z}^{[q_0]} - (z_{i_0} - s_{q_0}) \leq C''\lambda^2\tilde\vartheta^{-2} \leq C'\rho\tau^3\min(z_{i_0}-s_{q_0}, e_{q_0}-z_{i_0}) \leq C'n\rho
\]
for some universal constants $C''$ and $C'$.

For case (b), we define $\tilde{\mu} := \frac{1}{e_{q_0}-s_{q_0}}\sum_{t=1}^{e_{q_0}-s_{q_0}}\bar\mu_t$ to be the overall average of the $\bar\mu$ series, and let
\[
\tilde{\mu}_{\mathrm{L}} := \frac{1}{z_{i_0}-s_{q_0}}\sum_{t=1}^{z_{i_0}-s_{q_0}}\bar\mu_t - \tilde\mu, \quad \tilde{\mu}_{\mathrm{M}} := \bar\mu_{z_{i_0}+1-s_{q_0}}- \tilde\mu \ \ \text{and} \ \ \tilde{\mu}_{\mathrm{R}} := \frac{1}{e_{q_0}-z_{i_0+1}}\sum_{t=z_{i_0+1} - s_{q_0}+1}^{e_{q_0}-s_{q_0}}\bar\mu_t-\tilde\mu
\]
be the centred averages of the $\bar\mu$ series on the segments $(0,z_{i_0}-s_{q_0}]$, $(z_{i_0}-s_{q_0},z_{i_0+1} - s_{q_0}]$ and $(z_{i_0+1} - s_{q_0}, e_{q_0}-s_{q_0}]$ respectively. Using~\eqref{Eq:CUSUM}, we have that for $z_{i_0}-s_{q_0} \leq t \leq z_{i_0+1}-s_{q_0}$,
\begin{equation}
\label{Eq:Abar}
\bar{A}_t = [\mathcal{T}(\bar{\mu})]_t = \sqrt\frac{e_{q_0}-s_{q_0}}{t(e_{q_0}-s_{q_0}-t)}\Bigl\{(z_{i_0}-s_{q_0})(-\tilde\mu_{\mathrm{L}}) + (t-z_{i_0}+s_{q_0})(-\tilde\mu_{\mathrm{M}})\bigr\}\Bigr\}.
\end{equation}
We claim that $z_{i_0}-s_{q_0}\geq n\tau/15$.  For, if not, then in particular, $z_{i_0-1} < s_{q_0}$ and $\tilde\mu_{\mathrm{L}} = \bar\mu_{z_{i_0}-s_{q_0}} - \tilde\mu$. Hence $\tilde\mu_{\mathrm{M}} - \tilde\mu_{\mathrm{L}} = (\hat{v}^{[q_0]})^\top\theta^{(i_0)} \leq \|\theta^{(i_0)}\|_2$. By~\eqref{Eq:Abar} and the fact that $\bar{A}_{z_{i_0}-s_{q_0}} > 0$, we have $\tilde{\mu}_{\mathrm{L}} < 0$. On the other hand, a similar argument as in~\eqref{Eq:Abarbound2} shows that
\[
 \frac{\sqrt{n\tau}\|\theta^{(i_0)}\|_2}{\lambda}\geq \rho^{-1/2}\tau^{-3/2} \geq C^{1/2}.
\]
Thus, it follows from~\eqref{Eq:Abarbound1} that for sufficiently large $C > 0$,
\begin{align*}
0.4\sqrt{n\tau}(\tilde\mu_{\mathrm{M}} - \tilde{\mu}_{\mathrm{L}}) &\leq 0.4\sqrt{n\tau}\|\theta^{(i_0)}\|_2 \leq 0.5\sqrt{n\tau}\|\theta^{(i_0)}\|_2-2\lambda \leq \bar{A}_{z_{i_0}-s_{q_0}} \\
&= \sqrt{\frac{(e_{q_0}-s_{q_0})(z_{i_0}-s_{q_0})}{e_{q_0}-z_{i_0}}}(-\tilde\mu_{\mathrm{L}})\\
&\leq \sqrt\frac{n\tau+z_{i_0}-s_{q_0}}{n\tau}\sqrt{z_{i_0}-s_{q_0}}(-\tilde\mu_{\mathrm{L}}) \leq \frac{4\sqrt{n\tau}}{15}(-\tilde{\mu}_{\mathrm{L}}),
\end{align*}
which can be rearranged to give $-\tilde{\mu}_{\mathrm{M}} \geq (-\tilde{\mu}_{\mathrm{L}})/3$. Consequently,
\begin{align*}
\bar{A}_{z_{i_0+1}-s_{q_0}} & = \sqrt\frac{e_{q_0}-s_{q_0}}{(z_{i_0+1}-s_{q_0})(e_{q_0}-z_{i_0+1})}\Bigl\{(-\tilde\mu_{\mathrm{L}})(z_{i_0}-s_{q_0}) + (-\tilde\mu_{\mathrm{M}})(z_{i_0+1}-z_{i_0}) \Bigr\}\\
& >  \sqrt\frac{e_{q_0}-s_{q_0}}{(z_{i_0+1}-s_{q_0})(e_{q_0}-z_{i_0+1})}\Bigl\{(-\tilde\mu_{\mathrm{L}})(z_{i_0}-s_{q_0}) + (-\tilde\mu_{\mathrm{L}})(z_{i_0+1}-z_{i_0})/3 \Bigr\}\\
&\geq \sqrt\frac{(e_{q_0}-s_{q_0})(z_{i_0+1}-s_{q_0})}{e_{q_0}-z_{i_0+1}}(-\tilde\mu_{\mathrm{L}})/3\\
&\geq \frac{\bar{A}_{z_{i_0}-s_{q_0}}}{3}\sqrt\frac{z_{i_0+1}-s_{q_0}}{z_{i_0}-s_{q_0}}>\bar{A}_{z_{i_0}-s_{q_0}},
\end{align*}
contradicting the assumption of case (b). Hence we have established the claim. We can then apply Lemma~\ref{Lemma:TwoCP} in the online supplement, with $\bar{A}_t$, $e_{q_0}-s_{q_0}$, $z_{i_0} - s_{q_0}$, $z_{i_0+1} - s_{q_0}$, $-\tilde\mu_{\mathrm{L}}$, $-\tilde\mu_{\mathrm{M}}$ and $\tau/15$ taking the roles of $g(t)$, $n$, $z$, $z'$, $\mu_0$, $\mu_1$ and $\tau$ in the lemma respectively, to obtain on the event $\Omega_3$ that
\[
\hat{z}^{[q_0]}-(z_{i_0}-s_{q_0}) \leq \frac{2\lambda}{0.5\bar{A}_{z_{i_0}-s_{q_0}} n^{-1}\tau/15} \leq 150n \tau^{1/2}\rho^{1/2} \leq 150C^{-1/2} n\tau,
\]
where we have used~\eqref{Eq:Abarbound2} in the penultimate inequality and the condition $\rho\leq \tau/C$ in the final inequality. For sufficiently large $C$, we therefore have $\hat{z}^{[q_0]}-(z_{i_0}-s_{q_0}) \leq n\tau/30$. Thus, we can apply Lemma~\ref{Lemma:NoiseControl} and Lemma~\ref{Lemma:TwoCP} in the online supplement to obtain on $\Omega_4$ that
\begin{align*}
 \bigl|\bar{E}_{z_{i_0}-s_{q_0}} - \bar{E}_{\hat{z}^{[q_0]}}\bigr| & \leq 2\sqrt{2}\lambda\sqrt\frac{\hat{z}^{[q_0]} - (z_{i_0} - s_{q_0})}{n\tau/15} + 8\lambda\frac{\hat{z}^{[q_0]} - (z_{i_0} - s_{q_0})}{n\tau/15},\\
 \bar{A}_{z_{i_0}-s_{q_0}} - \bar{A}_{\hat{z}^{[q_0]}} & \geq \frac{0.5 \bar{A}_{z_{i_0}-s_{q_0}} n^{-1}\tau}{15}\{\hat{z}^{[q_0]} - (z_{i_0} - s_{q_0})\}  \geq  \frac{\lambda}{75n\tau^{1/2}\rho^{1/2}}\{\hat{z}^{[q_0]} - (z_{i_0} - s_{q_0})\},
\end{align*}
where we have used~\eqref{Eq:Abarbound2} in the final inequality. Since $\bar{T}_{z_{i_0}-s_{q_0}} \leq \bar{T}_{\hat{z}^{[q_0]}}$, we must have on $\Omega_4$ that
\[
 1\leq \frac{\bigl|\bar{E}_{z_{i_0}-s_{q_0}} - \bar{E}_{\hat{z}^{[q_0]}}\bigr|}{\bar{A}_{z_{i_0}-s_{q_0}} - \bar{A}_{\hat{z}^{[q_0]}}} \leq \frac{C''n^{1/2}\rho^{1/2}}{ \sqrt{\hat{z}^{[q_0]} - (z_{i_0} - s_{q_0})}} + C''C^{-1/2},
\]
for some universal constant $C'' > 0$.  Hence, for sufficiently large $C>0$, we have that
\[
\hat{z}^{[q_0]} - (z_{i_0} - s_{q_0}) \leq C'n\rho
\]
for some universal constant $C'>0$.

For case (c), by Lemma~\ref{Lemma:CusumShape} in the online supplement, the series $(\bar{A}_t: z_{i_0}-s_{q_0}\leq t\leq z_{i_0+1}-s_{q_0})$ must be strictly decreasing, then strictly increasing, while staying positive throughout. Define $\zeta := \max\{t \in [z_{i_0}-s_{q_0},z_{i_0+1}-s_{q_0}]: \bar{A}_t \leq \bar{A}_{z_{i_0+1}-s_{q_0}} - 2\lambda\}$.  Using a very similar argument to that in case (b), we find that $e_{q_0} - z_{i_0+1} \geq n\tau/15$, and therefore by Lemma~\ref{Lemma:TwoCP} in the online supplement again, $z_{i_0+1}-s_{q_0}-(\zeta+1)\leq 150C^{-1/2}n\tau$.  Now, on $\Omega_3$, we have $\bar{A}_{z_{i_0}-s_{q_0}} > \bar{A}_{\hat{z}^{[q_0]}} > \bar{A}_{z_{i_0+1}-s_{q_0}} - 2\lambda \geq \bar{A}_\zeta$ and $\zeta-(z_{i_0}-s_{q_0}) \geq n\tau - n\rho - 1$. So we can apply the same argument as in case (b) with $\zeta$ taking the role of $z_{i_0+1}$ and $\tau/2-1/n$ in place of $\tau$, and obtain that
\[
\hat{z}^{[q_0]}-(z_{i_0}-s_{q_0}) \leq C'n\rho
\]
for some universal constant $C' > 0$ as desired.
\end{proof}

\begin{proof}[Proof of Theorem~\ref{Thm:TemporalDependence}]
Writing $E^{(1)} := \mathcal{T}(W^{(1)})$ and $n_1 := n/2$, by Lemma~\ref{Lemma:TimeDependent} in the online supplement and a union bound, we have that the event $\Omega_* := \{\|E^{(1)}\|_\infty \leq \lambda\}$ satisfies
\[
\mathbb{P}(\Omega_*^c) = \mathbb{P}\bigl(\|E^{(1)}\|_\infty \geq \sigma\sqrt{8B\log(n_1p)}\bigr) \leq (n_1-1)pe^{-2 \log(n_1p)} \leq \frac{1}{n_1p}.
\]
Moreover, following the proof of Proposition~\ref{Prop:SinTheta}, on $\Omega_*$,
\[
\sin \angle(\hat{v}^{(1)},v) \leq \frac{64\sqrt{2}\sigma\sqrt{kB\log(n_1p)}}{\tau\vartheta\sqrt{n_1}} \leq \frac{1}{2},
\] 
provided that, in condition~\eqref{Eq:T3condprime}, we take the universal constant $C > 0$ sufficiently large.  Now following the notation and proof of Theorem~\ref{Thm:SingleCP}, but using Lemma~\ref{Lemma:TimeDependent} instead of Lemma~\ref{Lemma:BrownianBridge} in the online supplement, and writing $\lambda_1 := \sigma\sqrt{8B\log n_1}$, we have
\[
\mathbb{P}(\|\bar{E}\|_\infty \geq \lambda_1) \leq (n_1-1)e^{-2\log(n_1)} \leq \frac{1}{n_1}.
\]
Similarly, using Lemma~\ref{Lemma:TimeDependent} in the online supplement again instead of Lemma~\ref{Lemma:LIL}, the event $\Omega_1$ defined in~\eqref{Eq:Omega1} satisfies
\[
\mathbb{P}(\Omega_1^c) \leq 4n_1e^{-\lambda_1^2/(4B\sigma^2)} \leq \frac{4}{n_1}.
\]
The proof therefore follows from that of Theorem~\ref{Thm:SingleCP}.
\end{proof}

\clearpage

\begin{center}
{\Large Online supplementary material for `High-dimensional changepoint estimation via sparse projection'}
\end{center}

This is the online supplementary material for the main paper \citet{WangSamworth2016}, hereafter referred to as the main text.  We begin with several additional theoretical results, which are referred to in the main text.  Subsequent subsections consist of auxiliary results needed for the proofs of our main theorems.

\section{Additional theoretical results}

Our first result is an analogue of Proposition~\ref{Prop:SinTheta} for the (computationally inefficient) estimator of the $k$-sparse leading left singular vector.
\begin{prop}
\label{Prop:Wedin}
Let $X \sim P\in\mathcal{P}(n,p,k,1,\vartheta,\tau,\sigma^2)$, with the single changepoint located at $z$, say (so we may take $\tau = n^{-1}\min\{z,n-z\}$). Define $A,E$ and $T$ as in Section~\ref{Sec:SingleCP} of the main text. Let $v \in\argmax_{\tilde{v}\in\mathbb{S}^{p-1}}\|A^\top \tilde{v}\|_2$ and $\hat{v}\in\argmax_{\tilde{v}\in\mathbb{S}^{p-1}(k)}\|T^\top \tilde{v}\|_2$. If $n\geq 6$, then with probability at least $1-4(p\log n)^{-1/2}$,
\[
\sin\angle(\hat{v},v) \leq \frac{16\sqrt{2}\sigma}{\vartheta\tau}\sqrt\frac{k\log (p\log n)}{n}.
\]
\end{prop}
\begin{proof}
From the definition in Section~\ref{Sec:SingleCP} of the main text, $A = \theta\gamma^\top$, for some $\theta\in\mathbb{R}^{p}$ satisfying $\|\theta\|_0\leq k$ and $\|\theta\|_2\geq \vartheta$ and $\gamma$ defined by~\eqref{Eq:gamma} in the main text. Then we have $v = \theta/\|\theta\|_2$. Define also $u := \gamma/\|\gamma\|_2$ and $\hat{u}:= T^\top \hat{v} / \|T^\top \hat{v}\|_2$. Then by definition of $\hat{v}$, we have
\begin{equation}
\label{Eq:tmpL7}
\langle \hat{v}\hat{u}^\top, T\rangle = \|T^\top \hat{v}\|_2 \geq v^\top T u = \langle vu^\top, T\rangle.
\end{equation}
By Lemma~\ref{Lemma:CurvatureSingularVector} and~\eqref{Eq:tmpL7}, we obtain
\begin{align}
\|vu^\top - \hat{v}\hat{u}^\top\|_2^2& \leq \frac{2}{\|\theta\|_2\|\gamma\|_2}\langle A, vu^\top - \hat{v}\hat{u}^\top\rangle \nonumber\\
&\leq \frac{2}{\|\theta\|_2\|\gamma\|_2}\langle A-T,vu^\top - \hat{v}\hat{u}^\top\rangle\leq \frac{2}{\|\theta\|_2\|\gamma\|_2}\|E\|_\infty \|vu^\top - \hat{v}\hat{u}^\top\|_1.\label{Eq:tmp2L7}
\end{align}
Note that in fact $v \in \mathbb{S}^{p-1}(k)$, by definition of the matrix $A$.  Moreover, $\hat{v}\in\mathbb{S}^{p-1}(k)$ too, so the matrix $vu^\top - \hat{v}\hat{u}^\top$ has at most $2k$ non-zero rows. Thus, by the Cauchy--Schwarz inequality,
\[
\|vu^\top - \hat{v}\hat{u}^\top\|_1 \leq \sqrt{2kn}\|vu^\top - \hat{v}\hat{u}^\top\|_2.
\]
By~\eqref{Eq:maxsin} in the proof of Proposition~\ref{Prop:Perturbation}, and~\eqref{Eq:tmp2L7}, we find that
\[
\sin\angle (\hat{v},v) \leq \|vu^\top - \hat{v}\hat{u}^\top\|_2 \leq \frac{2\sqrt{2}\|E\|_\infty\sqrt{kn}}{\|\theta\|_2\|\gamma\|_2} \leq \frac{8\sqrt{2k}\|E\|_\infty}{\vartheta\tau \sqrt{n}},
\]
where we have used Lemma~\ref{Lemma:gamma} in the final inequality. The desired result follows from bounding $\|E\|_\infty$ with high probability as in~\eqref{Eq:TailProbability} of the main text.
\end{proof}
We next derive the closed-form expression for the solution to the optimisation problem~\eqref{Eq:SoftThresholding} in the main text.  Recall the defintions of the set $\mathcal{S}_2$ and the $\soft$ function, both given just before~\eqref{Eq:SoftThresholding} in the main text.
\begin{prop}
\label{Lemma:Duality}
Let $T\in\mathbb{R}^{p\times (n-1)}$ and $\lambda > 0$. Then the following optimisation problem
\[
\max_{M\in\mathcal{S}_2} \bigl\{\langle T, M\rangle - \lambda\|M\|_1\bigr\}
\]
has a unique solution given by 
\begin{equation}
\label{Eq:hatM}
\tilde{M} = \frac{\soft(T,\lambda)}{\|\soft(T,\lambda)\|_2}.
\end{equation}
\end{prop}
\begin{proof}
Define $\phi(M,R) := \langle T-R, M\rangle$ and $\mathcal{R} := \{R\in\mathbb{R}^{p\times (n-1)}: \|R\|_\infty\leq \lambda\}$. Then the objective function in the lemma is given by 
\[
f(M) = \min_{R\in\mathcal{R}} \phi(M,R).
\]
We also define
\[
g(R) := \max_{M\in\mathcal{S}_2}\phi(M,R) = \|T-R\|_2.
\]
Since $\mathcal{S}_2$ and $\mathcal{R}$ are compact, convex subsets of $\mathbb{R}^{p \times (n-1)}$ endowed with the trace inner product, and since $\phi$ is affine and continuous in both $M$ and $R$, we can use the minimax equality theorem \citet[Theorem~1]{Fan1953} to obtain
\[
\max_{M\in\mathcal{S}_2} f(M) = \max_{M\in\mathcal{S}_2} \min_{R\in\mathcal{R}} \phi(M,R) = \min_{R\in\mathcal{R}}\max_{M\in\mathcal{S}_2}\phi(M,R) = \min_{R\in\mathcal{R}} g(R).
\]
We note that the dual function $g$ has a unique minimum over $\mathcal{R}$ at $R^{(d)}$, say, where $R^{(d)}_{j,t} := \mathrm{sgn}(T_{j,t})\min(\lambda, |T_{j,t}|)$. Let 
\[
M^{(d)} \in \argmax_{M\in\mathcal{S}_2} \phi(M,R^{(d)}), \quad M^{(p)} \in \argmax_{M\in\mathcal{S}_2} f(M) \quad \text{and} \quad R^{(p)} \in \argmin_{R\in\mathcal{R}} \phi(M^{(p)},R).
\]
Then 
\[
\min_{R\in\mathcal{R}} g(R) = \langle T-R^{(d)},M^{(d)} \rangle \geq \langle T-R^{(d)}, M^{(p)}\rangle \geq \langle T-R^{(p)},M^{(p)}\rangle = \max_{M\in\mathcal{S}_2} f(M).
\]
Since the two extreme ends of the chain of inequalities are equal, we necessarily have
\[
R^{(d)}\in\argmin_{R\in\mathcal{R}} \langle T-R,M^{(p)}\rangle,
\]
and consequently, 
\[
M^{(p)}\in\argmax_{M\in\mathcal{S}_2} \langle T-R^{(d)}, M\rangle.
\]
The objective $M \mapsto \langle T-R^{(d)}, M\rangle = \langle \soft(T,\lambda), M\rangle$ has a unique maximiser over $\mathcal{S}_2$ at $\tilde{M}$ defined in~\eqref{Eq:hatM}. Thus, $M^{(p)}$ is unique and has the form given in the proposition.
\end{proof}
Proposition~\ref{Prop:Minimax} below gives a minimax lower bound for the single changepoint estimation problem.  In conjunction with Theorem~\ref{Thm:SingleCP}, this confirms that the \texttt{inspect} algorithm attains the minimax optimal rate of estimation up to a factor of $\log \log n$.  
\begin{prop}
 \label{Prop:Minimax}
 Assume $n\geq 3$, $\tau \leq 1/3$.  Then for every $c \in (0,\sqrt{2})$, we have
 \[
 \inf_{\hat{z}}\sup_{P\in \mathcal{P}(n,p,k,1,\vartheta,\tau,\sigma^2)} \mathbb{E}_P \bigl\{n^{-1}|\hat{z}-z|\bigr\} \geq \begin{cases}
   \frac{\sigma}{13n\vartheta}\exp\{-\frac{\vartheta^2}{8\sigma^2}\} & \text{if $\vartheta/\sigma > 1$}\\
   \frac{\sigma^2}{16n\vartheta^2} & \text{if $(n\tau)^{-1/2}\leq \vartheta/\sigma \leq 1$}\\
   \frac{1}{12}\bigl(1 - \frac{c}{\sqrt{2}}\bigr) & \text{if $\vartheta/\sigma < c(n\tau)^{-1/2}$} 
 \end{cases},
 \]
where the infimum is taken over all estimators $\hat{z}$ of $z$.  
\end{prop}
\textbf{Remark}: In this result, the second and third regions overlap when $c \in (1,\sqrt{2})$.  In that case, both lower bounds hold.  The most interesting region is where $\sqrt{2}(n\tau)^{-1/2}\leq \vartheta/\sigma \leq 1$, corresponding to challenging but feasible problems.  When $\vartheta/\sigma < \sqrt{2}(n\tau)^{-1/2}$, consistent estimation of changepoints is impossible, while when the signal-to-noise ratio $\vartheta/\sigma$ is a large constant, one can estimate the changepoint location exactly with high probability.  
\begin{proof}
Since $\tau \leq 1/3$, we may assume without loss of generality that $z \leq n/3$, and $\tau = z/n$. We first assume that $z^{-1/2}\leq \vartheta/\sigma \leq 1$. Consider the two distributions $Q, Q' \in \mathcal{P}(n,p,k,1,\vartheta,\tau,\sigma^2)$ with mean matrices $\boldsymbol \mu = (\mu_{j,t})_{1\leq j\leq p, 1\leq t\leq n}$ and $\boldsymbol \mu' = (\mu'_{j,t})_{1\leq j\leq p, 1\leq t\leq n}$ given respectively by
 \[
  \mu_{j,t} = \begin{cases} \vartheta/\sqrt{k} & \text{if $j\leq k$ and $t\leq z$}\\ 0 & \text{otherwise}\end{cases}
 \quad \text{and}\quad
  \mu'_{j,t} = \begin{cases} \vartheta/\sqrt{k} & \text{if $j\leq k$ and $t\leq z + \Delta$}\\ 0 & \text{otherwise}\end{cases},
 \]
 where $\Delta \in (0, n/3]$ is an integer to be chosen.  Let $d_{\mathrm{TV}}(Q, Q'):= \sup_{A}|Q(A)-Q'(A)|$ denote the total variation distance between $Q$ and $Q'$, where the supremum is taken over all measurable subsets of $\mathbb{R}^{p \times n}$, and write $D(Q\|Q'):=  \mathbb{E}_Q(\log \frac{dQ}{dQ'})$ for the Kullback--Leibler divergence.  Then by a standard bound between these two quantities (see, e.g.~\citet[p.~62]{Pollard2002}),
 \[
  d^2_{\mathrm{TV}}(Q,Q') \leq \frac{1}{2}D(Q\|Q') = \frac{1}{4\sigma^2}\|\boldsymbol\mu - \boldsymbol\mu'\|_2^2 = \frac{\vartheta^2\Delta}{4\sigma^2}.
 \]
 Therefore, 
 \begin{align*}
  \inf_{\hat{z}}\sup_{P\in \mathcal{P}(n,p,k,1,\vartheta,\tau,\sigma^2)} \mathbb{E}_P \bigl\{n^{-1}|\hat{z}-z|\bigr\} &\geq \inf_{\hat{z}} \max_{P\in\{Q,Q'\}}\mathbb{E}_P \bigl\{n^{-1}|\hat{z}-z|\bigr\}\\
  & \geq \frac{\Delta}{2n}\inf_{\hat{z}} \max\bigl\{\mathbb{P}_Q(\hat{z}\geq z+\Delta/2), \mathbb{P}_{Q'}(\hat{z} < z+\Delta/2)\bigr\}\\
  &\geq \frac{\Delta}{2n}\frac{1-d_{\mathrm{TV}}(Q,Q')}{2} \geq \frac{\Delta}{4n}\biggl(1-\frac{\vartheta\Delta^{1/2}}{2\sigma}\biggr).
 \end{align*}
 The desired bounds follows from setting $\Delta = \lfloor (\sigma/\vartheta)^2 \rfloor$ and observing that for $1\leq \sigma/\vartheta\leq z^{1/2}$ we have $\sigma^2/(2\vartheta^2) \leq \Delta\leq n/3$.

 For the case $\vartheta/\sigma > 1$, we consider the same two distributions $Q$ and $Q'$ as in the previous case, but set $\Delta = 1$.  Writing $\Phi$ for the standard normal distribution function, we can use the following alternative bound on the total variation distance:
 \[
 \frac{1-d_{\mathrm{TV}}(Q,Q')}{2} = 1 - \Phi(\|\boldsymbol{\mu} - \boldsymbol{\mu}'\|_2/2) \geq \frac{\vartheta /(2\sigma)}{\vartheta^2/(4\sigma^2)+1}(2\pi)^{-1/2} e^{-\frac{\vartheta^2}{8\sigma^2}}\geq \frac{2\sigma}{5\vartheta} (2\pi)^{-1/2}e^{-\frac{\vartheta^2}{8\sigma^2}}.
 \]
 We therefore obtain the desired minimax lower bound
 \[
 \inf_{\hat{z}}\sup_{P\in \mathcal{P}(n,p,k,1,\vartheta,\tau,\sigma^2)} \mathbb{E}_P \bigl\{n^{-1}|\hat{z}-z|\bigr\} \geq \frac{\Delta}{2n}\frac{1-d_{\mathrm{TV}}(Q,Q')}{2} \geq \frac{\sigma}{13n\vartheta}e^{-\frac{\vartheta^2}{8\sigma^2}}.
 \]
 
 Finally, for the case $\vartheta/\sigma < cz^{-1/2}$ for some $c \in (0,\sqrt{2})$, we consider two different distributions $Q, Q' \in \mathcal{P}(n,p,k,1,\vartheta,\tau,\sigma^2)$ with mean matrices $\boldsymbol \mu = (\mu_{j,t})_{1\leq j\leq p, 1\leq t\leq n}$ and $\boldsymbol \mu' = (\mu'_{j,t})_{1\leq j\leq p, 1\leq t\leq n}$ given respectively by
 \[
  \mu_{j,t} = \begin{cases} \vartheta/\sqrt{k} & \text{if $j\leq k$ and $t\leq z$}\\ 0 & \text{otherwise}\end{cases}
 \quad \text{and}\quad
  \mu'_{j,t} = \begin{cases} \vartheta/\sqrt{k} & \text{if $j\leq k$ and $t > n-z$}\\ 0 & \text{otherwise}\end{cases}.
  \]
  Then
  \[
   d^2_{\mathrm{TV}}(Q,Q') \leq \frac{1}{2}D(Q\|Q') = \frac{1}{4\sigma^2}\|\boldsymbol\mu - \boldsymbol\mu'\|_2^2 = \frac{\vartheta^2z}{2\sigma^2} < \frac{c^2}{2}.
   \]
   Therefore,
 \begin{align*}
  \inf_{\hat{z}}\sup_{P\in \mathcal{P}(n,p,k,1,\vartheta,\tau,\sigma^2)} \mathbb{E}_P \bigl\{n^{-1}|\hat{z}-z|\bigr\} &\geq \inf_{\hat{z}} \max_{P\in\{Q,Q'\}}\mathbb{E}_P \bigl\{n^{-1}|\hat{z}-z|\bigr\}\\
  & \geq \biggl(\frac{1}{2}-\frac{z}{n}\biggr)\inf_{\hat{z}} \max\bigl\{\mathbb{P}_Q(\hat{z}\geq n/2), \mathbb{P}_{Q'}(\hat{z} < n/2)\bigr\}\\
  &\geq \biggl(\frac{1}{2}-\tau\biggr)\frac{1-d_{\mathrm{TV}}(Q,Q')}{2} \geq \frac{1}{12}\biggl(1 - \frac{c}{\sqrt{2}}\biggr).
 \end{align*}
 as desired.
\end{proof}
Finally in this section, we provide theoretical guarantees for the performance of our modified \texttt{inspect} algorithm (Algorithm~\ref{Algo:SpatDep}) in cases of both local and global spatial dependence.
\begin{thm}
\label{Thm:SpatDep}
\textup{\textbf{(Local spatial dependence)}} Suppose that $\Sigma = (\Sigma_{i,j}) = (\rho^{|i-j|})$ for some $\rho \in (-1,1)$.  Let $\hat{z}$ be the output of Algorithm~\ref{Algo:SpatDep} in the main text with $\lambda := 2\sqrt{\log(p\log n)}$, where in Step 4, we let $\hat{\Theta}^{(1)}$ be the estimator of $\Sigma^{-1}$ based on $W_1',\ldots,W_m'$ defined in Lemma~\ref{Lemma:LocalCS}.  There exist universal constants $C,C' > 0$ such that if $n \geq 12$ is even, $z$ is even, $m(p-1) \geq 4(1-|\rho|)^2\log m$ and
\begin{equation}
\label{Eq:condLocalCS}
\frac{(1+|\rho|)^3\log m}{m^{1/2}(p-1)^{1/2}(1-|\rho|)^3} + \frac{2^{1/2}(1+|\rho|)^4\lambda k^{1/2}}{\vartheta \tau n_1^{1/2}(1-|\rho|)^4} \leq \frac{1}{C},
\end{equation}
then for $h(\rho) := (1-|\rho|)^{-4}\{9+\rho^2+20\rho^2(1-\rho^2)^{-1}\}$, we have
\[
\mathbb{P}\biggl(|\hat{z} - z| > \frac{C'\log \log n}{n \vartheta^2}\biggl(\frac{1+|\rho|}{1-|\rho|}\biggr)^3\biggr) \leq \frac{4}{\{p \log(n/2)\}^{1/2}} + \frac{9}{\log(n/2)} + \frac{144 h(\rho)}{\log^2 m}.
\]
\textup{\textbf{(Global spatial dependence)}} Suppose that $\Sigma = I_p + \frac{\rho}{p}\mathbf{1}_p\mathbf{1}_p^\top$ for some $-1 < \rho \leq p$.  Let $\hat{z}$ be the output of Algorithm~\ref{Algo:SpatDep} with $\lambda := 2\sqrt{2\log(p\log n)}$, where in Step 4, we let $\hat{\Theta}^{(1)}$ be the estimator of $\Sigma^{-1}$ based on $W_1',\ldots,W_m'$ defined in Lemma~\ref{Lemma:GlobalCS}.  There exist universal constants $C,C' > 0$ such that if $n \geq 12$ is even, $z$ is even, $m \geq 10$ and
\begin{equation}
\label{Eq:condGlobalCS}
\frac{\log m}{m^{1/2}} + \frac{\lambda k^{1/2}}{\vartheta \tau n^{1/2}} \leq \frac{\min\{(1+\rho)^2,(1+\rho)^{-2}\}}{C},
\end{equation}
then 
\[
\mathbb{P}\biggl(|\hat{z} - z| > \frac{C'\log \log n\,\max(1,1+\rho)^2}{n \vartheta^2\min(1,1+\rho)}\biggr) \leq \frac{4}{\{p \log(n/2)\}^{1/2}} + \frac{9}{\log(n/2)} + \frac{21}{(1+\rho)^2\log^2 m}.
\]
\end{thm}  
\begin{proof}
\textbf{(Local spatial dependence)} Let
\[
y := \frac{(1+|\rho|)\log m}{m^{1/2}(p-1)^{1/2}(1-|\rho|)} + \frac{2^{1/2}(1+|\rho|)^2\lambda k^{1/2}}{\vartheta \tau n_1^{1/2}(1-|\rho|)^2}.
\]
By Lemmas~\ref{Lemma:AuBv} and~\ref{Lemma:LocalCS} together with Proposition~\ref{Prop:SinTheta} (which still applies in this context), there is an event $\Omega_0$ with probability at least $1 - 4(p \log n_1)^{-1/2} - 144h(\rho)\log^{-2} m$ such that on $\Omega_0$, for $C\geq 40$ in~\eqref{Eq:condLocalCS}, we have
\[
\sin \angle(\hat{v}_{\mathrm{proj}}^{(1)},v_{\mathrm{proj}}) \leq 6y+2y^2 \leq \frac{\sigma_{\min}(\Sigma)}{5\sigma_{\max}(\Sigma)}.
\]
Then, on the same event $\Omega_0$, $(\hat{v}_{\mathrm{proj}}^{(1)})^\top X^{(2)}$ is a univariate series with a signal to noise ratio of
\begin{align*}
 \frac{|(\hat{v}_{\mathrm{proj}}^{(1)})^\top \theta|}{\{(\hat{v}_{\mathrm{proj}}^{(1)})^\top \Sigma \hat{v}_{\mathrm{proj}}^{(1)}\}^{1/2}} &\geq \frac{|v_{\mathrm{proj}}^\top \theta| - \vartheta\|\hat{v}_{\mathrm{proj}}^{(1)} - v_{\mathrm{proj}}\|_2}{\{v_{\mathrm{proj}}^\top\Sigma v_{\mathrm{proj}} + 2\sigma_{\max}(\Sigma)\|\hat{v}_{\mathrm{proj}}^{(1)} - v_{\mathrm{proj}}\|_2\}^{1/2}}\\
 & \geq \frac{\vartheta\bigl\{\sigma_{\max}(\Sigma)^{-1}v_{\mathrm{proj}}^\top\Sigma v_{\mathrm{proj}} - 2^{1/2}\sin\angle(\hat{v}_{\mathrm{proj}}^{(1)}, v_{\mathrm{proj}})\bigr\}}{\bigl\{v_{\mathrm{proj}}^\top\Sigma v_{\mathrm{proj}} + 2^{3/2}\sigma_{\max}(\Sigma) \sin\angle(\hat{v}_{\mathrm{proj}}^{(1)}, v_{\mathrm{proj}})\bigr\}^{1/2}}\\
 &\geq \frac{\vartheta}{2}\sigma_{\max}(\Sigma)^{-1}\bigl(v_{\mathrm{proj}}^\top\Sigma v_{\mathrm{proj}}\bigr)^{1/2}\geq \frac{\vartheta}{2}\biggl(\frac{1-|\rho|}{1+|\rho|}\biggr)^{3/2}.
\end{align*}
Therefore, following proof of Theorem~\ref{Thm:SingleCP} in the main text, we obtain that
\[
\mathbb{P}\biggl(|\hat{z} - z| > \frac{C'\log \log n}{n \vartheta^2}\biggl(\frac{1+|\rho|}{1-|\rho|}\biggr)^3\biggr) \leq \frac{4}{\{p \log(n/2)\}^{1/2}} + \frac{9}{\log(n/2)} + \frac{144 h(\rho)}{\log^2 m}. 
\]

\noindent \textbf{(Global spatial dependence)} Let
\[
y := \frac{1}{\min(1,1+\rho)}\biggl\{\frac{\log m}{m^{1/2}} + \max(1,1+\rho)\frac{\lambda k^{1/2}}{\vartheta \tau n_1^{1/2}}\biggr\}.
\]
By Lemmas~\ref{Lemma:AuBv} and~\ref{Lemma:GlobalCS} together with Proposition~\ref{Prop:SinTheta}, there is an event $\Omega_1$ with probability at least $1-4(p\log n_1)^{-1/2}-21(1+\rho)^{-2}\log^{-2}m$ such that on $\Omega_1$, for $C\geq 40$ in~\eqref{Eq:condGlobalCS}, we have  
\[
\sin \angle(\hat{v}_{\mathrm{proj}}^{(1)},v_{\mathrm{proj}}) \leq  6y+2y^2\leq \frac{\sigma_{\min}(\Sigma)}{5\sigma_{\max}(\Sigma)}.
\]
Then, by a similar calculation as in the local spatial dependence case, we find that the univariate series $(\hat{v}_{\mathrm{proj}}^{(1)})^\top X^{(2)}$ has signal to noise ratio
\[
 \frac{|(\hat{v}_{\mathrm{proj}}^{(1)})^\top \theta|}{\{(\hat{v}_{\mathrm{proj}}^{(1)})^\top \Sigma \hat{v}_{\mathrm{proj}}^{(1)}\}^{1/2}} \geq \frac{\vartheta}{2}\sigma_{\max}(\Sigma)^{-1}\bigl(v_{\mathrm{proj}}^\top\Sigma v_{\mathrm{proj}}\bigr)^{1/2}\geq \frac{\vartheta}{2} \frac{\min(1,1+\rho)^{1/2}}{\max(1,1+\rho)}.
\]
Therefore, following the proof of Theorem~\ref{Thm:SingleCP} in the main text, we obtain that
\[
\mathbb{P}\biggl(|\hat{z} - z| > \frac{C'\log \log n\,\max(1,1+\rho)^2}{n \vartheta^2\min(1,1+\rho)}\biggr) \leq \frac{4}{\{p \log(n/2)\}^{1/2}} + \frac{9}{\log(n/2)} + \frac{21}{(1+\rho)^2\log^2 m},
\]
as desired.
\end{proof}

\section{Auxiliary results}

\subsection{Auxiliary results for the proof of Proposition~\ref{Prop:SinTheta} in the main text}

The lemma below gives a characterisation of the nuclear norm of a real matrix.
\begin{lemma}
\label{Lemma:NuclearNorm}
For $n,p\geq 1$, let $\mathcal{V}_n$ and $\mathcal{V}_p$ be respectively the sets of $n\times \min(n,p)$ and $p\times \min(n,p)$ real matrices having orthonormal columns. Let $A\in\mathbb{R}^{p\times n}$.  Then
\[
\|A\|_{*} = \sup_{V\in\mathcal{V}_p, U\in\mathcal{V}_n}\langle VU^\top, A\rangle.
\]
\end{lemma}
\begin{proof}
Suppose we have the singular value decomposition $A = \tilde{V} D \tilde{U}^\top$ where $\tilde{V}\in \mathbb{R}^{p\times p}$ and $\tilde{U}\in\mathbb{R}^{n\times n}$ are orthogonal matrices and where $D = (D_{ij}) \in\mathbb{R}^{p\times n}$ has entries arranged in decreasing order along its main diagonal and is zero off the main diagonal.  Writing $v_j^\top$ and $u_j^\top$ for the $j$th row of $V$ and $U$ respectively, we have
\begin{align*}
\sup_{V\in\mathcal{V}_p, U\in\mathcal{V}_n}\langle VU^\top, A\rangle &= \sup_{V\in\mathcal{V}_p, U\in\mathcal{V}_n}\langle VU^\top, \tilde{V}D\tilde{U}^\top\rangle = \sup_{V\in\mathcal{V}_p, U\in\mathcal{V}_n}\langle VU^\top, D\rangle\\
& = \sup_{V\in\mathcal{V}_p, U\in\mathcal{V}_n}\sum_{j=1}^{\min(n,p)}D_{jj} v_j^\top u_j = \sum_{j=1}^{\min(n,p)}D_{jj} = \|A\|_*,
\end{align*}
as desired.
\end{proof}
Next, we present a generalisation of the curvature lemma of \citet[][Lemma~3.1]{Vuetal2013}.
\begin{lemma}
\label{Lemma:CurvatureSingularVector}
Let $v \in\mathbb{S}^{p-1}$ and $u \in \mathbb{S}^{n-1}$ be the leading left and right singular vectors of $A\in\mathbb{R}^{p\times n}$ respectively. Suppose that the first and second largest singular values of $A$ are separated by $\delta > 0$. Let $M \in\mathbb{R}^{p\times n}$. If either of the following two conditions holds,
\begin{enumerate}[label={(\alph*)}, noitemsep]
\item $\mathrm{rank}(A) = 1$ and $\|M\|_2\leq 1$,
\item $\|M\|_*\leq 1$,
\end{enumerate}
then
\[
\|vu^\top - M\|_2^2 \leq \frac{2}{\delta}\langle A, vu^\top - M\rangle.
\]
\end{lemma}
\textbf{Remark}: We note that if $v\in\mathbb{S}^{p-1}$ and $u\in\mathbb{S}^{n-1}$ are the leading left and right singular vectors respectively of $A\in\mathbb{R}^{p\times n}$, then since the matrix operator norm and the nuclear norm are dual norms with respect to the trace inner product, we have that
\[
\langle A,  vu^\top\rangle = v^\top A u = \|A\|_{\mathrm{op}} = \sup_{M \in \mathcal{S}_1} \langle A,M\rangle.
\]
Thus, Lemma~\ref{Lemma:CurvatureSingularVector} provides a lower bound on the curvature of the function $M\mapsto \langle A,M\rangle$ as $M$ moves away from the maximiser of the function in $\mathcal{S}_1$.
\begin{proof}
Let $A = VDU^\top$ be the singular value decomposition of $A$, where $V\in\mathbb{R}^{p\times p}$ and $U\in\mathbb{R}^{n\times n}$ are orthogonal matrices with column vectors $v_1 = v,v_2,\ldots,v_p$ and $u_1 = u,u_2,\ldots,u_n$ respectively, and $D\in\mathbb{R}^{p\times n}$ is a rectangular diagonal matrix with nonnegative entries along its main diagonal. The diagonal entries $\sigma_i := D_{ii}$ are the singular values of $A$, and we may assume without loss of generality that $\sigma_1\geq \cdots \geq \sigma_r > 0$ are all the positive singular values, for some $r\leq \min\{n,p\}$. 

Let $\tilde{M} := V^\top M U$ and denote $e_1^{[d]} := (1,0,\ldots,0)^\top\in\mathbb{R}^d$. Then by unitary invariance of the Frobenius norm, we have
\begin{equation}
\label{Eq:PerturbationSquared}
\|v_1 u_1^\top - M\|_2^2 = \|e_1^{[p]}(e_1^{[n]})^\top - \tilde{M}\|_2^2 = \|\tilde{M}\|_2^2 + 1 - 2\tilde{M}_{11}.
\end{equation}
On the other hand, 
\begin{equation}
\label{Eq:InnerProductDifference}
\langle A, v_1u_1^\top - M\rangle = \langle D, e_1^{[p]}(e_1^{[n]})^\top - \tilde{M}\rangle = \sigma_1 - \sum_{i=1}^r \sigma_i \tilde M_{ii} \geq \sigma_1(1-\tilde M_{11}) - \sigma_2\sum_{i=2}^r|\tilde M_{ii}|.
\end{equation}
If condition~(a) holds, then $\sigma_2 = 0$ and $\delta = \sigma_1$, so by~\eqref{Eq:PerturbationSquared} and~\eqref{Eq:InnerProductDifference}, we have
\[
\|v_1 u_1^\top - M\|_2^2 \leq 2(1 -\tilde{M}_{11}) = \frac{2}{\delta}\langle A, v_1u_1^\top - M\rangle,
\]
as desired.

On the other hand, if condition~(b) holds, then by the characterisation of the nuclear norm in Lemma~\ref{Lemma:NuclearNorm}, as well as its unitary invariance, we have
\begin{equation}
\label{Eq:NuclearNorm}
\sum_{i=1}^r |\tilde M_{ii}| = \sup_{\substack{\text{$U$ $\in \mathbb{R}^{p \times n}$ diagonal}\\ U_{ii}\in\{\pm 1\} \;\forall i}}\langle U, \tilde{M}\rangle \leq \|\tilde M\|_*  = \|M\|_* \leq 1.
\end{equation}
But if $\|M\|_* \leq 1$, then $\sigma_i\leq 1$ for all $i$, so
\begin{equation}
\label{Eq:L1L2relation}
\|M\|_2 = \biggl(\sum_{i=1}^r \sigma_i^2\biggr)^{1/2} \leq \biggl(\sum_{i=1}^r \sigma_i\biggr)^{1/2} \leq 1.
\end{equation}
Using~\eqref{Eq:PerturbationSquared}, \eqref{Eq:InnerProductDifference}, \eqref{Eq:NuclearNorm} and~\eqref{Eq:L1L2relation}, we therefore have
\begin{align*}
\langle A, v_1u_1^\top-M\rangle& \geq \sigma_1(1-\tilde M_{11}) - \sigma_2\sum_{i=2}^r|\tilde M_{ii}| \geq (\sigma_1-\sigma_2)(1-\tilde{M}_{11})\\
& \geq \frac{\delta}{2}(\|\tilde{M}\|_2^2 + 1 - 2\tilde{M}_{11}) = \frac{\delta}{2}\|v_1 u_1^\top - M\|_2^2,
\end{align*}
as desired.
\end{proof}
\begin{prop}
\label{Prop:Perturbation}
Suppose the first and second largest singular values of $A\in\mathbb{R}^{p\times n}$ are separated by $\delta > 0$. Let $v\in \mathbb{S}^{p-1}(k)$ and $u\in \mathbb{S}^{n-1}(\ell)$ be left and right leading singular vectors of $A$ respectively. Let $T \in \mathbb{R}^{p\times n}$ satisfy $\|T-A\|_\infty \leq \lambda$ for some $\lambda > 0$, and let $\mathcal{S}$ be a subset of $p\times n$ real matrices containing $vu^\top$. Suppose one of the following two
conditions holds:
\begin{enumerate}[label={(\alph*)}, noitemsep]
\item $\mathrm{rank}(A) = 1$ and $\mathcal{S}\subseteq \{M\in\mathbb{R}^{p\times n}: \|M\|_2\leq 1\}$
\item $\mathcal{S}\subseteq \{M\in\mathbb{R}^{p\times n}: \|M\|_*\leq 1\}$.
\end{enumerate}
Then for any 
\[
\hat{M} \in \argmax_{M\in\mathcal{S}} \bigl\{\langle T, M\rangle - \lambda\|M\|_1\bigr\},
\]
we have 
\[
\|vu^\top - \hat{M}\|_2 \leq \frac{4\lambda\sqrt{k\ell}}{\delta}.
\]
Furthermore, if $\hat{v}$ and $\hat{u}$ are leading left and right singular vectors of $\hat{M}$ respectively, then 
\begin{equation}
\max\{\sin\angle(\hat{v},v), \sin\angle(\hat{u},u)\}\leq \frac{8\lambda\sqrt{k\ell}}{\delta}.
\label{Eq:AngleBound}
\end{equation}
\label{Prop:SparseSingularVector}
\end{prop}
\begin{proof}
Using Lemma~\ref{Lemma:CurvatureSingularVector}, we have
\begin{align}
\|vu^\top - \hat{M}\|_2^2 &\leq \frac{2}{\delta}\langle A, vu^\top - \hat{M}\rangle\nonumber\\
&=\frac{2}{\delta}\bigl(\langle T, vu^\top - \hat{M}\rangle + \langle A-T, vu^\top -\hat{M}\rangle\bigr).
\label{Eq:TBC}
\end{align}
Since $\hat{M}$ is a maximiser of the objective function $M\mapsto \langle T,M\rangle - \lambda\|M\|_1$ over the set $\mathcal{S}$, and since $vu^\top\in\mathcal{S}$, we have the basic inequality
\begin{equation}
\label{Eq:BasicInequality}
\langle T, vu^\top - \hat{M}\rangle \leq \lambda(\|vu^\top\|_1 - \|\hat{M}\|_1).
\end{equation}
Denote $S_v := \{j: 1\leq j\leq p, v_j\neq 0\}$ and $S_u := \{t: 1\leq t\leq n, u_t\neq 0\}$. From~\eqref{Eq:TBC} and~\eqref{Eq:BasicInequality} and the fact that $\|T-A\|_\infty\leq \lambda$, we have
\begin{align*}
\|vu^\top - \hat{M}\|_2^2&\leq \frac{2}{\delta}\bigl(\lambda\|vu^\top\|_1-\lambda\|\hat{M}\|_1 + \lambda \|vu^\top - \hat{M}\|_1\bigr)\\
&= \frac{2\lambda}{\delta}\bigl(\|v_{S_v}u_{S_u}^\top\|_1 - \|\hat{M}_{S_vS_u}\|_1 + \|v_{S_v}u_{S_u}^\top - \hat{M}_{S_vS_u}\|_1\bigr)\\
& \leq \frac{4\lambda}{\delta}\|v_{S_v}u_{S_u}^\top - \hat{M}_{S_vS_u}\|_1 \leq \frac{4\lambda\sqrt{k\ell}}{\delta}\|vu^\top - \hat{M}\|_2.
\end{align*}
Dividing through by $\|vu^\top - \hat{M}\|_2$, we have the first desired result.

Now, by definition of the operator norm, we have
\begin{align*}
\|vu^\top - \hat{M}\|_2^2& = 1+\|\hat{M}\|_2^2-2v^\top\hat{M}u\\
&\geq 1+\|\hat{M}\|_2^2-2\|\hat{M}\|_{\mathrm{op}} = 1+\|\hat{M}\|_2^2-2\hat{v}^\top\hat{M}\hat{u} = \|\hat{v}\hat{u}^\top - \hat{M}\|_2^2.
\end{align*}
Thus,
\begin{equation}
\|vu^\top - \hat{v}\hat{u}^\top\|_2\leq \|vu^\top-\hat{M}\|_2 + \|\hat{v}\hat{u}^\top - \hat{M}\|_2 \leq 2 \|vu^\top - \hat{M}\|_2 \leq \frac{8\lambda\sqrt{k\ell}}{\delta}.
\label{Eq:MatrixBound}
\end{equation}
We claim that 
\begin{equation}
\max \bigl\{\sin^2\angle(\hat{u},u), \sin^2\angle(\hat{v},v)\bigr\}\leq \|vu^\top - \hat{v}\hat{u}^\top\|_2^2.
\label{Eq:maxsin}
\end{equation}
Let $v_0 := (v+\hat{v})/2$ and $\Delta := v-v_0$. Then
\begin{align*}
\|vu^\top - \hat{v}\hat{u}^\top\|_2^2 &= \|(v_0+\Delta)u^\top - (v_0-\Delta)\hat{u}^\top\|_2^2 = \|v_0(u-\hat{u})^\top\|_2^2 + \|\Delta(u+\hat{u})^\top\|_2^2\\
& = \|v_0\|_2^2 \|u-\hat{u}\|_2^2 + \|\Delta\|_2^2\|u+\hat{u}\|_2^2\\
& \geq (\|v_0\|_2^2+\|\Delta\|_2^2)\min(\|u-\hat{u}\|_2^2, \|u+\hat{u}\|_2^2)\\
& \geq 1-(\hat{u}^\top u)^2 = \sin^2\angle(\hat{u},u),
\end{align*}
where the penultimate step uses the fact that $\|v_0\|_2^2+\|\Delta\|_2^2 = 1$.  A similar inequality holds for $\sin^2\angle(\hat{v},v)$, which establishes the desired claim~\eqref{Eq:maxsin}.  Inequality~\eqref{Eq:AngleBound} now follows from~\eqref{Eq:MatrixBound} and~\eqref{Eq:maxsin}.
\end{proof}
The final lemma in this subsection provides bounds on different norms of the vector $\gamma$, which is proportional to each row of the CUSUM transformation of the mean matrix.
\begin{lemma}
\label{Lemma:gamma}
Let $\gamma\in\mathbb{R}^{n-1}$ be defined as in~\eqref{Eq:gamma} of the main text for some $n\geq 6$ and $2\leq z\leq n-2$. Let $\tau:= n^{-1}\min(z,n-z)$. Then 
\begin{align*}
\frac{1}{4}n\tau &\leq \|\gamma\|_2\leq n\tau\sqrt{\log (en/2)}\\
\frac{1}{2}n^{3/2}\tau &\leq \|\gamma\|_1\leq 2.1n^{3/2}\tau.
\end{align*}
\end{lemma}
\begin{proof}
Since the norms of $\gamma$ are invariant under substitution $z\mapsto n-z$, we may assume without loss of generality that $z\leq n-z$. Hence $n\tau = z$. We have that
\begin{align*}
\|\gamma\|_2^2 &= \frac{1}{n}\biggl\{\sum_{t=1}^z \frac{t(n-z)^2}{n-t} + \sum_{t=z+1}^{n-1}\frac{(n-t)z^2}{t}\biggr\} \\
&= n^2\biggl\{\sum_{t=1}^z \frac{(t/n)(1-z/n)^2}{(1-t/n)}\cdot \frac{1}{n} + \sum_{t=z+1}^{n-1}\frac{(1-t/n)(z/n)^2}{t/n}\cdot\frac{1}{n}\biggr\},
\end{align*}
where the expression inside the bracket can be interpreted as a Riemann sum approximation to an integral.  We therefore find that 
\[
n^2\biggl\{I_1 - \frac{(z/n)(1-z/n)}{n}\biggr\}\leq \|\gamma\|_2^2\leq n^2\biggl\{I_1 + \frac{(z/n)(1-z/n)}{n}\biggr\},
\]
where
\begin{align*}
I_1 &:= (1-z/n)^2\int_0^{z/n}\frac{r}{1-r}\,dr + (z/n)^2\int_{z/n}^1\frac{1-r}{r}\,dr\\
&=(1-z/n)^2\bigl\{-\log(1-z/n)-z/n\bigr\} + (z/n)^2\bigl\{-\log(z/n)-(1-z/n)\bigr\}.
\end{align*}
Since $-\log(1-x)\geq x+x^2/2$ for $0\leq x<1$, we have 
\[
I_1\geq (z/n)^2(1-z/n)^2.
\]
When $n \geq 6$ and $2\leq z\leq n/2$, we find $\frac{(z/n)(1-z/n)}{n}\leq 3I_1/4$. Hence,
\[
\|\gamma\|_2\geq \frac{1}{2}n(z/n)(1-z/n)\geq \frac{1}{4}z.
\]
On the other hand, under the assumption that $z\leq n/2$, we have
\[
-\log(1-z/n)-z/n\leq (z/n)^2.
\]
Hence
\[
\|\gamma\|_2^2 \leq n^2\bigl\{(1-z/n)^2(z/n)^2 + (z/n)^2 \log (n/2)\bigr\} \leq z^2\log (en/2),
\]
as required.

For the $\ell_1$ norm, we similarly write $\|\gamma\|_1$ as a Riemann sum:
\begin{align*}
\|\gamma\|_1 &= \frac{1}{\sqrt{n}}\biggl\{\sum_{t=1}^z \sqrt\frac{t}{n-t}(n-z) + \sum_{t=z+1}^{n-1}\sqrt\frac{(n-t)}{t}z\biggr\} \\
&= n^{3/2}\biggl\{\sum_{t=1}^z \sqrt\frac{t/n}{1-t/n}(1-z/n)\cdot \frac{1}{n} + \sum_{t=z+1}^{n-1}\frac{1-t/n}{t/n}(z/n)\cdot\frac{1}{n}\biggr\}.
\end{align*}
So
\[
n^{3/2}\biggl\{I_2 - \frac{\sqrt{z/n(1-z/n)}}{n}\biggr\}\leq \|\gamma\|_1\leq n^{3/2}\biggl\{I_2 + \frac{\sqrt{z/n(1-z/n)}}{n}\biggr\},
\]
where 
\[
I_2 := (1-z/n)\int_0^{z/n}\sqrt{\frac{r}{1-r}}\,dr + (z/n)\int_{z/n}^1\sqrt\frac{1-r}{r}\,dr = (1-z/n)g(z/n) + (z/n)g(1-z/n),
\]
where function $g(a):=\int_0^a \sqrt{r/(1-r)}\,dr = \arcsin(\sqrt{a}) -\sqrt{a(1-a)}$.  We can check that $g(a)/a^{3/2}$ has positive first derivative throughout $(0,1)$, and $g(a)/a^{3/2} \searrow 2/3$ as $a\searrow 0$. This implies that $2a^{3/2}/3 \leq  g(a)\leq \pi a^{3/2}/2$. Consequently,
\[
\frac{2z}{3n}\biggl(1-\frac{z}{n}\biggr)\biggl(\sqrt\frac{z}{n}+\sqrt{1-\frac{z}{n}}\biggr)\leq I_2 \leq \frac{\pi}{2}\frac{z}{n}\biggl(1-\frac{z}{n}\biggr)\biggl(\sqrt\frac{z}{n}+\sqrt{1-\frac{z}{n}}\biggr)
\]
Also, for $n\geq 6$ and $2\leq z\leq n/2$, 
\[
\frac{\sqrt{z/n(1-z/n)}}{n}\leq \frac{\sqrt{3}}{4+2\sqrt{2}}\frac{z}{n}\biggl(1-\frac{z}{n}\biggr)\biggl(\sqrt\frac{z}{n}+\sqrt{1-\frac{z}{n}}\biggr).
\]
Therefore,
\[
\|\gamma\|_1\leq (\pi/2+\sqrt{3}/(4+2\sqrt{2}))\sqrt{n}z\sup_{0\leq y\leq 1/2}(1-y)(\sqrt{y}+\sqrt{1-y})\leq 2.1\sqrt{n}z,
\]
and
\[
\|\gamma\|_1\geq (1-\sqrt{3}/(4+2\sqrt{2}))\sqrt{n}z\inf_{0\leq y\leq 1/2}(1-y)(\sqrt{y}+\sqrt{1-y})\geq \frac{1}{2}\sqrt{n}z,
\]
as required.
\end{proof}

\subsection{Auxiliary results for the proof of Theorem~\ref{Thm:SingleCP} in the main text}

The first three lemmas below are used to control the probabilities of rare events in the independent noise vector case.
\begin{lemma}
\label{Lemma:BrownianBridge}
Let $W = (W_1,\ldots,W_n)$ have independent components, each with a $N(0,\sigma^2)$ distribution, and let $E:=\mathcal{T}(W)$. Then for $u>0$, we have 
\[
\mathbb{P}\bigl(\|E\|_\infty \geq u\sigma\bigr) \leq \sqrt\frac{2}{\pi}\, \lceil \log n\rceil (u+2/u)e^{-u^2/2}.
\]
\end{lemma}
\begin{proof}
Let $B$ be a standard Brownian bridge on $[0,1]$. Then 
\[
(E_1,\ldots,E_{n-1}) \stackrel{d}{=} \biggl(\frac{\sigma B(t)}{\sqrt{t(1-t)}}\biggr)_{t = \frac{1}{n},\ldots, \frac{n-1}{n}}.
\]
Let $t = t(s) := e^{2s}/(e^{2s}+1)$ and define the process $X$ by $X(s) := \{t(s)(1-t(s))\}^{-1/2}B(t(s))$. Recall that the Ornstein--Uhlenbeck process is the centred continuous Gaussian process $\{U(s) : s \in \mathbb{R}\}$ having covariance function $\mathrm{Cov}(U(s_1),U(s_2)) = e^{-|s_1-s_2|}$. We compute that
\begin{align*}
\mathrm{Cov}\bigl(X(s_1),X(s_2)\bigr) & = \mathrm{Cov}\biggl(\frac{B\bigl(e^{2s_1}/(e^{2s_1}+1)\bigr)}{\sqrt{e^{2s_1}/(e^{2s_1}+1)^2}}, \frac{B\bigl(e^{2s_2}/(e^{2s_2}+1)\bigr)}{\sqrt{e^{2s_2}/(e^{2s_2}+1)^2}}\biggr)\\
& = \biggl(\frac{e^{s_1}}{e^{2s_1}+1}\frac{e^{s_2}}{e^{2s_2}+1}\biggr)^{-1}\frac{e^{2\min(s_1,s_2)}}{e^{2\min(s_1,s_2)}+1}\frac{1}{e^{2\max(s_1,s_2)}+1} = e^{-|s_1-s_2|}.
\end{align*}
Thus, $X$ is the Ornstein--Uhlenbeck process and we have 
\begin{align*}
\mathbb{P}(\|E\|_\infty \geq u\sigma) = \mathbb{P}\biggl\{\sup_{t\in[1/n, 1-1/n]} \biggl|\frac{B(t)}{\sqrt{t(1-t)}}\biggr| \geq u\biggr\} &= \mathbb{P}\biggl\{\sup_{s\in[0, \log (n-1)]}  |X(s)| \geq u\biggr\}\\
&\leq \lceil \log n\rceil \mathbb{P}\biggl\{\sup_{s\in[0,1]}|X(s)|\geq u\biggr\},
\end{align*}
where the inequality follows from the stationarity of the Ornstein--Uhlenbeck process and a union bound. Let $Y = \{Y(t): t\in\mathbb{R}\}$ be a centred continuous Gaussian process with covariance function $\mathrm{Cov}(Y(s),Y(t)) =\max(1-|s-t| ,0)$. Since $\mathbb{E}X(t)^2 = \mathbb{E}Y(t)^2 = 1$ for all $t$ and $\mathrm{Cov}(X(s),X(t))\geq \mathrm{Cov}(Y(s),Y(t))$, by Slepian's inequality \citep{Slepian1962}, $\sup_{s\in[0,1]}|Y(s)|$ stochastically dominates $\sup_{s\in[0,1]}|X(s)|$. Hence it suffices to establish the required bound with $Y$ in place of $X$. The process $Y$, known as the Slepian process, has excursion probabilities given by closed-form expressions \citep{Slepian1961, Shepp1971}: for $x < u$,
\[
\mathbb{P}\biggl\{\sup_{s\in[0,1]}Y(s)\geq u \biggm| Y(0) = x\biggr\} = 1-\Phi(u) + \frac{\phi(u)}{\phi(x)}\Phi(x),
\]
where $\phi$ and $\Phi$ are respectively the density and distribution functions of the standard normal distribution. Hence for $u > 0$ we can write
\begin{align*}
\mathbb{P}\biggl\{\sup_{s\in[0,1]} |Y(s)| \geq u\biggr\} & = \int_{-\infty}^\infty \mathbb{P}\biggl\{\sup_{s\in[0,1]}|Y(s)|\geq u\biggm| Y(0) = x\biggr\}\phi(x)\,dx\\
&\leq \mathbb{P}(|Y(0)|\geq u) + 2\int_{-u}^u\mathbb{P}\biggl\{\sup_{s\in[0,1]}Y(s)\geq u\biggm| Y(0)=x\biggr\}\phi(x)\,dx\\
&=2\Phi(-u) + 2\int_{-u}^u \bigl\{\phi(x)\Phi(-u) + \phi(u)\Phi(x)\bigr\}\,dx\\
&= 2u\phi(u) + 4\Phi(-u)\{1-\Phi(-u)\} \\
&\leq 2(u+2u^{-1})\phi(u),
\end{align*}
as desired.
\end{proof}
\begin{lemma}
\label{Lemma:LIL}
  Let $W_1,\ldots,W_n \stackrel{\mathrm{iid}}{\sim} N(0,\sigma^2)$ and for $1\leq t\leq n$, define $Z_t:=t^{-1/2}\sum_{r=1}^t W_r$. Then for $n\geq 5$ and $u \geq 0$,
  \[
   \mathbb{P}\Bigl(\max_{1\leq t\leq n} Z_t \geq u\sigma\Bigr) \leq 2 e^{-u^2/4}\log n.
  \]
\end{lemma}
\textbf{Remark:} This lemma can be viewed as a finite sample version of the law of iterated logarithm.
\begin{proof}
Without loss of generality, we may assume $\sigma = 1$. Suppose we have an infinite sequence of independent standard normal random variables $(W_t)_t$ and define $S_t := \sum_{r=1}^t W_r$. Then $(S_t)_t$ is a martingale and $(e^{S_t})_t$ is a non-negative submartingale. By Doob's martingale inequality, we have that
  \begin{align*}
    \mathbb{P}\Bigl(\max_{1\leq t\leq n} Z_t \geq u\Bigr) & \leq \sum_{j=1}^{\lceil \log_2(n+1)\rceil} \mathbb{P}\Bigl(\max_{2^{j-1}\leq t < 2^j} Z_t\geq u\Bigr)\leq \!\!\sum_{j=1}^{\lceil \log_2(n+1)\rceil}\!\! \inf_{\lambda>0}\mathbb{P}\Bigl(\max_{2^{j-1}\leq t<2^j} e^{\lambda S_t} \geq e^{2^{(j-1)/2} \lambda u}\Bigr)\\
    & \leq \sum_{j=1}^{\lceil \log_2(n+1)\rceil} \inf_{\lambda>0}\mathbb{E}(e^{\lambda S_{2^j}})e^{-2^{(j-1)/2}\lambda u} = \sum_{j=1}^{\lceil \log_2(n+1)\rceil} e^{-u^2/4} \leq 2e^{-u^2/4}\log n,
  \end{align*}
  as desired, where the final bound follows from the fact that for $n\geq 5$, we have $\lceil \log_2(n+1)\rceil \leq 2\log n$.
\end{proof}
\begin{lemma}
\label{Lemma:NoiseControl}
Let $W = (W_1,\ldots,W_n)$ be a row vector and let $E := \mathcal{T}(W)$. Suppose $n\geq 5$ and $z\in\{1,\ldots,n-1\}$ satisfies $\min(z,n-z)\geq n\tau$. If
\[
 \biggl|\sum_{r=1}^s W_r - \sum_{r=1}^t W_r\biggr| \leq \lambda \sqrt{|s-t|}, \qquad \forall \; 0	\leq t\leq n, s\in\{0,z,n\}
\]
then for any $t$ satisfying $|z-t|\leq n\tau/2$, we have
  \[
  |E_z - E_t| \leq 2\sqrt{2}\lambda\sqrt\frac{|z-t|}{n\tau} + 8\lambda\frac{|z-t|}{n\tau}.
  \]
\end{lemma}
\begin{proof}
We first assume that $t < z$. By definition of the CUSUM transformation $\mathcal{T}$, we obtain that
  \begin{align}
  E_z - E_t &= \sqrt\frac{n}{z(n-z)}\biggl(\frac{z}{n}\sum_{r=1}^n W_r - \sum_{r=1}^z  W_r\biggr) -
  \sqrt\frac{n}{t(n-t)}\biggl(\frac{t}{n}\sum_{r=1}^n W_r - \sum_{r=1}^t  W_r\biggr) \nonumber\\
  &= \sqrt\frac{n}{z(n-z)}\biggl(\frac{z-t}{n}\sum_{r=1}^n W_r - \sum_{r=t}^z  W_r\biggr)\nonumber\\
  &\qquad\qquad + \biggl(\sqrt\frac{n}{z(n-z)} - \sqrt\frac{n}{t(n-t)}\biggr)\biggl(\frac{t}{n}\sum_{r=1}^n W_r - \sum_{r=1}^t  W_r\biggr).\label{Eq:L3tmp0}
  \end{align}
  Under the assumption of the lemma, we have that, 
  \begin{align}
    \biggl|\frac{z-t}{n}\sum_{r=1}^n W_r - \sum_{r=t}^z  W_r \biggr| &\leq \frac{z-t}{n}\biggl|\sum_{r=1}^n W_r\biggr| + \biggl|\sum_{r=t+1}^zW_r\biggr| \nonumber\\
    & \leq \lambda(z-t)n^{-1/2} + \lambda (z-t)^{1/2} \leq 2\lambda (z-t)^{1/2}. \label{Eq:L3tmp1}
  \end{align}
Moreover,
\begin{align}
    \biggl|\frac{t}{n}\sum_{r=1}^n W_r - \sum_{r=1}^t  W_r\biggr| &= \min\biggl\{\biggl|\frac{t}{n}\sum_{r=1}^n W_r - \sum_{r=1}^t  W_r\biggr| , \biggl|\frac{n-t}{n}\sum_{r=1}^n W_r - \sum_{r=t+1}^n  W_r\biggr|\biggr\}\nonumber\\
    &\leq \min \biggl\{\lambda \bigl(t n^{-1/2} + t^{1/2}\bigr), \lambda \bigl[(n-t)n^{-1/2} + (n-t)^{1/2}\bigr]\biggr\}\nonumber\\
    &\leq 2\lambda\min\{t^{1/2},(n-t)^{1/2}\}\leq 2\lambda\min\bigl\{z^{1/2},(n-z+n\tau/2)^{1/2}\bigr\}.\label{Eq:L3tmp2}
  \end{align}
Now, by the mean value theorem there exists $\xi\in[t,z]$ such that
\begin{equation}
    \biggl|\sqrt\frac{n}{z(n-z)} - \sqrt\frac{n}{t(n-t)}\biggr|  \leq (z-t) \biggl|\frac{\xi}{n}-\frac{1}{2}\biggr|\biggl(\frac{n}{\xi(n-\xi)}\biggr)^{3/2}\leq \frac{\sqrt{2}(z-t)}{\min\bigl\{(z-n\tau/2)^{3/2},(n-z)^{3/2}\bigr\}}.\label{Eq:L3tmp3}
  \end{equation}
Combining~\eqref{Eq:L3tmp0},~\eqref{Eq:L3tmp1},~\eqref{Eq:L3tmp2} and~\eqref{Eq:L3tmp3}, we obtain
  \[
  |E_z - E_t|\leq 2\lambda\sqrt\frac{(z-t)n}{z(n-z)} + 4\lambda\frac{z-t}{n\tau} \leq 2\sqrt{2}\lambda\sqrt\frac{z-t}{n\tau} + 8\lambda\frac{z-t}{n\tau},
  \]
  as desired. The case $t>z$ can be handled similarly.
\end{proof}
The following lemma is used to control the rate of decay of the univariate CUSUM statistic from its peak in the single changepoint setting.
\begin{lemma}
\label{Lemma:Peak}
For $n \in \mathbb{N}$ and $z \in \{1,\ldots,n-1\}$, let $\gamma\in\mathbb{R}^{n-1}$ be defined as in~\eqref{Eq:gamma} of the main text, and let $\tau := n^{-1}\min\{z, n-z\}$. Then, for $t\in [z-n\tau/2, z+n\tau/2]$, we have that
\[
 \gamma_z-\gamma_t \geq \frac{2}{3\sqrt{6}}\frac{|z-t|}{\sqrt{n\tau}}.
\]
\end{lemma}
\begin{proof}
We note first that $\gamma_t$ is maximised at $t=z$. We may assume without loss of generality that $t\leq z$ (the case $t>z$ is symmetric). Hence $\gamma_t = \sqrt\frac{t}{n(n-t)}(n-z)$.  By the mean value theorem, we have that for some $\xi\in [t,z]$,
\begin{equation}
\label{Eq:Peaktmp1}
 \gamma_z -\gamma_t = \frac{1}{2}(z-t)\frac{n^{1/2}(n-z)}{\xi^{1/2}(n-\xi)^{3/2}}. 
\end{equation}
We consider two cases. If $z\leq n/2$, then
\begin{equation}
\label{Eq:Peaktmp2}
 \frac{n^{1/2}(n-z)}{\xi^{1/2}(n-\xi)^{3/2}} \geq \frac{n^{1/2}(n-z)}{(n-z/2)^{3/2}}z^{-1/2} \geq \frac{4}{3\sqrt{3}}z^{-1/2}.
\end{equation}
If $z > n/2$, then
\begin{equation}
 \label{Eq:Peaktmp3}
 \frac{n^{1/2}(n-z)}{\xi^{1/2}(n-\xi)^{3/2}} = \frac{n^{1/2}(n-z)^{3/2}}{\xi^{1/2}(n-\xi)^{3/2}}(n-z)^{-1/2} \geq \frac{4}{3\sqrt{6}}(n-z)^{-1/2}.
\end{equation}
The desired result follows from~\eqref{Eq:Peaktmp1},~\eqref{Eq:Peaktmp2} and~\eqref{Eq:Peaktmp3}.
\end{proof}

\subsection{Auxiliary results for the proof of Theorem~\ref{Thm:MultipleCP} in the main text}

In addition to auxiliary results given in the previous subsection, the proof of Theorem~\ref{Thm:MultipleCP} in the main text also requires the following two lemmas, which study the mean structure of the CUSUM transformation in the multiple changepoint setting.
\begin{lemma}
\label{Lemma:CusumShape}
Suppose that $0 = z_0 < z_1 <\cdots < z_\nu < z_{\nu+1} = n$ are integers and that $\mu \in \mathbb{R}^n$ satisfies $\mu_t = \mu_{t'}$ for all $z_i < t \leq t' \leq z_{i+1}$, $0\leq i\leq \nu$. Define $A := \mathcal{T}(\mu) \in \mathbb{R}^{n-1}$, where we treat $\mu$ as a row vector. If the series $(A_t: z_i+1\leq t\leq z_{i+1})$ is not constantly zero, then one of the following is true:
\begin{enumerate}[label={(\alph*)}, noitemsep]
\item $i = 0$ and $(A_t: z_i+1\leq t\leq z_{i+1})$ does not change sign and has strictly increasing absolute values,
\item $i = \nu$ and $(A_t: z_i+1\leq t\leq z_{i+1})$ does not change sign and has strictly decreasing absolute values,
\item $1\leq i\leq \nu-1$ and $(A_t: z_i+1\leq t\leq z_{i+1})$ is strictly monotonic,
\item $1\leq i\leq \nu-1$ and $(A_t: z_i+1\leq t\leq z_{i+1})$ does not change sign and its absolute values are strictly decreasing then strictly increasing.
\end{enumerate}
\end{lemma}
\begin{proof}
This follows from the proof of \citet[Lemma~2.2]{Venkatraman1992}.
\end{proof}
\begin{lemma}
\label{Lemma:TwoCP}
Let $1\leq z < z' \leq n-1$ be integers and $\mu_0,\mu_1 \in \mathbb{R}$.  Define $g:[z,z'] \rightarrow \mathbb{R}$ by 
\[
g(y) := \sqrt\frac{n}{y(n-y)}\{z\mu_0 + (y-z)\mu_1\}
\]
Suppose that $\min\{z,z'-z\} \geq n\tau$ and
\begin{equation}
\label{Eq:L11cond1}
G := \max_{y \in [z,z']} |g(y)| =  g(z). 
\end{equation}
Then 
\[
\sup_{y\in[z,z+0.2n\tau]} g'(y) \leq -0.5Gn^{-1}\tau.
\]
\end{lemma}
\begin{proof}
Define $r := z/n$, $r' := z'/n$, $B := r(\mu_0 - \mu_1)$ and $f(x) := n^{-1/2}g(nx)$ for $x\in[r,r']$. Then
\[
f(x) = \frac{B + \mu_1 x}{\sqrt{x(1-x)}} \quad \text{and} \quad f'(x) = \frac{(\mu_1+2B)x - B}{2\{x(1-x)\}^{3/2}}.
\]
Condition~\eqref{Eq:L11cond1} is equivalent to
\begin{equation}
\label{Eq:L11tmp1}
Gn^{-1/2} = \max_{x\in[r,r']}|f(x)| = f(r) = \frac{r\mu_0}{\sqrt{r(1-r)}}.
\end{equation}
The desired result of the lemma is equivalent to 
\[
\sup_{x\in[r,r+0.2\tau]} f'(x) \leq -0.5Gn^{-1/2}\tau.
\]
We may assume without loss of generality that it is not the case that $\mu_0 = \mu_1 = 0$, because otherwise $f$ is the zero function and $G=0$, so the result holds.  In that case, $G > 0$, so $\mu_0 > 0$, and we prove the above inequality by considering the following three cases.

\emph{Case 1}: $B \leq 0$.  Then $\mu_1 \geq \mu_0$ and in fact $\mu_1 + 2B < 0$, because otherwise $f'$ is non-negative on $[r,r']$, and if $f'(r) = 0$ (which is the only remaining possibility from~\eqref{Eq:L11tmp1}) then $B = 0$ and $\mu_1 = 0$, so $\mu_0 = 0$, a contradiction.  Moreover, since $\mathrm{sgn}(f'(x)) = \mathrm{sgn}\bigl((\mu_1+2B)x - B\bigr)$, we deduce that $\frac{B}{\mu_1+2B}\leq r\leq 1$. In particular, $\mu_1 \leq -B = r(\mu_1 - \mu_0) \leq \mu_1 - \mu_0$ and hence $\mu_0 \leq 0$, again a contradiction.

\emph{Case 2}: $B > 0$ and $\mu_1 + 2B \leq 0$.  By~\eqref{Eq:L11tmp1} and the fact that $\mu_1 < 0$, so that $B > r\mu_0$, we have for $x \in [r,r+\tau]$ that
\begin{align*}
f'(x) &\leq \frac{-B}{2\{x(1-x)\}^{3/2}} \\
&\leq \frac{-B}{2\{r(1-r)\}^{1/2}}\inf_{x \in [r,r+\tau]}\frac{\{r(1-r)\}^{1/2}}{\{x(1-x)\}^{3/2}} \leq -2Gn^{-1/2}\inf_{x \in [r,r+\tau]}\frac{r^{1/2}}{x^{1/2}} \leq -\sqrt{2}Gn^{-1/2}.
\end{align*}
Here, we used the fact that $\min\{r,r'-r\} \geq \tau$ in the final bound.

\emph{Case 3}: $B > 0$ and $\mu_1 + 2B > 0$, so that $\mu_0 > \mu_1$. In this case, considering $\mathrm{sgn}(f'(x))$ again yields $r \leq \frac{B}{\mu_1 + 2B}$. We claim that 
\begin{equation}
\label{Eq:L11tmp2}
\frac{B}{\mu_1+2B} \geq r+0.4\tau.
\end{equation}
By the fundamental theorem of calculus, 
\begin{align}
f(r) - f\Bigl(\frac{B}{\mu_1+2B}\Bigr) &= \int_r^{\frac{B}{\mu_1+2B}} \frac{B-(\mu_1+2B)x}{2\{x(1-x)\}^{3/2}}\,dx\nonumber\\
&= (\mu_1+2B)\biggl(\frac{B}{\mu_1+2B}-r\biggr)^2\int_{0}^1 \frac{u}{2\{x(u)(1-x(u))\}^{3/2}}\,du,\label{Eq:L11tmp3}
\end{align}
where we have used the substitution $x = x(u) := \frac{B}{\mu_1+2B} - (\frac{B}{\mu_1+2B} - r)u$ in the second step. Similarly, 
\begin{align}
f(r+\tau) - f\Bigl(\frac{B}{\mu_1+2B}\Bigr) &= \int_{\frac{B}{\mu_1+2B}}^{r+\tau} \frac{B-(\mu_1+2B)\tilde{x}}{2\{\tilde{x}(1-\tilde{x})\}^{3/2}}\,d\tilde{x}\nonumber\\
&= (\mu_1+2B)\biggl(r+\tau - \frac{B}{\mu_1+2B}\biggr)^2\int_{0}^1 \frac{u}{2\{\tilde{x}(u)(1-\tilde{x}(u))\}^{3/2}}\,du,\label{Eq:L11tmp4}
\end{align}
using the substitution $\tilde{x} = \tilde{x}(u):=\frac{B}{\mu_1+2B}+ (r+\tau-\frac{B}{\mu_1+2B})u$.  For every $u\in[0,1]$, we have $x(u)\leq \tilde{x}(u)\leq (1+u)x(u)$.  It follows that 
\begin{align}
\label{Eq:L11tmp4.5}
\frac{\int_{0}^1 u\{\tilde{x}(u)(1-\tilde{x}(u))\}^{-3/2}\,du}{\int_{0}^1 u\{x(u)(1-x(u))\}^{-3/2} \,du} &\geq \frac{\int_{0}^1 ux(u)^{-3/2}(1+u)^{-3/2}\,du}{\int_{0}^1 ux(u)^{-3/2}\,du} = \frac{1}{2^{1/2}}\biggl\{\frac{(\frac{B}{\mu_1+2B})^{1/2}+r^{1/2}}{(\frac{2B}{\mu_1+2B})^{1/2}+r^{1/2}}\biggr\}^2 \nonumber \\
&\geq \frac{1}{2^{1/2}}\biggl\{\frac{(r+\tau)^{1/2} + r^{1/2}}{2^{1/2}(r+\tau) + r^{1/2}}\biggr\}^2 \geq 0.45.
\end{align}
Therefore, using~\eqref{Eq:L11tmp3}, \eqref{Eq:L11tmp4} and~\eqref{Eq:L11tmp4.5}, together with the fact that $f(r)\geq f(r+\tau)$, we deduce that
\[
\frac{B}{\mu_1+2B} -r \geq \frac{\tau}{1+0.45^{-1/2}} > 0.4\tau.
\]
Hence~\eqref{Eq:L11tmp2} holds. For $x \in [r, r+ 0.2\tau]$, we have 
\begin{equation}
\label{Eq:L11tmp5}
f'(x)\leq \frac{-(\mu_1+2B)\bigl(\frac{B}{2(\mu_1+2B)} - \frac{r}{2}\bigr)}{2\{x(1-x)\}^{3/2}} \leq \frac{-0.4\tau(\mu_1+2B)}{\sqrt{1.2r(1-r)}}.
\end{equation}
If $\mu_1 \geq 0$, then $r\leq \frac{B}{\mu_1+2B} \leq 1/2$ and
\begin{equation}
\label{Eq:Case1}
\mu_1+2B = 2r\mu_0 + (1-2r)\mu_1 \geq 2r\mu_0.
\end{equation}
If $\mu_1 < 0$ and $r \geq 1/2$, then 
\begin{equation}
\label{Eq:Case2}
\mu_1 + 2B = 2r\mu_0 + (2r-1)(-\mu_1) \geq 2r\mu_0.
\end{equation}
Finally, if $\mu_1<0$ and $r < 1/2$, then, writing $a := 1-2r$ and $b := \frac{2B}{\mu_1+2B}-1$, we have from~\eqref{Eq:L11tmp2} that $a+b\geq 0.8\tau$ and 
\begin{align}
(\mu_1+2B)\biggl(\frac{B}{\mu_1+2B} - r\biggr) &=  r(1-2r)\mu_0 - 2r(1-r)\mu_1 = ar\mu_0 + \frac{(1-a^2)B}{1+b^{-1}} \nonumber\\
&\geq \biggl(a+\frac{1-a^2}{1+(0.8\tau-a)^{-1}}\biggr)r\mu_0 \geq 0.57\tau r\mu_0.
\label{Eq:Case3}
\end{align}
It follows from~\eqref{Eq:L11tmp5},~\eqref{Eq:Case1},~\eqref{Eq:Case2},~\eqref{Eq:Case3} and~\eqref{Eq:L11tmp1} that for $x \in [r,r+0.2\tau]$,
\[
f'(x) \leq \frac{-0.57\tau r\mu_0}{\sqrt{1.2r(1-r)}} \leq -0.5Gn^{-1/2}\tau,
\]
as desired.
\end{proof}

\subsection{Auxiliary results for theoretical guarantees under dependence}

Lemma~\ref{Lemma:TimeDependent} below, which is used in the proof of Theorem~\ref{Thm:TemporalDependence} in the main text, provides weaker conclusions than those of Lemmas~\ref{Lemma:BrownianBridge} and~\ref{Lemma:LIL}, but under more general conditions, which in particular allow for time-dependent noise.
\begin{lemma}
\label{Lemma:TimeDependent}
 Suppose that $W = (W_1,\ldots,W_n)$ is a univariate, centred, stationary Gaussian process with covariance function $K(u):=\mathrm{cov}(W_t,W_{t+u})$ satisfying $\sum_{u=0}^{n-1}K(u) \leq B$ for some universal constant $B > 0$.  Let $E:= \mathcal{T}(W)$ and $Z_t:= t^{-1/2}\sum_{r=1}^t W_r$.  Then, for $u \geq 0$,
 \begin{align*}
\mathbb{P}\Bigl(\|E\|_\infty \geq u\Bigr) &\leq (n-1)e^{-u^2/(4B)}, \\
\mathbb{P}\Bigl(\max_{1\leq t\leq n} Z_t \geq u \Bigr) &\leq \frac{1}{2}n e^{-u^2/(4B)}.
\end{align*}
\end{lemma}
\begin{proof}
Fix $t \in \{1,\ldots,n-1\}$ and define the `contrast' vector $\kappa = (\kappa_1,\ldots,\kappa_n)^\top \in\mathbb{R}^{n}$ by
 \[
  \kappa_r := \begin{cases} -\sqrt\frac{n-t}{tn} & \text{for $1\leq r\leq t$}\\ \sqrt\frac{t}{(n-t)n} & \text{for $t+1\leq r\leq n$.}\end{cases}
 \]
 Then $E_t = W\kappa$ and
 \begin{align*}
  \mathrm{var}(E_t) &= \sum_{r_1=1}^n \sum_{r_2=1}^n \kappa_{r_1}\kappa_{r_2} K(|r_2-r_1|) \leq 2\sum_{u=0}^{n-1}K(u)\sum_{r=1}^{n-u}\kappa_r\kappa_{r+u}\\
  &\leq 2\sum_{u=0}^{n-1}K(u)\biggl(\frac{(n-t)(t-u)}{tn} + \frac{t(n-t-u)}{(n-t)n}\biggr) \leq 2B.
 \end{align*}
 Similarly, 
 \[
  \mathrm{var}(Z_t) = \frac{1}{t} \biggl(tK(0) + \sum_{u=1}^{t-1} 2(t-u)K(u)\biggr) \leq  2B.
 \]
 Since both $E$ and $Z$ have Gaussian entries, the desired results follow by combining a union bound with the fact that $\mathbb{P}(Y \geq t) \leq e^{-t^2/2}/2$ when $Y \sim N(0,1)$ and $t \geq 0$.
\end{proof}
Our final results are used in the proof of Theorem~\ref{Thm:SpatDep}, which provides theoretical guarantees on the performance of our modified \texttt{inspect} algorithm in the presence of spatial dependence.
\begin{lemma}
\label{Lemma:AuBv}
Let $u,v\in\mathbb{S}^{p-1}$ and that $A,B\in\mathbb{R}^{p\times p}$.  Then
\[
 \sin\angle(Au, Bv) \leq 6y+2y^2,
\]
where $y := \{\|A-B\|_{\mathrm{op}} + 2^{1/2}\sigma_{\max}(B)\sin\angle(u,v)\}/\sigma_{\min}(B)$.
\end{lemma}
\begin{proof}
We initially consider the case where $\|A - B\|_{\mathrm{op}} \leq \sigma_{\min}(B)/2$.  For unit vectors $u_*,v_* \in \mathbb{R}^p$, we have $0 \leq (1-u_*^\top v_*)^2 = - \sin^2 \angle(u_*,v_*) + \|u_*-v_*\|_2^2$.  By this fact and the mean value theorem,
\begin{align*}
\sin \angle(Au, Bv) &\leq \biggl\|\frac{Au}{\|Au\|_2} - \frac{Bv}{\|Bv\|_2}\biggr\|_2 \leq \frac{\|Au\|_2\|Au-Bv\|_2}{\min(\|Au\|_2^2,\|Bv\|_2^2)} + \frac{\|Au-Bv\|_2}{\|Bv\|_2} \\
&\leq 2\|Au-Bv\|_2\biggl(\frac{1}{\|Au\|_2} + \frac{\|Au\|_2}{\|Bv\|_2^2}\biggr) \\
&\leq 2\|Au-Bv\|_2\biggl(\frac{3}{\sigma_{\min}(B)} + \frac{\|Au-Bv\|_2}{\sigma_{\min}^2(B)}\biggr).
\end{align*}
Since the left-hand side of our desired inequality is invariant under sign changes of either argument, we may assume without loss of generality that $u^\top v \geq 0$, in which case $\|u-v\|_2 \leq 2^{1/2}\sin \angle(u,v)$.  Hence
\[
\|Au-Bv\|_2 \leq \|A - B\|_{\mathrm{op}} + \sigma_{\max}(B)\|u-v\|_2 \leq \|A - B\|_{\mathrm{op}} + 2^{1/2}\sigma_{\max}(B)\sin\angle(u,v).
\]
The result in the case $\|A - B\|_{\mathrm{op}} \leq \sigma_{\min}(B)/2$ follows.  But if $\|A - B\|_{\mathrm{op}} > \sigma_{\min}(B)/2$, then $y \geq 1/2$, so the bound is trivial.
\end{proof}
\begin{lemma}
 \label{Lemma:LocalCS}
Assume $p \geq 2$. Suppose $W_1,\ldots,W_m \stackrel{\mathrm{iid}}{\sim} N_p(0, \Sigma)$ for $\Sigma = (\Sigma_{i,j}) = (\rho^{|i-j|})$, where $\rho\in (-1,1)$.  Then 
\[
\frac{1-|\rho|}{1+|\rho|} \leq \sigma_{\min}(\Sigma) \leq \sigma_{\max}(\Sigma) \leq \frac{1+|\rho|}{1-|\rho|}.
 \]
There exists a maximum likelihood estimator $\hat{\rho}$ of $\rho$ in $[-1,1]$ based on $W_1,\ldots,W_m$.  Moreover, writing $\hat\Sigma = (\hat\rho^{|i-j|})$, for $t > 0$ and $m(p-1) \geq 4(1-|\rho|)^2t^2$, 
\[
\mathbb{P}\biggl(m^{1/2}(p-1)^{1/2}\|\hat\Sigma^{-1} - \Sigma^{-1}\|_{\mathrm{op}} > t\biggr) \leq \frac{144}{(1-|\rho|)^4t^2}\biggl(9+\rho^2 + \frac{20\rho^2}{1-\rho^2}\biggr).
\]
\end{lemma}
\begin{proof}
 Define $\mathrm{tridiag}(\alpha,\beta,\gamma)$ to be the $p\times p$ Toeplitz tridiagonal matrix whose entries on the main diagonal, superdiagonal and subdiagonal are equal to $\alpha,\beta,\gamma \in \mathbb{R}$ respectively.  Then $\det\Sigma = (1-\rho^2)^{p-1}$ and 
 \[
\Theta := \Sigma^{-1} = \frac{1}{1-\rho^2}\bigl\{\mathrm{tridiag}(1+\rho^2,-\rho,-\rho) - \rho^2 (e_1e_1^\top + e_pe_p^\top)\bigr\},
 \]
where $e_j \in \mathbb{R}^p$ is the $j$th standard basis vector.
 For $\alpha,\beta,\gamma\in\mathbb{C}$, by, e.g., \citet[Theorem~4 and Theorem~5]{Yueh2005} we have that the eigenvalues of $\mathrm{tridiag}(\alpha,\beta,\gamma)$ are
 \[
  \biggl\{\alpha+2\sqrt{\beta\gamma}\cos\frac{j\pi}{p+1}\,:\,1\leq j\leq p\biggr\}
 \]
 and for $\xi\in\{-1,1\}$, the eigenvalues of $\mathrm{tridiag}(\alpha,\beta,\gamma) + \xi\sqrt{\beta\gamma}(e_1e_1^\top + e_pe_p^\top)$ are
 \begin{equation}
\label{Eq:Yueh}
  \biggl\{\alpha + 2\xi\sqrt{\beta\gamma}\cos\frac{j\pi}{p}\,:\,1\leq j\leq p\biggr\}.
 \end{equation}
 Since $(1-\rho^2)^{-1} \bigl\{\mathrm{tridiag}(1+\rho^2,-\rho,-\rho) - |\rho| (e_1e_1^\top + e_pe_p^\top)\bigr\} \leq \Theta \leq (1-\rho^2)^{-1} \mathrm{tridiag}(1+\rho^2,-\rho,-\rho)$ in the usual matrix semidefinite ordering, we conclude that
 \[
  \frac{(1-|\rho|)^2}{1-\rho^2} \leq \sigma_{\min}(\Theta) \leq \sigma_{\max}(\Theta) \leq \frac{(1+|\rho|)^2}{1-\rho^2},
 \]
from which the first claim of the lemma follows. 

Now let $S = (S_{i,j}) := m^{-1}\sum_{t=1}^m W_t W_t^\top$ and write 
\begin{align*}
&\ell(\rho; W_1,\ldots,W_m) = \tilde{\ell}(\Sigma; S) := -\frac{m}{2}\log\det\Sigma - \frac{m}{2}\mathrm{tr}(\Sigma^{-1}S) \\
&= -\frac{m(p-1)}{2}\log(1-\rho^2) - \frac{m}{2(1-\rho^2)}\biggl\{(1+\rho^2)\mathrm{tr}(S) - \rho^2(S_{1,1}+S_{p,p}) - 2\rho\sum_{j=1}^{p-1} S_{j,j+1}\biggr\} 
\end{align*}
for the log-likelihood.  Now any $\hat{\rho} \in (-1,1)$ satisfies
\[
\biggl|\frac{\rho^2}{1-\rho^2} - \frac{\hat{\rho}^2}{1-\hat{\rho}^2}\biggr| \leq \biggl|\frac{\rho}{1-\rho^2} - \frac{\hat{\rho}}{1-\hat{\rho}^2}\biggr|.
\]
Thus, if $\hat{\rho}$ is a maximum likelihood estimator, then writing $\hat{\Sigma} = (\hat{\Sigma}_{i,j}) = (\hat{\rho}^{|i-j|})$, it follows from this and~\eqref{Eq:Yueh} that 
 \begin{align*}
  \|\hat\Sigma^{-1} - \Sigma^{-1}\|_{\mathrm{op}} &= \biggl\|\mathrm{tridiag}\biggl(\frac{1+\hat\rho^2}{1-\hat\rho^2} - \frac{1+\rho^2}{1-\rho^2}, \frac{\rho}{1-\rho^2}- \frac{\hat\rho}{1-\hat\rho^2}, \frac{\rho}{1-\rho^2} - \frac{\hat\rho}{1-\hat\rho^2}\biggr) \\
  &\qquad - \biggl(\frac{\hat\rho^2}{1-\hat\rho^2} - \frac{\rho^2}{1-\rho^2}\biggr)(e_1e_1^\top + e_pe_p^\top)\biggr\|_{\mathrm{op}} \\
&\leq \biggl|\frac{1+\hat\rho^2}{1-\hat\rho^2} - \frac{1+\rho^2}{1-\rho^2}\biggr| + 2\biggl|\frac{\hat\rho}{1-\hat\rho^2}-\frac{\rho}{1-\rho^2}\biggr|\\
  &= \max\biggl\{\biggl|\frac{1+\hat{\rho}}{1-\hat\rho} - \frac{1+\rho}{1-\rho} \biggr| \, , \, \biggl|\frac{1-\hat{\rho}}{1+\hat\rho} - \frac{1-\rho}{1+\rho} \biggr|\biggr\} \leq \frac{2|\hat\rho-\rho|}{\{1-\max(|\rho|,|\hat\rho|)\}^2},
 \end{align*}
where the final step uses the mean value theorem.  Writing $\eta := 1 - |\rho| \in (0,1]$, we therefore have that for $s > 0$, 
\begin{align*}
\mathbb{P}\bigl(\|\hat\Sigma^{-1} - \Sigma^{-1}\|_{\mathrm{op}} > s\bigr) &\leq \mathbb{P}\biggl(\frac{2|\hat{\rho}-\rho|}{(\eta- |\hat{\rho} - \rho|)^2} > s\biggr) \\
&= \mathbb{P}\bigl(s|\hat{\rho} - \rho|^2 - 2(1+\eta s)|\hat{\rho} - \rho| + \eta^2 s < 0\bigr) \\
&\leq \mathbb{P}\biggl(|\hat{\rho} - \rho| > \frac{1 + \eta s - \sqrt{1+2\eta s}}{s}\biggr) \leq \mathbb{P}\biggl(|\hat{\rho}-\rho| > \frac{\eta^2 s}{2(1+\eta s)}\biggr).
\end{align*}
Now
 \[
  -\frac{1}{m(p-1)}\frac{\partial}{\partial \rho}\ell(\rho; W_1,\ldots,W_m) = \frac{\rho^3 - a\rho^2 + (b-1)\rho - a}{(1-\rho^2)^2}
 \]
where $a := (p-1)^{-1}\sum_{j=1}^{p-1}S_{j,j+1}$ and $b := (p-1)^{-1}(2\mathrm{tr}(S) - S_{1,1} - S_{p,p})$.  The form of the derivative of the log-likelihood shows that a maximum likelihood estimator $\hat{\rho}$ exists.  Define the event
 \[
 \Omega_0 := \biggl\{|a - \rho|\leq \frac{\eta^2 s}{8(1+\eta s)}, |b-2| \leq \frac{\eta^2 s}{8(1+\eta s)}\biggr\}. 
 \]
Writing $f(\rho) := \rho^3 - a\rho^2 + (b-1)\rho - a$, we have on $\Omega_0$ that for $\eta s \in \bigl(0,1/2]$,
 \[
\frac{d}{d\rho}f(\rho) = 3\rho^2 - 2a\rho + b-1 \geq -\frac{a^2}{3} + b-1 \geq 1 - \frac{\eta^2 s}{8(1+\eta s)} - \frac{1}{3}\biggl(1+\frac{\eta^2s}{8(1+\eta s)}\biggr)^2 \geq \frac{1}{2},
\]
so the log-likelihood is strictly concave.  Moreover, $f(a) = a(b-2)$, so it follows that on $\Omega_0$, the maximum likelihood estimator $\hat{\rho}$ is unique, and for $\eta s \in (0,1/2]$,
\[
|\hat{\rho} - \rho| \leq |\hat{\rho}-a| + \frac{\eta^2 s}{8(1+\eta s)} \leq 2|a(b-2)| + \frac{\eta^2 s}{8(1+\eta s)} \leq \frac{\eta^2 s}{2(1+\eta s)}.
\]
Now, $\mathbb{E}(a) = \rho$, and by Isserlis's theorem \citep{Isserlis1918},
\begin{align*}
\mathrm{var}(a) &= \frac{1}{m}\mathrm{var}\biggl(\frac{1}{p-1}\sum_{j=1}^{p-1} W_{j,1}W_{j+1,1}\biggr) \\
&= \frac{1}{m(p-1)^2}\biggl\{(p-1)(1+\rho^2) + 4\sum_{j=1}^{p-2} (p-j-1)\rho^{2j}\biggr\} \\
&\leq \frac{1}{m(p-1)}\biggl(1+\rho^2 + \frac{4\rho^2}{1-\rho^2}\biggr).
\end{align*}
Similarly, $\mathbb{E}(b) = 2$ and by Isserlis's theorem again,
\begin{align*}
\mathrm{var}(b) \leq \frac{4}{m(p-1)^2}\mathrm{var}\biggl(\sum_{j=1}^{p-1} W_{j,1}^2\biggr) &= \frac{4}{m(p-1)^2}\biggl\{2(p-1) + 4\sum_{j=1}^{p-2}(p-j-1)\rho^{2j}\biggr\} \\
&\leq \frac{8}{m(p-1)}\biggl(1+ \frac{2\rho^2}{1-\rho^2}\biggr).
\end{align*}
We conclude by Chebychev's inequality that provided $m(p-1) \geq 4(1-|\rho|)^2t^2$,
\begin{align*}
\mathbb{P}\bigl(m^{1/2}(p-1)^{1/2}&\|\hat\Sigma^{-1} - \Sigma^{-1}\|_{\mathrm{op}} > t\bigr) \leq \mathbb{P}(\Omega_0^c) \\
&\leq \mathbb{P}\biggl(|a - \rho| > \frac{\eta^2 t}{12m^{1/2}(p-1)^{1/2}} \biggr) + \mathbb{P}\biggl(|b-2| > \frac{\eta^2 t}{12m^{1/2}(p-1)^{1/2}}\biggr) \\
&\leq \frac{144}{(1-|\rho|)^4t^2}\biggl(9+\rho^2 + \frac{20\rho^2}{1-\rho^2}\biggr),
\end{align*}
as required.
\end{proof}
\begin{lemma}
\label{Lemma:GlobalCS}
Suppose $W_1,\ldots,W_m \stackrel{\mathrm{iid}}{\sim} N_p(0, \Sigma)$ for $\Sigma = I_p + \frac{\rho}{p}\mathbf{1}_p\mathbf{1}_p^\top$, where $\rho > -1$.  There exists a unique maximum likelihood estimator $\hat{\rho}$ of $\rho$ in $[-1,\infty)$ based on $W_1,\ldots,W_m$.  Moreover, if $m\geq 10$, then writing $\hat\Sigma = I_p + \frac{\hat{\rho}}{p}\mathbf{1}_p\mathbf{1}_p^\top$, for $t > 0$,
\[
\mathbb{P}\biggl(m^{1/2}\|\hat\Sigma^{-1} - \Sigma^{-1}\|_{\mathrm{op}} > t\biggr) \leq \frac{21}{(1+\rho)^2t^2}.
\]
\end{lemma}
\begin{proof}
By the Woodbury formula, $\Theta := \Sigma^{-1} = I_p - \frac{\rho}{p(1+\rho)}\mathbf{1}_p\mathbf{1}_p^\top$.  Writing $S = (S_{i,j}) := m^{-1}\sum_{t=1}^m W_t W_t^\top$, it follows that the log-likelihood is given by
\begin{align*}
\ell(\rho; W_1,\ldots,W_m) = \tilde{\ell}(\Sigma; S) &:= -\frac{m}{2}\log\det\Sigma - \frac{m}{2}\mathrm{tr}(\Theta S) \\
&= -\frac{m}{2}\log (1+\rho) - \frac{m}{2}\biggl\{\mathrm{tr}(S) - \frac{\rho}{p(1+\rho)} \sum_{i=1}^p \sum_{j=1}^p S_{i,j}\biggr\}.
\end{align*}
Hence there exists a unique maximum likelihood estimator $\hat{\rho}$, given by
\[
\hat{\rho} = \frac{1}{p}\sum_{i=1}^p \sum_{j=1}^p S_{i,j} - 1.
\]
Therefore, $1+\hat\rho\sim (1+\rho)\chi^2_m/m$ and $\frac{1+\rho}{1+\hat\rho}$ has mean $m/(m-2)$ and variance $2m^2(m-2)^{-2}(m-4)^{-1}$.  From the statement of the lemma, we may assume that $(1+\rho)^2t^2 \geq 21$, in which case for $m \geq 10$, we have
\[
\frac{2}{m-2} \leq \frac{5}{2m} \leq \frac{(1+\rho)t}{2m^{1/2}}.  
\]
Hence by Chebychev's inquality, for $t > 0$,
\begin{align*}
\mathbb{P}(m^{1/2}\|\hat{\Theta} - \Theta\|_{\mathrm{op}} > t) &= \mathbb{P}\biggl(\biggl|\frac{1}{1+\hat\rho} - \frac{1}{1+\rho}\biggr| > \frac{t}{m^{1/2}} \biggr) \leq \mathbb{P}\biggl(\biggl|\frac{1+\rho}{1+\hat\rho} - \frac{m}{m-2}\biggr| > \frac{(1+\rho)t}{2m^{1/2}} \biggr) \\
&\leq \frac{8m^3}{(1+\rho)^2t^2(m-2)^2(m-4)}\leq \frac{21}{(1+\rho)^2t^2},
\end{align*}
as required.
\end{proof}

\end{document}